\documentclass{article}
\usepackage{amsmath, amssymb}
\usepackage{amsfonts}
\usepackage{appendix}
\usepackage{graphicx}
\usepackage{color}
\usepackage{url}
\usepackage{bm}
\usepackage[all]{xy}
\usepackage{arydshln}
\usepackage{tikz}
\usepackage{float}

\newtheorem{theorem}{Theorem}
\newtheorem{lemma}{Lemma}
\newtheorem{corollary}{Corollary}
\newtheorem{property}{Property}
\newtheorem{definition}{Definition}

\newtheorem{remark}{Remark}

\newtheorem{example}{Example}
\newtheorem{construction}{Construction}

\newenvironment{proof}{{\noindent\it Proof}\quad}{\hfill $\square$\par}

\newcommand{\Z}{\ensuremath{\mathbb Z}}

\newcommand{\done}{\hfill $\Box$ }

\newcommand{\ls}[1]
    {\dimen0=\fontdimen6\the\font\lineskip=#1\dimen0
     \advance\lineskip.5\fontdimen5\the\font
     \advance\lineskip-\dimen0
     \lineskiplimit=0.9\lineskip
     \baselineskip=\lineskip
     \advance\baselineskip\dimen0
     \normallineskip\lineskip\normallineskiplimit\lineskiplimit
     \normalbaselineskip\baselineskip
     \ignorespaces}

\begin{document}

\bibliographystyle{abbrv}

\title{New Construction of Complementary Sequence (or Array) Sets and Complete Complementary Codes}
\author{Zilong Wang$^1$,
 Dongxu Ma$^{1, 2}$, Guang Gong$^2$, IEEE Fellow, Erzhong Xue$^1$\\
\small $^1$ State Key Laboratory of Integrated Service Networks, Xidian University \\[-0.8ex]
\small Xi'an, 710071, China\\
\small $^2$Department of Electrical and Computer Engineering, University of Waterloo \\
\small Waterloo, Ontario N2L 3G1, Canada  \\
\small\tt zlwang@xidian.edu.cn,  xidiandongxuma@foxmail.com, ggong@uwaterloo.ca, 2524384374@qq.com\\
}

\maketitle

\thispagestyle{plain} \setcounter{page}{1}

\begin{abstract}
A new method to  construct $q$-ary complementary sequence sets (CSSs) and complete complementary codes (CCCs) of size $N$ is proposed by using desired para-unitary (PU) matrices.  The concept of seed PU matrices is introduced and
 a systematic approach  on how to compute the explicit forms of the functions in constructed CSSs and CCCs from the seed PU matrices is  given.  A general form of these functions only depends on a basis of the functions from  $\Z_N$ to $\Z_q$ and representatives in the equivalent class of Butson-type Hadamard (BH) matrices.  Especially,  the realization   of  Golay pairs  from the our general form exactly coincides  with the standard Golay pairs. The realization of ternary complementary sequences of size $3$  is first  reported here.
 For the realization of the quaternary complementary sequences of size 4, almost all the sequences derived here are  never reported before. Generalized seed PU matrices and the recursive constructions of the desired PU matrices are also studied, and a large number of new constructions of CSSs and CCCs are given accordingly. From the perspective of this paper, all the known results of CSSs and CCCs with explicit GBF form in the literature (except non-standard Golay pairs) are  constructed from the Walsh matrices of order 2. This  suggests that the proposed method with the BH matrices of higher orders will yield a large number of new CSSs and CCCs with the exponentially increasing number of the sequences of  low peak-to-mean envelope power ratio.
\end{abstract}


\ls{1.5}
\section{Introduction}

{\em Golay sequence pair} (GSP) was first introduced by Golay in \cite{Golay61} when he studied infrared spectrometry \cite{Golay51}. The concept of GSP was extended later to {\em complementary sequence set} (CSS) for binary case \cite{Tseng72} and polyphase case
\cite{Sivaswamy78}.

Sequences in GSPs and CSSs have found many applications in physics, combinatorics and telecommunications such as channel measurement, synchronisation and spread spectrum communications. In particular, orthogonal frequency division multiplexing (OFDM) has recently seen rising popularity in international standards including coming 5G cellular systems. However, a major drawback of OFDM is the high peak-to-mean envelope power ratio (PMEPR) of uncoded OFDM signals \cite{Litsynbook}. Among all the approaches that have been proposed to deal with this power control problem, one of the most attractive methods is coding across the subcarriers and selecting codewords with lower PMEPR. These codewords can be a linear code or codewords drawn from cosets of a linear code \cite{offset,Tarokh00,PT2000}. It has been shown in \cite{Popovic91} the use of Golay sequences as codewords results in OFDM signals
with PMEPR of at most 2.

Golay \cite{Golay61} first proposed two recursive
 algorithms  to construct
binary GSPs of length $2^m$.  In 1999, Davis and Jedwab's milestone work \cite{DavisJedwab99} showed that
the  Golay  sequences
described in \cite{Golay61} can be obtained from specific second-order
cosets of the  first-order generalized Reed-Muller (GRM) code \cite{MacWilliams}.  This established an effective way of combining the coding approach  and the use of Golay sequences. However, the rate of this code suffers a dramatic decline when the
length increases. In order to increase the code rate, further cosets of the first order GRM code with  a slightly larger PMEPR were also proposed by exhaustive search  in \cite{DavisJedwab99}. Paterson \cite{Paterson00,Paterson02} studied general second-order cosets of the first-order GRM code, and showed that the codewords of such a coset lie in CSSs of size $N$, where $N$ is related to the graph associated with the quadratic form of the coset. Schmidt \cite{Schmidt06,Schmidt07} developed a sophisticated theory and proposed a construction of sequences with low PMEPR which were contained in higher-order cosets of the first-order GRM code.

Several other CSSs constructions have been given in \cite{Chen06,Stinchcombe,Wu2016}. These results yield more sequences  and provide better upper bounds on the PMEPR than the work in \cite{Paterson00,Schmidt07}. But it is difficult to use them in a practicable coding scheme for OFDM.

A large volume of works  have been done along the line of above methods in the past years. For example, non-standard Golay sequences were studied in \cite{Fiedler06,Fiedler2,Li05}, the exact value of PMEPR of binary Golay sequences was determined in \cite{Wang14}, and GSPs with low PMEPR over QAM constellation were constructed in  \cite{Chong03,Li10,Liu13,Tarokh01,Tarokh03,Wang09},  just to list a few.

The concept of {\em complete mutually orthogonal complementary set} (CMOCS) or {\em complete complementary code} (CCC) was  introduced in \cite{Suehiro88}. CCC is a source to design zero correlation zone (ZCZ) sequences, which can be used to eliminate the multiple access interference and multi-path interference in QS-CDMA system.  Several important methods to construct ZCZ sequences presented in \cite{ZCZ03,ZCZ13,ZCZ06,CCC,Tang2010} are all based on the construction of CCCs. In particular, the constructions of CCCs from  GRM codes were proposed in \cite{Chen08,Liu14,CCC}.

In the  aforementioned constructions, the sequences are  directly given by generalized Boolean functions (GBFs), so we refer  to it simply  as \emph{GBF-based constructions} in this paper. However,
in addition to those GBF-based constructions, there
exists another approach based on Hadamard matrices to construct CSSs and CCCs. For instance, the first recursive construction of CSSs  in \cite{Tseng72} was obtained by Hadamard matrices. However, how to obtain the explicit function form of the sequences derived by Hadamard-matrix-based methods constitutes an extremely challenge problem for several decades, so these methods are not as popular as the GBF approaches in the literature. On the other hand, based on the observation  on the filter bank theory and the design of GSPs \cite{BudSpas}, a new method to construct GSPs was proposed in \cite{BudIT} by {\em para-unitary} (PU) matrices. Here the word \lq\lq para-unitary\rq\rq\ refers to a unitary matrix with polynomial entries. It has been shown in \cite{BudIT} that this {\em PU-matrix-based construction}  can explain all the standard Golay sequences in \cite{DavisJedwab99} and all known GSPs over QAM constellation in \cite{Li10,Liu13}. Moreover, some new GSPs over QAM \cite{BudQAM,BudIT} are also derived from this method. Note that PU matrices are rarely used in the literature because they may not be able to map to $q$-ary sequences.

In 2016 , we \cite{SETA2016} proposed a new PU-matrix-based construction to generate $q$-ary CSSs of size $N$ and length $N^m$ where $N$ and $m$ are both positive integers. Almost at the same time, Budi\v{s}in \cite{Bud-unp} also proposed a PU-matrix-based construction on CSSs and CCCs. These two constructions are similar. However, we have introduced the {\em Butson-type Hadamard} (BH) matrices to make the PU-matrix-based construction producing $q$-ary sequences. Das \emph{et al.} \cite{TSP2018} proposed a construction by applying the BH matrices to Budi\v{s}in's construction. Actually, from the viewpoint of the PU matrices, the construction in \cite{TSP2018} is identical to our construction in \cite{SETA2016}, but they were  not aware of  our work \cite{SETA2016} until 2018. A hardware implementation,  an enumeration of the sequences, and a comparison with known constructions of CCCs in \cite{Marziani,Suehiro88} were also given in \cite{TSP2018}. Furthermore,   we have showed in \cite{BSC} that the  new sequences in CSSs of size 4  can be obtained by the PU-matrix-based construction \cite{SETA2016}.  However, the relationship between the sequences constructed by  PU-matrix-based  method and GBF-based method is unsolved  until now.

The concept of GSP was generalized to {\em Golay array pair} (GAP) in \cite{Luke}, which has found applications in coded imaging \cite{Ding16}. A powerful three-stage process was presented in \cite{Luke}
by showing that all known GSPs of length $2^m$ can be derived from the seed GAPs.

In this paper, inspired by the excellent idea introduced in \cite{Array2,Array1, Parker11}, we  generalize the concepts of CSS and CCC to a {\em complementary array set} (CAS) and a {\em complete complementary array} (CCA), and show that a large number of CSSs and CCCs of size $N$ can be constructed from a single CCA of size $N$. We also prove that a CCA of size $N$ can be constructed by a multivariate polynomial matrix of order $N$ satisfying some properties, which is referred to as a {\em desired PU matrix} in this paper. A principal objective of this paper is to construct desired PU matrices together with extracting the explicit forms of the functions (arrays) from the desired PU matrices. Consequently, CSSs and CCCs can be constructed from these functions.  The proposed constructions using designed PU matrices in this paper  not only generate new CSSs and CCCs, but also reveal that  the  known constructions of CSSs and CCCs of length $2^m$ \cite{DavisJedwab99, Paterson00,CCC,Schmidt07} are special cases of those constructions.

Since this paper consists of a large volume of the contents,  in order to easily read the paper,  in the following we provide a summary of our main contributions of this paper together with their dependency relationships among our new constructions.
\begin{enumerate}
	\item  We introduce a general framework for constructing desired PU matrices (Theorem 1 in Section 3) and a general treatment of sequences, arrays, and their respective generating functions.
	\item  By extending the univariate  PU matrices in \cite{SETA2016} to the array version, we obtain the first construction of the desired PU matrices in this paper, which is called  {\em seed PU matrices} (Theorem 2 in Section 4). More precisely, such a seed PU matrix is   constructed by iteratively multiplying $m$ times the product of BH matrices and a delay matrix  to the BH matrices.
	
	\begin{enumerate}
		\item

	 One important contribution of this paper is to provide a systematic approach to extract the functions (arrays) from the seed PU matrices. These functions depends on a basis of functions from  $\Z_N$ to $\Z_q$ and representatives in the equivalence class of BH matrices.
	
	\item From this construction,  all the functions fill up a larger number of cosets of a linear code.  Furthermore, since the newly constructed  functions from the seed PU matrices in this paper also determine all the sequences from PU-matrix-based constructed in \cite{TSP2018, Marziani,Suehiro88,SETA2016}, our enumeration based on a rigorous proof corrects  some results in \cite{TSP2018}.
\end{enumerate}

\item  Using  the construction from the seed PU matrices (Theorem 5 in Section 5), we realize four different classes with the explicit function forms of CCAs and the sequences in both CSSs and CCCs,  which are summarized as follows.

\begin{enumerate}

	\item
 For $q$-ary pair ($N=2$, $q$ is even), our Construction 1 exactly coincides  with the standard  GSPs \cite{DavisJedwab99} by Davis and Jedwab.

\item  For the case size 3 ($N=q=3$), there are $m!\cdot 2^{m-2}\cdot 3^{2m+1}$ new constructed ternary sequences in CSSs of 3.

\item For respective binary and quaternary cases of size 4 ($N=4$, $q=2$ and $q=4$), these binary and quaternary sequences can be represented by explicit GBFs in a general form respectively.
 On the other hand,  both binary and quaternary sequences fill up a large number of cosets of a linear code $S_L(q, 4)$ for $q=2$ and $q=4$ respectively, which contains  the first-order GRM code as a sub-code.  For $q=2$, the sequences are not reported in Paterson's construction \cite{Paterson00} and Schmidt's construction \cite{Schmidt07}, but can be explained by the results in \cite{Wu2016}. However, for $q=4$,  most of them are never reported before. Note that sequences in CSSs of size $4$ have been deeply studied for PMEPR control for almost 20 years.

\end{enumerate}
\item However, the construction using the seed PU matrices cannot explain all the known constructions.  We then turn our focus to  the arrays of size $2\times 2 \times \cdots \times 2$. We generalize the seed PU matrices of order $N=2^n$, where the delay matrix  in a  seed PU matrix in Theorem 2  is replaced by a generalized delay  matrix which is the Kronecker product of $m$ delay matrices of order 2 (Theorem 6 in Section 6). And the explicit GBF forms are derived in Theorem 7 (in Section 6). This result extends the univariate PU matrices in \cite{SPL2019} to the array version and determines all the sequences in PU-matrix-based construction \cite{SPL2019}.

\item  Following these two types of seed PU matrices, we introduce two basic recursive PU-matrix-based constructions.  Roughly speaking, the first  one is the product of two PU matrices of order $2^{n+1}$  where the multiplier is a block diagonal matrix of order $2^{n+1}$  by  two BH matrices of order $2^n$ as diagonal matrices, and the multiplicand  is constructed  using  $n$  seed PU matrices of order 2 to form a block diagonal matrix of order $2^{n+1}$ (Theorem 8 in Section 7). The second one is the product of three matrices: one is from the PU matrix in Theorem 8, the second matrix is a block diagonal matrix with two generalized delay matrices as the diagonal matrices, and the third is a block diagonal matrix with two BH matrices as the diagonal matrices (Theorem 9 in Section 7).   Now  the  known GBF-based  constructions of CSSs  by Paterson \cite{Paterson00} and by Schmidt \cite{Schmidt07} and CCCs by Rathinakumar and Chaturvedi \cite{CCC} can be constructed by the PU-matrix-based constructions in Theorem 9, where all BH matrices involved are the Kronecker product of the Walsh matrix of order 2.

\item We further generalize the above recursive constructions on CSSs and CCCs (Theorems 10 and 11 in Section 8).   Those results can be used in bounding  the PMEPR.

\end{enumerate}

Basically, we have explicitly showed that the  known constructions of CSSs and CCCs in \cite{DavisJedwab99, Paterson00,CCC,Schmidt07} can be constructed by our method only involving in seed PU matrices of order 2, or equivalently,  these results on CSSs and CCCs can be  constructed from the Walsh matrix of order 2 or its Kronecker product.  On the other hand,
there are two different equivalent class for quaternary BH matrices of order 4. One is equivalent to the Kronecker product of the Walsh matrix of order 2, and the other is equivalent to the Fourier matrix of order 4. It has been verified in our Construction 4 (in Subsection 5.5) that  any sequences (totally 10 cosets of $S_L(q=4, 4)$) of length $2^4$ related to the Fourier matrix of order 4  were not reported before. Thus, the results for the seed PU matrices of order $4$, combined with  the recursive constructions in this paper, may exponentially increase the number of the quaternary  sequences in CSSs of size $4$. Moreover, there are 15 equivalent classes of quaternary BH matrices of order $8$, and 319 equivalent classes of  BH matrices of order $12$. Thus, the observation in this paper sheds light on the PMEPR control problem by practical codes with higher rate.
It is expected that the results presented in this paper suggest  that we should go back to the first paper for CSSs \cite{Tseng72} where the CSSs are constructed by the Hadamard matrices of order $N$ for finding more new CSSs and CCCs.

The rest of our paper is organized as follows. In the next section, we introduce most of our notations, give a brief viewpoint of $m$-dimensional arrays, and define the concepts of  CCA, CCC, CAS and CSS. The evaluation  from  an array to a sequence which gives a way from a CCA to a large number of CCCs and CSSs is also given. In Section 3, we introduce desired PU matrices,  BH matrices and their equivalence relationship. In Section 4, we  propose a construction of seed PU matrices, and develop a systematic approach to extract the explicit form of functions and sequences from seed PU matrices. Binary sequences in GSPs, ternary sequences in CSSs of size 3, binary and quaternary sequences in CSSs of size 4 are  provided in Section 5. The results for seed PU matrices of order $2^n$ are generalized in Section 6.
We develop a general method to construct the desired PU matrices, and present a framework to recursively construct CCCs and CSSs in Sections 7 and 8. We conclude the paper with some open problems in Section 9.

\section{Preliminaries}

In this section, we introduce the notations of sequence and array, and define the concepts of  CCA, CCC, CAS and CSS. The evaluation  from  an array to a sequence gives our perspective on how to construct a large number of CSSs and CCCs from a single CCA. The following notations will be used throughout the paper.

\begin{itemize}
	\item For integer $p$, $\Z_p=\{0,1, \cdots, p-1\}$ is the residue class ring modulo $p$. $\Z^*_p$ denotes $\Z_p\backslash \{0\}$.

    \item For function $f(y_0,y_1,\cdots, y_{m-1})$, the permutation $\pi$  acting  on the function $f$, denoted as $\pi\cdot f$,  is defined by $\pi\cdot f=f(y_{\pi(0)},y_{\pi(1)},\cdots, y_{\pi(m-1)}).$

    \item $\omega=e^{\frac{2\pi \sqrt{-1}}{q}}$ is a $q$th primitive root of unity.
    \item $\bm{I}_N$ and $\bm{J}_N$ denote the identity matrix and all 1 matrix of order $N$, respectively. $\bm{E}_{i}$ is a matrix with single-entry 1 at $(i, i)$ and zero elsewhere.
\end{itemize}

\subsection{Sequences, Arrays and Generating Functions}

A $q$-ary sequence $\bm{f}$ of length $L$ is defined as
$$\bm{f}=(f(0), f(1),\cdots, f(L-1))$$
for each entry $f(t)\in \mathbb{Z}_q$ ($t\in \Z_L$). This sequence can be represented by a function
$f(t): \Z_L\rightarrow \Z_q,$
 which is referred to as the {\em corresponding function} of sequence $\bm{f}$.

Moreover, the sequence $\bm{f}$ can be associated with a polynomial defined by
\begin{equation}\label{seq-gene}
F(Z)=\sum_{t=0}^{L-1}\omega^{f(t)}Z^t.
\end{equation}
The polynomial $F(Z)$ is called the {\em generating function} of sequence $\bm{f}$.

Note that both the corresponding function $f(t)$ and the generating function $F(Z)$ are uniquely determined by the sequence $\bm{f}$, and vice versa. So we can use either the corresponding function $f(t)$ or the generating function $F(Z)$ to represent a sequence $\bm{f}$.

The above concepts about sequences were  extended to arrays in \cite{Array2}.
An $m$-dimensional $q$-ary array of size $\underbrace{p\times p \times \cdots \times p}_m$ can be represented by a {\em corresponding function}, i.e., a multivariate polynomial function mapping from $\Z_p^m$ to $\Z_q$:
$$f(\bm{y})=f(y_0, y_1, \cdots, y_{m-1}): \Z_p^m \rightarrow \Z_q$$
for $\bm{y}=(y_0, y_1, \cdots, y_{m-1})$ and $y_k\in \Z_{p}$. So in this paper,    an $m$-dimensional $q$-ary array of size $\underbrace{p\times p \times \cdots \times p}_m$ corresponds to
a multivariate polynomial function $f(y_0, y_1, \cdots, y_{m-1})$ mapping from $\Z_p^m$ to $\Z_q$.
In particular, if $p=2$,  An $m$-dimensional $q$-ary array of size $\underbrace{2\times 2 \times \cdots \times 2}_m$ can be realized by a generalized Boolean function (GBF):
$$f(\bm{x})=f(x_0, x_1, \cdots, x_{m-1}): \Z_{2}^m \rightarrow \Z_q.$$
Note that we always use $x_k$ to instead $y_k$ if $y_k$ is a Boolean variable in the whole paper.

The {\em (multivariate) generating function} of array $f(y_0, y_1, \cdots, y_{m-1})$ is defined by
\begin{equation}\label{array-gene}
F(\bm{z})=\sum_{y_0=0}^{p-1}\sum_{y_1=0}^{p-1}\cdots \sum_{y_{m-1}=0}^{p-1}\omega^{f(y_0, y_1, \cdots, y_{m-1})}z_0^{y_0}z_1^{y_1}\cdots z_{m-1}^{y_{m-1}}
\end{equation}
for $\bm{z}=(z_0, z_1, \cdots, z_{m-1})$. The term $z_0^{y_0}z_1^{y_1}\cdots z_{m-1}^{y_{m-1}}$ is denoted by $\bm{z}^{\bm{y}}$ for short in the whole paper.

From the definition above,  a sequence $f(t): \Z_L\rightarrow \Z_q$ can be regarded as an array of dimension $1$.

\subsection{CSS and CCC}
In this subsection, we introduce the definitions of CSS and CCC from the viewpoints of aperiodic correlation and the generating functions of the sequences, respectively.

\begin{definition}
For two $q$-ary sequences $\bm{f}_1$ and $\bm{f}_2$ of length $L$, the {\em aperiodic cross-correlation} of $\bm{f}_1$ and $\bm{f}_2$ at shift $\tau$ ($-L<\tau< L$) is defined by
$$C_{\bm{f}_1,\bm{f}_2}(\tau)=
\left\{
\begin{aligned}
&\sum_{t=0}^{L-1-\tau}{\omega^{\bm{f}_1(t+\tau)-\bm{f}_2(t)}},\
0\leq \tau<L,\\
&\sum_{t=0}^{L-1+\tau}{\omega^{\bm{f}_1(t)-\bm{f}_2(t-\tau)}},\
-L<\tau<0.
\end{aligned} \right.
$$
If $\bm{f}_1=\bm{f}_2=\bm{f}$, the {\em aperiodic autocorrelation} of  sequence $\bm{f}$ at shift $\tau$ is denoted by
$$C_{\bm{f}}(\tau)=C_{\bm{f},\bm{f}}(\tau).$$
\end{definition}

\begin{definition}
A set of $N$ sequences $S=\{\bm{f}_0, \bm{f}_1,\cdots, \bm{f}_{N-1}\}$ is called  a {\em complementary sequence set} (CSS) of size $N$ if
\begin{equation}\label{CSS-def-1}
\sum_{j=0}^{N-1}C_{\bm{f}_j}(\tau)=0\  \mbox{for}\  \tau\neq 0.
\end{equation}
\end{definition}
If the set size $N=2$, such a set is called a {\em Golay sequence pair} (GSP). Each sequence in GSP is called a {\em Golay sequence}.

Two CSSs $S_1=\{\bm{f}_{1,0}, \bm{f}_{1,1},\cdots, \bm{f}_{1,N-1}\}$ and $S_2=\{\bm{f}_{2,0}, \bm{f}_{2,1},\cdots, \bm{f}_{2,N-1}\}$ are said to be {\em mutually orthogonal}
if
\begin{equation}\label{CCC-def-1}
\sum_{j=0}^{N-1}C_{\bm{f}_{1,j},\bm{f}_{2,j}}(\tau)=0, \ \forall \tau.
\end{equation}
It is known that the number of CSSs which are
pairwise mutually orthogonal is at most equal to the number of sequences
in a CSS.

\begin{definition}
Let $S_i=\{\bm{f}_{i,0}, \bm{f}_{i,1},\cdots, \bm{f}_{i,N-1}\}$  be CSSs of size $N$ for $0\leq i<N$, which are pairwise mutually orthogonal. Such a collection of $S_i$ is called {\em complete mutually orthogonal complementary sets} (CMOCS) or {\em complete complementary codes} (CCC).
\end{definition}

The concept of CCC is better to view through a matrix. Let $S$ be a matrix where  the $i$th row is given by $S_i$, i.e.,
\begin{equation} \label{matrix1}
S=\begin{bmatrix}
  \bm{f}_{0,0} & \bm{f}_{0,1} & \cdots & \bm{f}_{0,N-1} \\
  \bm{f}_{1,0} & \bm{f}_{1,1} & \cdots & \bm{f}_{1,N-1} \\
  \vdots & \vdots & \ddots & \vdots \\
  \bm{f}_{N-1,0} & \bm{f}_{N-1,1} & \cdots & \bm{f}_{N-1,N-1} \\
\end{bmatrix}.
\end{equation}
Then $\{\bm{f}_{i,j}\}_{0\leq i,j<N}$ is a CCC if and only if  every row of the matrix $S$ is a CSS of size $N$, and any two rows  are mutually orthogonal.

\begin{example}\label{exam-1-1}
For $q=4, N=2$ and $L=4$, a matrix of sequences given by
\begin{equation*}
S=\begin{bmatrix}
  \bm{f}_{0,0}=(0,1,0,3) & \bm{f}_{0,1}=(0,1,2,1) \\
  \bm{f}_{1,0}=(0,3,0,1) & \bm{f}_{1,1}=(0,3,2,3)  \\
\end{bmatrix}
\end{equation*}
form a CCC of size $2$.
\end{example}

A useful tool for studying  CSS and CCC is the generating function $F(Z)$  of the  sequence $\bm{f}$.  We denote
the complex conjugate of $F(Z)$ by $\overline{F}(Z)$, i.e., $$\overline{F}(Z)=\sum_{t=0}^{L-1}\omega^{-f(t)}Z^t.$$
Suppose that  $F_1(Z)$ and $F_2(Z)$ are the generating functions of sequences $\bm{f}_1$ and $\bm{f}_2$, respectively. It is easy to verify
\begin{equation}\label{seq-correlation}
F_1(Z)\overline{F}_2(Z^{-1})=\sum_{\tau=1-N}^{N-1}C_{\bm{f}_1,\bm{f}_2}(\tau)Z^{\tau}.
\end{equation}
Then
$S=\{\bm{f}_0, \bm{f}_1,\cdots, \bm{f}_{N-1}\}$  forms a CSS if and only if their  generating functions
$\{F_0(Z), F_1(Z),$ $\cdots, F_{N-1}(Z)\}$ satisfy
\begin{equation}\label{CSS-def-2}
\sum_{j=0}^{N-1} F_j(Z)\overline{F}_j(Z^{-1})=NL.
\end{equation}
Moreover, two CSSs $S_1=\{\bm{f}_{1,0}, \bm{f}_{1,1},\cdots, \bm{f}_{1,N-1}\}$ and $S_2=\{\bm{f}_{2,0}, \bm{f}_{2,1},\cdots, \bm{f}_{2,N-1}\}$ are {\em mutually orthogonal} if and only if their respective generating functions $\{F_{i,0}(Z), F_{i,1}(Z),\cdots, F_{i, N-1}(Z)\}$ ($i=1,2$) satisfy
\begin{equation}\label{CCC-def-2}
\sum_{j=0}^{N-1} F_{1,j}(Z)\overline{F}_{2,j}(Z^{-1})=0.
\end{equation}
Thus the conditions given in the equations (\ref{CSS-def-1}) and (\ref{CCC-def-1}), defined by aperiodic correlation of the sequences,  are equivalent to the  conditions in  (\ref{CSS-def-2}) and (\ref{CCC-def-2}) defined  by their generating functions.

\subsection{CAS and CCA}

In this subsection, we generalize the concepts of CSS and CCC to CAS and CCA, respectively.

\begin{definition}
Let  $f_1(\bm{y})$ and $f_2(\bm{y})$ be two arrays from $\Z_p^m$ to $\Z_q$. For $\bm{\tau}=(\tau_0, \tau_1, \cdots ,\tau_{m-1})$ ($1-p\le \tau_k\le p-1$), \lq\lq$\bm{y}+\bm{\tau}$\rq\rq  is the element-wise addition of vectors over $\Z$. The aperiodic cross-correlation of arrays $f_1$ and $f_2$ at shift $\bm{\tau}=(\tau_0, \tau_1, \cdots ,\tau_{m-1})$ is defined by
$$C_{f_1,f_2}(\bm{\tau})=
\sum_{\bm{y}\in \mathbb{Z}_p^m}\omega^{f_1(\bm{y}+\bm{\tau})-f_2(\bm{y})},
$$
where $\omega^{f_1(\bm{y}+\bm{\tau})-f_2(\bm{y})}=0$ if  $f_1(\bm{y}+\bm{\tau})$ or $f_2(\bm{y})$ is not defined.
If $f_1=f_2=f$, then it is called the {\em aperiodic autocorrelation} of array $f$ at shift $\bm{\tau}$,  denoted by
$$C_f(\bm{\tau})=C_{f,f}(\bm{\tau}).$$
\end{definition}

Suppose that $F_1(\bm{z})$ and $F_2(\bm{z})$ are the generating functions of arrays $f_1(\bm{y})$ and $f_2(\bm{y})$, respectively. Here we denote $\bm{z}^{-1}=(z_0^{-1},z_1^{-1},\cdots, z_{m-1}^{-1})$ for short. Similar to the case of sequence, we have
\begin{equation}\label{array-correlation}
F_1(\bm{z})\cdot \overline{F}_2(\bm{z}^{-1})=\sum_{\bm{\tau}}C_{f_1,f_2}(\bm{\tau})z_0^{
\tau_0}z_1^{\tau_1}\cdots z_{m-1}^{\tau_{m-1}}.
\end{equation}

\begin{definition}
A set of arrays $\{f_0, f_1,\cdots, f_{N-1}\}$ from $\Z_p^m$ to $\Z_q$ is called  a {\em complementary array set} (CAS) of size $N$ if their generating functions $\{F_0(\bm{z}), F_1(\bm{z}),\cdots, F_{N-1}(\bm{z})\}$ satisfy
\begin{equation}\label{CAS-def-2}
\sum_{j=0}^{N-1}F_j(\bm{z})\cdot \overline{F}_j(\bm{z}^{-1})=N\cdot p^{m}.
\end{equation}
\end{definition}
If the set size $N=2$, such a set is called a {\em Golay array pair} (GAP). Each array (i.e.,  the function) in GAP is called a {\em Golay array}.

Two CASs $S_1=\{f_{1,0}, f_{1,1},\cdots, f_{1,N-1}\}$ and $S_2=\{f_{2,0}, f_{2,1},\cdots, f_{2,N-1}\}$ are said to be {\em mutually orthogonal}
if  their generating functions  $\{F_{i,0}(\bm{z}), F_{i,1}(\bm{z}),\cdots, F_{i, N-1}(\bm{z})\}$ ($i=1,2$) satisfy
\begin{equation}\label{CCA-def-2}
\sum_{j=0}^{N-1} F_{1,j}(\bm{z})\overline{F}_{2,j}(\bm{z}^{-1})=0.
\end{equation}

\begin{definition}
Let $S_i=\{f_{i,0}, f_{i,1},\cdots, f_{i,N-1}\}$ ($0\leq i<N$) be CASs  of size $N$, which  are pairwise mutually orthogonal. We call such a collection of $S_i$ ($0\leq i<N$) a {\em complete mutually orthogonal array set} or a {\em complete complementary arrays} (CCA).
\end{definition}

\begin{remark}
It is obvious that the equation (\ref{CAS-def-2}) is equivalent to $\sum_{j=0}^{N-1}C_{f_j}(\bm{\tau})=0$ for $\forall \bm{\tau}\neq \bm{0}$, and the equation (\ref{CCA-def-2}) is equivalent to $C_{f_1,f_2}(\bm{\tau})=0$ for $\forall \bm{\tau}$.
\end{remark}

\subsection{Relationship Between  Arrays and Sequences}

For an  array $f(\bm{y}): \Z_{p}^m \rightarrow \Z_q$, we define a sequence $f(t): \Z_{p^m}\rightarrow \Z_q$ by
  $$f(t)=f(\bm{y})\  \mbox{for}\ t=\sum_{k=0}^{m-1}y_k\cdot p^k.$$
We say that  the sequence $f(t)$ is evaluated by the function (or array) $f(\bm{y})$ in this paper.

      For an array $f(\bm{y}): \Z_{p}^m \rightarrow \Z_q$ and its  evaluated sequence $f(t): \Z_{p^m}\rightarrow \Z_q$, their respective generating functions $F(\bm{z})$ and $F(Z)$ can be connected by restricting $z_k=Z^{p^k}$. The derivation  is given as follows:
\begin{eqnarray*}
F(\bm{z})
&=&F(Z^{p^0}, Z^{p^1}, \cdots, Z^{p^{m-1}})\\
&=&\sum_{y_0,y_1,\cdots, y_{m-1}}\omega^{f(y_0,y_1,\cdots, y_{m-1})}\cdot Z^{y_0p^0}Z^{y_1p^1}\cdots Z^{y_{m-1}p^{m-1}}\\
&=&\sum_{y_0,y_1,\cdots, y_{m-1}}\omega^{f(y_0,y_1,\cdots, y_{m-1})}\cdot Z^{\sum_{k=0}^{m-1} y_kp^{k}}\\
&=&\sum_{t=0}^{L-1}\omega^{f(t)}\cdot Z^t\\
&=&F(Z).
\end{eqnarray*}

We summarize the above discussions in Figure \ref{fig-1}.
\begin{figure}[h]
	\centering
	$ \begin{gathered}\xymatrix@C=5.5cm@R=1cm{
		{\displaystyle{f(\bm{y})}
			\ar@{<->}[r]^{\textstyle{F(\bm{z})=\sum_{\bm{y}}{\omega^{f(\bm{y})}} \bm{z}^{\bm{y}}}} }
		\ar[d]^{\textstyle t=\sum_{k=0}^{m-1}y_k\cdot p^k}
		& F(\bm{z})\ar[d]^{\textstyle z_k=Z^{p^k}}\\
		f(t)\ar@{<->}[r]_{\textstyle{F(Z)=\sum_{t=0}^{p^m-1}\omega^{f(t)}Z^t}} & F(Z) \\
	}
	\end{gathered} $
	\caption{A Diagram of the Relationship between  Arrays and   Sequences, and the Respective Generating Functions}\label{fig-1}
\end{figure}
\begin{example}\label{exam-1-2}
For $q=4$, $p=m=2$, let $\omega=e^{\frac{\pi \sqrt{-1}}{2}}$ be a $4$th primitive root of unity and GBF $f(\bm{x})=f(x_0,x_1)=2x_0x_1+x_0$. The generating function of array $f(\bm{x})$ is given by
$$F(\bm{z})=F(z_0,z_1)=\omega^0+\omega^1z_0+\omega^0z_1+\omega^3z_0z_1.$$
The sequence evaluated by the GBF $f(\bm{x})$ is $\bm{f}=(0,1,0,3)$ with  generating function
$$F(Z)=\omega^0+\omega^1Z+\omega^0Z^2+\omega^3Z^3.$$
It is clear that $F(\bm{z})=F(Z)$ by setting $z_0=Z$ and $z_1=Z^2$.
\end{example}

Thus, if $\{f_0, f_1,\cdots, f_{N-1}\}$ constitutes a CAS, then its evaluated  sequences constitute a CSS. Moreover, recall the permutation $\pi$  acting on the array:
$\pi\cdot f=f(\pi\cdot\bm{y})=f(y_{\pi(0)},y_{\pi(1)},\cdots, y_{\pi(m-1)}).$
Since the equations (\ref{CAS-def-2}) and (\ref{CCA-def-2}) still hold if we apply a permutation acting on the arrays, it is easy to know that the set $\{f_0, f_1,\cdots, f_{N-1}\}$ is a CAS if and only if $\{\pi\cdot f_0, \pi\cdot f_1,\cdots, \pi\cdot f_{N-1}\}$ is also a CAS.

\begin{property}\label{prop-1}
\begin{itemize}
\item[(1)] If a set of arrays $\{f_0, f_1,\cdots, f_{N-1}\}$  is a CAS, then the sequences evaluated  by functions  $\{\pi\cdot f_0, \pi\cdot f_1,\cdots, \pi\cdot f_{N-1}\}$ form a CSS for $\forall \pi$.
\item[(2)]    Similarly, if $S_i=\{f_{i,0}, f_{i,1},\cdots, f_{i,N-1}\}$ ($0\leq i<N$) is a CCA, the sequences evaluated by functions
$S'_i=\{\pi\cdot f_{i,0}, \pi\cdot f_{i,1},\cdots, \pi\cdot f_{i,N-1}\}$ ($0\leq i<N$) form a CCC for $\forall \pi$.
\end{itemize}
\end{property}

From Property \ref{prop-1}, we can construct a large number of CSSs from a CAS (or a large number of CCCs from a CCA). This is one of the main reasons that we study the CAS and CCA instead of CSS and CCC.

\subsection{Special Case for $p=2$}
 We set $p=2$ in this subsection.  For GBFs $f_1, f_2$, and any affine GBF $f'=\sum_{k=0}^{m-1}c_kx_k+c'$ ($c_k, c'\in \Z_q$), we have
\begin{equation}\label{affine-fun}
C_{f_1+f',f_2+f'}(\bm{\tau})=
\omega^{f'(\bm{\tau})-c'}\sum_{\bm{x}\in \mathbb{Z}_2^m}\omega^{f_1(\bm{x}+\bm{\tau})-f_2(\bm{x})}=\omega^{f'(\bm{\tau})-c'}C_{f_1, f_2}(\bm{\tau}) .
\end{equation}
Then the  following assertions are  immediately obtained from the definitions.

\begin{property}\label{prop-2}
For arbitrary permutation $\pi$ and affine function $f'$,
\begin{itemize}
\item[(1)] if a set of GBFs $\{f_0, f_1,\cdots, f_{N-1}\}$ is a CAS,  the sequences evaluated by $\{\pi\cdot f_0+f', \pi\cdot f_1+f',\cdots, \pi\cdot f_{N-1}+f'\}$
form a CSS; and
\item[(2)] if $S_i=\{f_{i,0}, f_{i,1},\cdots, f_{i,N-1}\}$ ($0\leq i<N$) is a CCA, the sequences evaluated by $S'_i=\{\pi\cdot f_{i,0}+f', \pi\cdot f_{i,1}+f',\cdots, \pi\cdot f_{i,N-1}+f'\}$ ($0\leq i<N$) form  a CCC.
\end{itemize}
\end{property}

Note that the above property was  proved in  \cite[Lemma 8]{Array2} for the pair case.
The process from the array  $f$ to the sequences evaluated by $(\pi\cdot f_0+f')$ in this paper is actually equivalent to the simplified  three-stage process in \cite{Array2}.

From Property \ref{prop-2}, we can construct a large number of CSSs (or CCCs) from a CAS (or CCA). On the other hand, the set $\{\pi\cdot f+f'\}$  where $\pi$ can be any permutation  and $f'$ can be  any affine GBF, must be comprised of some cosets of the first-order GRM code. This is  another strong supporting evidence that we investigate the constructions for CSSs through the constructions of  CASs.

For a given GBF $f$, let $N_0$ be the number of pairs $(\pi, f')$ such that $\pi\cdot f+f'=f$, and $N_1$ be  the number of functions in the set $\{\pi\cdot f+f'\}$. From the orbit-stabilizer theorem \cite{Serge2002}, we have
$$N_1=\frac{m!\cdot q^{m+1}}{N_0}.$$

\begin{example}\label{exam-1-3}
Let $q$ be an even integer. It is known that  Rudin-Shapiro function $f(\bm{x})=\frac{q}{2}\sum_{k=0}^{m-2}x_kx_{k+1}$ and $f(\bm{x})+\frac{q}{2}x_0$ form a CAP. From Property \ref{prop-2}, the sequences evaluated by
$$\left\{
\begin{aligned}
&\pi\cdot f+f'=\frac{q}{2}\sum_{k=1}^{m-1} x_{\pi(k-1)}x_{\pi(k)}+\sum_{k=0}^{m-1}c_kx_{k}+c',\\
&\pi\cdot f+f'+x_{\pi(0)}\\
\end{aligned}\right.
$$
form a GSP for  $\forall \pi, c_k, c'\in \Z_q$, which agree with the standard Golay pair in \cite{DavisJedwab99}. Moreover, the number of the standard Golay sequences  can be  calculated by $N_1=\frac{1}{2}m!q^{m+1}$ for $N_0=2$.
\end{example}

\section{Para-unitary Matrices and Hadamard Matrices}

For $0\leq i,j \leq N-1$, let $f_{i,j}(\bm{y})$ be an $m$-dimensional $q$-ary arrays of size $\underbrace{p\times p \times \cdots \times p}_m$, i,e., $f_{i,j}(\bm{y})$ is a function from $\Z_p^m$ to $\Z_q$. We introduce three types of matrices.

The first type of matrices is called the {\em function matrices},  denoted by $\bm{{\widetilde{M}}}(\bm{y})$ where the entries are denoted as ${\widetilde{{M}}}_{i,j}(\bm{y})={f}_{i,j}(\bm{y})$ for $0\leq i,j \leq N-1$, i.e.,
\begin{equation}\label{func-matrix}
\bm{{\widetilde{M}}}({\bm{y}})=	\begin{bmatrix}
f_{0,0}({\bm{y}})&f_{0,1}({\bm{y}})&\dots&f_{0,N-1}({\bm{y}})\\
f_{1,0}({\bm{y}})&f_{1,1}({\bm{y}})&\dots&f_{1,N-1}({\bm{y}})\\
\vdots             &\vdots             &\ddots&\vdots             \\
f_{N-1,0}({\bm{y}})&f_{N-1,1}({\bm{y}})&\dots&f_{N-1,N-1}({\bm{y}})\\
\end{bmatrix}.
\end{equation}
The permutation $\pi$ action on the matrix $\bm{{\widetilde{M}}}(\bm{y})$ is denoted by
$\pi\cdot \widetilde{\bm{M}}(\bm{y})$ where each entry is given by
$$(\pi\cdot \widetilde{\bm{M}}(\bm{y}))_{i,j}=\pi\cdot f_{i,j}(\bm{y})$$
for $0\leq i,j\leq N-1$.

The second type of matrices is called the {\em generating (polynomial) matrices},  denoted by $\bm{M}(\bm{z})$ where each entry is given by ${{M}}_{i,j}(\bm{z})={F}_{i,j}(\bm{y})$, the generating function of array ${f}_{i,j}(\bm{y})$, i.e.,
\begin{equation}\label{gene-matrix}
\bm{M}({\bm{z}})=	\begin{bmatrix}
F_{0,0}({\bm{z}})&F_{0,1}({\bm{z}})&\dots&F_{0,N-1}({\bm{z}})\\
F_{1,0}({\bm{z}})&F_{1,1}({\bm{z}})&\dots&F_{1,N-1}({\bm{z}})\\
\vdots             &\vdots             &\ddots&\vdots             \\
F_{N-1,0}({\bm{z}})&F_{N-1,1}({\bm{z}})&\dots&F_{N-1,N-1}({\bm{z}})\\
\end{bmatrix}.
\end{equation}
The third type of the matrices is denoted by $\bm{M}(\bm{y})$ where the entries are ${{M}}_{i,j}(\bm{y})=\omega^{{f}_{i,j}(\bm{y})}$, i.e.,
\begin{equation}\label{polyphase-func-matrix}
{\bm{{M}}}(\bm{y})=	\begin{bmatrix}
\omega^{{f}_{0,0}(\bm{y})}&\omega^{{f}_{0,1}(\bm{y})}&\dots&\omega^{{f}_{0,N-1}(\bm{y})}\\
\omega^{{f}_{1,0}(\bm{y})}&\omega^{{f}_{1,1}(\bm{y})}&\dots&\omega^{{f}_{1,N-1}(\bm{y})}\\
\vdots             &\vdots             &\ddots&\vdots             \\
\omega^{{f}_{N-1,0}(\bm{y})}&\omega^{{f}_{N-1,1}(\bm{y})}&\dots&\omega^{{f}_{N-1,N-1}(\bm{y})}\\
\end{bmatrix}.
\end{equation}
The matrix $\bm{M}(\bm{y})$ will be frequently used in the proofs of the paper, since the generating matrix can be written in the form of a multivariate polynomial matrix
\begin{equation}\label{gene-matrix-dec}
\bm{M}(\bm{z})=\sum_{0\leq y_0, y_1, \cdots, y_{m-1}< p}\bm{M}(\bm{y})\bm{z}^{\bm{y}},
\end{equation}
where $\bm{M}(\bm{y})$ is the coefficient matrix of $\bm{M}(\bm{z})$ of the monomial term $z_0^{y_0}z_1^{y_1}\cdots z_{m-1}^{y_{m-1}}$.

\subsection{Multivariate Para-unitary Matrices and CCA}

A multivariate para-unitary matrix is straightforwardly generalized from a univariate para-unitary matrix, which plays a central role in signal processing, in particular in the areas of filter banks and wavelets.

\begin{definition}
Let $\bm{M}(\bm{z})$ be a square multivariate polynomial matrix of order $N$. If
\begin{equation*}
\bm{M}(\bm{z})\cdot\bm{M}^{\dagger}(\bm{z}^{-1})=c\cdot \bm{I}_N,
\end{equation*}
where $(\cdot)^\dagger$ denotes the Hermitian transpose and $c$ is a real constant.
We say that $\bm{M}(\bm{z})$ is a {\em multivariate para-unitary  (PU) matrix}.
\end{definition}

Recall the definitions of function matrices and generating matrices in (\ref{func-matrix}) and (\ref{gene-matrix}), respectively. The following assertions follow immediately from the definition of PU matrices and Properties \ref{prop-1} and \ref{prop-2}.

\begin{theorem}\label{thm-1}
If a PU matrix $\bm{M}(\bm{z})$ is a generating matrix of function matrix $\bm{{\widetilde{M}}}(\bm{y})$, we have
\begin{enumerate}
\item[(1)] The arrays $\{{f}_{i,j}(\bm{y})\}$ in function matrix $\bm{{\widetilde{M}}}(\bm{y})$ form a CCA.
\item[(2)] For $\forall \pi$,  the sequences evaluated by the functions in $\pi\cdot \widetilde{\bm{M}}(\bm{y})$ form a CCC.
\item[(3)] The arrays $\{{f}_{i,j}(\bm{y})\}$ in each row (or column) of $\bm{{\widetilde{M}}}(\bm{y})$  form a CAS.
\item[(4)] For $\forall \pi$, the sequences evaluated by the functions in  each row (or column) of the matrix $\pi\cdot \widetilde{\bm{M}}(\bm{y})$ form a  $q$-ary CSS of size $N$ and length $p^m$.
\item[(5)] If $p=2$,  for $\forall \pi$ and affine GBF $f'$, the sequences evaluated by the functions in $\left(\pi\cdot \widetilde{\bm{M}}(\bm{x})+f'(x)\cdot \bm{J}_N\right)$ form a CCC, and the sequences evaluated by the functions in  each row (or column) of the matrix $\left(\pi\cdot \widetilde{\bm{M}}(\bm{x})+f'(x)\cdot \bm{J}_N\right)$ form a  $q$-ary CCS of size $N$ and length $2^m$.
\end{enumerate}
\end{theorem}

According to Theorem \ref{thm-1}, a large number of CSSs and CCCs can be constructed if $\bm{M}(\bm{z})$ is not only a PU matrix, but also the generating matrix of a
function matrix ${\bm{\widetilde{M}}}(\bm{y})$. Such a matrix $\bm{M}(\bm{z})$ is called a {\em desired PU matrix} in the rest of the paper.

One principal objective of this paper is to find the desired PU matrices.
However, each entry of a desired PU matrix $\bm{M}(\bm{z})$, denoted by $M_{i,j}(\bm{z})$,
is a multivariate polynomial, but not a function from $\Z_p^m$ to $\Z_q$. So another principal objective of this paper is to exact the function matrix $\bm{{\widetilde{M}}}(\bm{y})$ from its generating matrix $\bm{M}(\bm{z})$.

\subsection{Butson-Type Hadamard Matrices}

A complex Hadamard matrix of order $N$ is a complex matrix $\bm{H}$ of order $N$
satisfying $|H_{i,j}|=1$ ($0\leq i,j<N$) and
$\bm{H}\cdot\bm{H}^{\dagger}=N\cdot \bm{I}_N$.
A complex Hadamard matrix $\bm{H}$ of order $N$ is
called {\em Butson-type} \cite{Butson62} if all the
entries of \(\bm{H}\) are $q$th roots of unity.  For fixed $N$ and $q$,  we denote the set of all Butson-type
Hadamard (BH) matrices by $H(q,N)$.

It is obvious that the BH matrix is the simplest desired PU matrix for $p=m=1$, i.e., the entry of the desired PU matrix is associated with a sequence of length $1$.

Two BH matrices, $\bm{H}_1, \bm{H}_2\in H(q,N)$ are
called {\em equivalent}, denoted by $\bm{H}_1\simeq \bm{H}_2$, if there exist
diagonal unitary matrices $\bm{Q}_1, \bm{Q}_2$ where each diagonal entry
 is a $q$th root of unity and permutation matrices \(\bm{P}_1\),
\(\bm{P}_2\) such that:

\begin{equation}\label{Hadamard-1}
\bm{H}_1 = \bm{P}_1 \cdot\bm{Q}_1 \cdot\bm{H}_2\cdot \bm{Q}_2 \cdot \bm{P}_2.
\end{equation}

\begin{example}\label{exam-1}
For $q$ odd, $H(q,2)=\varnothing$.
For $q$ even, all Hadamard  matrices in $H(q,2)$ are
equivalent to the Walsh matrix of order 2:
$$
\begin{bmatrix}
1 & 1 \\
  1 &-1
\end{bmatrix}.
$$
\end{example}

We introduce the Kronecker product of matrices.
\begin{definition}
If $\bm{A}=\{a_{i,j}\}$ is a square matrix of order $n_1$ and $\bm{B}$ is a square matrix of order $n_2$, then the Kronecker product  $\bm{A}\otimes\bm{B}$ is a block matrix of  order $n_1\cdot n_2$, given by
\begin{equation*}
 \bm{A}\otimes\bm{B}=\begin{bmatrix}a_{0,0}\cdot\bm{B}&a_{0,1}\cdot\bm{B}&\cdots&a_{0,n_1-1}\cdot\bm{B}\\
 a_{1,0}\cdot\bm{B}&a_{1,1}\cdot\bm{B}&\cdots&a_{1,n_1-1}\cdot\bm{B}\\
 \vdots&\vdots&\ddots&\vdots\\
 a_{n_1-1,0}\cdot\bm{B}&a_{n_1-1,1}\cdot\bm{B}&\cdots&a_{n_1-1,n_1-1}\cdot\bm{B}
 \end{bmatrix}.
\end{equation*}
\end{definition}

\begin{example}\label{exam-2}
All binary Hadamard  matrices in  \(H(2,4)\) are
equivalent to the Hadamard matrix
$$
\begin{bmatrix}
1 & 1 &1&1\\1 & -1&1&-1\\
 1 &1&-1&-1\\
  1 &-1&-1& 1
\end{bmatrix},
$$
which is the Kronecker product of Walsh matrix of order 2.
\end{example}

\begin{example}\label{exam-3}
All  quaternary  BH  matrices in  \(H(4,4)\) are
equivalent to one of  the following two
BH matrices:
$$
\begin{bmatrix}
1 & 1 &1&1\\1 & -1&1&-1\\
 1 &1&-1&-1\\
  1 &-1&-1& 1
\end{bmatrix}
~and~ \begin{bmatrix}
1 & 1 &1&1\\1 &\sqrt{-1} &-1&-\sqrt{-1}\\
 1 &-1&1&-1\\
  1 &-\sqrt{-1}&-1& \sqrt{-1}
\end{bmatrix}.
$$
The first matrix is the Kronecker product of Walsh matrix of order 2, and the second matrix is the Fourier matrix of order 4.
\end{example}

\subsection{Phase Matrices of Hadamard Matrices}\label{phase-hadamard}

For a BH matrix $\bm{H}\in H(q,N)$, define its phase matrix $\widetilde{\bm{H}}$ by
$\widetilde{H}_{i,j}=s$ if $H_{i,j}=\omega^s$. Actually, the phase matrix $\widetilde{\bm{H}}$ is the function matrix of the generating matrix $\bm{H}$. For given $N$ and $q$, we denote $\widetilde{H}(q,N)$ the set consisting of all phase matrices of BH matrices.

\begin{definition}
Let $\widetilde{\bm{H}}_1$ and $\widetilde{\bm{H}}_2$ be two phase matrices induced by BH matrices $\bm{H}_1$ and $\bm{H}_2$, respectively. The equivalence of the phase matrices is defined by the equivalence of the Hadamard matrices, i.e., $\widetilde{\bm{H}}_1\simeq \widetilde{\bm{H}}_2$ if and only if $\bm{H}_1\simeq \bm{H}_2$.
\end{definition}




\begin{example}\label{exam-4}
For $q$ odd, $\widetilde{H}(q,2)=\varnothing$.
For $q$ even, all phase  matrices in $\widetilde{H}(q,2)$ are
equivalent to
$$
\begin{bmatrix}
0 & 0 \\
  0 &q/2
\end{bmatrix}.
$$
\end{example}

\begin{example}\label{exam-5}
All binary phase  matrices in  \(\widetilde{H}(2,4)\) are
equivalent to the $4\times 4$ phase matrix:
$$
\begin{bmatrix}
0 & 0 &0&0\\
0 & 1&0&1\\
 0 &0&1&1\\
  0 &1&1& 0
\end{bmatrix}.
$$
\end{example}

\begin{example}\label{exam-6}
All quaternary   phase  matrices in  \(\widetilde{H}(4,4)\) are
equivalent to one of  the following two
 $4\times 4$ phase matrices:
$$
\begin{bmatrix}
0 & 0 &0&0\\
0 & 2&0&2\\
 0 &0&2&2\\
  0 &2&2& 0
\end{bmatrix}
~and~ \begin{bmatrix}
0 & 0 &0&0\\
0 &1 &2&3\\
 0 &2&0&2\\
 0 &3&2& 1
\end{bmatrix}.
$$
\end{example}

The representatives of the equivalent classes in $\widetilde{H}(q,N)$ play a crucial role in the process of extracting explicit forms of the functions from the desired PU matrix. We denoted $S_{\widetilde{H}}(q, N)$ the set consisting of all the representatives of the equivalent classes of the phase matrices.

\section{Seed PU Matrices and Corresponding Function Matrices}

In this section and the next section, we always assume $N=p$. We will first introduce a construction of the desired PU matrices, and then develop a new approach to study the explicit function forms of arrays in  corresponding function matrices.

\subsection{Construction of Seed PU Matrices}

Let $\bm{D}(z)$, called the {\em delay matrix},  be a $p$ by $p$ diagonal
matrix with the form $\bm{D}(z)=diag\{z^0, z^1, z^2,$ $\cdots, z^{p-1}\}$, i.e.,
\begin{equation}\label{delay-matrix}
\bm{D}(z)=\begin{bmatrix}z^0&0&0&0&\cdots&0\\0&z^1&0&0&\cdots&0
\\0&0&z^2&0&\cdots&0\\0&0&0&z^3&\cdots&0\\\vdots&\vdots&\vdots&\vdots&\ddots&\vdots\\0&0&0& 0&\cdots&z^{p-1}\end{bmatrix}.
\end{equation}
Let $\bm{H}^{\{k\}}$ be an arbitrary BH matrix chosen from  $H(q, N)$ for $0\leq k\leq m$. Define a multivariate polynomial matrix $\bm{M}(\bm{z})$ as follow.
\begin{eqnarray}\label{seed-PU}
\bm{M}(\bm{z})&=&\bm{H}^{\{0\}}\cdot \bm{D}(z_0)\cdot
\bm{H}^{\{1\}}\cdot\bm{D}(z_1)\cdots \bm{H}^{\{m-1\}}\cdot
\bm{D}(z_{m-1})\cdot \bm{H}^{\{m\}}\\
&=&\left(\prod_{k=0}^{m-1}\left(\bm{H}^{\{k\}}\cdot \bm{D}(z_k)\right)\right)\cdot \bm{H}^{\{m\}}\nonumber.
\end{eqnarray}

\begin{theorem}\label{thm-2}
$\bm{M}(\bm{z})$ defined above is a desired PU matrix, i.e.,
\begin{itemize}
	\item[(1)] $\bm{M}(\bm{z})\cdot\bm{M}^{\dagger}(\bm{z}^{-1})=N^{m+1}\cdot \bm{I}_N$;
	\item[(2)] Each entry of $\bm{M}(\bm{z})$ can be expressed by the generating function of an array $f(\bm{y}): \Z_p^m \rightarrow \Z_q$.
\end{itemize}
\end{theorem}
$\bm{M}(\bm{z})$ defined in (\ref{seed-PU}) is called the {\em seed PU matrix}.

\subsection{The Proof for Seed PU Matrices}

The first assertion in Theorem \ref{thm-2} on the para-unitary condition follows from
$$\bm{H}^{\{k\}}\cdot\bm{H}^{\{k\}\dagger}=N\cdot \bm{I}_N$$
and
$$\bm{D}(z_k)\cdot\bm{D}^{\dagger}(z_k^{-1})=\bm{I}_N.$$

We expand $\bm{M}(\bm{z})$ in the  form:
\begin{equation}\label{matrix-ex}
\bm{M}(\bm{z})=\sum_{y_0=0}^{p-1}\sum_{y_1=0}^{p-1}\cdots \sum_{y_{m-1}=0}^{p-1}\bm{M}(y_0, y_1, \cdots ,y_{m-1})\bm{z}^{\bm{y}}.
\end{equation}
Note that the second assertion in Theorem \ref{thm-2} is valid if and only if all the entries of $\bm{M}(y_0, y_1, \cdots, y_{m-1})$ are $q$th roots of unity for $0\leq y_0, y_1, \cdots, y_{m-1}\leq p-1$.

The coefficient matrices $\bm{M}(y_0, y_1, \cdots ,y_{m-1})$ can be determined by the following lemma.

\begin{lemma}\label{lem-1}
The coefficient matrix $\bm{M}(y_0, y_1, \cdots, y_{m-1})$ can be determined by BH matrices as follow.
\begin{equation}\label{co-matrix}
\bm{M}(y_0, y_1, \cdots, y_{m-1})=\prod^{m-1}_{k=0}\left(\bm{H}^{\{k\}}\cdot
\bm{E}_{y_k}\right)\cdot \bm{H}^{\{m\}}.
\end{equation}
\end{lemma}

\begin{proof}
We write the delay matrix $\bm{D}(z_k)$ in the form
$$ \bm{D}(z_k)=\sum_{y_k=0}^{p-1} \bm{E}_{y_k}z_k^{y_k}.$$
Then multivariate polynomial matrix $\bm{M}(\bm{z})$ can be re-expressed as follows.
\begin{eqnarray*}
\bm{M}(\bm{z})
&=&\left(\prod_{k=0}^{m-1}\left(\bm{H}^{\{k\}}\cdot \bm{D}(z_k)\right)\right)\cdot \bm{H}^{\{m\}}\\
&=&\left(\prod_{k=0}^{m-1}\left(\bm{H}^{\{k\}}\cdot \sum_{y_k=0}^{p-1} \bm{E}_{y_k}z_k^{y_k}\right)\right)\cdot \bm{H}^{\{m\}}\\
&=&\left(\prod_{k=0}^{m-1}\sum_{y_k=0}^{p-1} \left(\bm{H}^{\{k\}}\cdot \bm{E}_{y_k}z_k^{y_k}\right)\right)\cdot \bm{H}^{\{m\}}\\
&=&\left(\sum_{y_0=0}^{p-1}\sum_{y_1=0}^{p-1}\cdots \sum_{y_{m-1}=0}^{p-1} \prod_{k=0}^{m-1} \left(\bm{H}^{\{k\}}\cdot \bm{E}_{y_k}z_k^{y_k}\right)\right)\cdot \bm{H}^{\{m\}}\\
&=&\sum_{y_0=0}^{p-1}\sum_{y_1=0}^{p-1}\cdots \sum_{y_{m-1}=0}^{p-1} \left(\prod_{k=0}^{m-1} \left(\bm{H}^{\{k\}}\cdot \bm{E}_{y_k}\right)\cdot \bm{H}^{\{m\}}z_0^{y_0}z_1^{y_1}\cdots z_{m-1}^{y_{m-1}}\right).
\end{eqnarray*}
We complete the proof according to the expansion of  $\bm{M}(\bm{z})$ in (\ref{matrix-ex}).
\end{proof}

\begin{lemma}\label{lem-2}
Let $M_{i,j}(y_0, y_1, \cdots, y_{m-1})$ be the entry of the coefficient matrix  $\bm{M}(y_0, y_1, \cdots, y_{m-1})$ at $(i, j)$. Then $M_{i,j}(y_0, y_1, \cdots, y_{m-1})$ must be a product of some entries of BH matrices $\bm{H}^{\{k\}}$. More precisely, we have
$$M_{i,j}(y_0, y_1, \cdots, y_{m-1})=H^{\{0\}}_{i,y_0}\left(\prod^{m-1}_{k=1}H^{\{k\}}_{y_{k-1},y_k}
\right)\cdot H^{\{m\}}_{y_{m-1},j}.$$
\end{lemma}
\begin{proof}
It follows from Lemma \ref{lem-1} and $\bm{E}_{y_k}$ being a matrix with single-entry 1 at $(y_k, y_k)$ and zero elsewhere.
\end{proof}

According to Lemma \ref{lem-2}, we know that $M_{i,j}(y_0, y_1, \cdots, y_{m-1})$ must be a $q$th root of unity, denoted by $M_{i,j}(y_0, y_1, \cdots y_{m-1})=\omega^{f_{i,j}(y_0,y_1,\cdots, y_{m-1})}$. Then each entry of $\bm{M}(\bm{z})$ can be written in the form
\begin{eqnarray}
M_{i,j}(\bm{z})&=&\sum_{y_0, y_1, \cdots ,y_{m-1}}M_{i,j}(y_0, y_1, \cdots ,y_{m-1})\bm{z}^{\bm{y}} \nonumber \\
&=&\sum_{y_0,y_1,\cdots, y_{m-1}}\omega^{f_{i,j}(y_0,y_1,\cdots, y_{m-1})}\bm{z}^{\bm{y}} \label{PU-function},
\end{eqnarray}
which completes the proof of  Theorem \ref{thm-2}. \done

\begin{remark}
If we restrict  $z_k=Z^{p^{\pi(k)}}$ in  multivariate seed PU matrices in this subsection, we obtain the univariate PU matrices proposed in \cite{TSP2018,SETA2016}.
\end{remark}

\subsection{Extracting Functions from Seed PU Matrices}

In this subsection, we introduce an approach to extract the explicit function form $f_{i,j}(y_0,y_1,\cdots, y_{m-1})$ shown in (\ref{PU-function}) in  the corresponding function matrices of the seed PU matrices.

Recall $N=p$ in this section.  We first introduce a basis of the functions from $\mathbb{Z}_N$ to  $\mathbb{Z}_q$.

\begin{definition}\label{basis}
Let $g_i$ be a function: $\mathbb{Z}_N\rightarrow \mathbb{Z}_q$ such that $g_i(j)=\delta_{i,j}$ where $\delta_{i,j}$ is the Kronecker-delta function, i.e.,
$$g_i(j)=\left\{
\begin{aligned}
&1, \mbox{if} \ i=j, \\
&0, \mbox{if} \ i\neq j.
\end{aligned} \right.$$
\end{definition}
Then any function $g$: $\mathbb{Z}_N\rightarrow \mathbb{Z}_q$ can be represented by
$$g=\sum_{i=0}^{N-1}g(i)g_i.$$

\begin{remark}
An alternative terminology to easily understand the above notations and the results in this subsection is group ring and module theory \cite{Serge2002}. Let $(G, \cdot)$ be a semi-group and $\Z_q$ a ring. The semigroup ring of $G$ over $\Z_q$, denoted by $\Z_q[G]$, is the set of mappings $g$: $G \rightarrow \mathbb{Z}_q$. Furthermore, $\Z_q[G]$ is a $\Z_q$-module. In this subsection, the semi-group $(G, \cdot)$ can be replaced by $(\Z_N, \cdot)$ or $(\Z_N^m, \cdot)$ where the operator $(\cdot)$ denotes production and element-wise production, respectively.
\end{remark}

\begin{example}\label{exam-7}
If $N=p=2^n$,  $\forall y\in \Z_N$ can be represented by a vector $(x_0, x_1, \cdots, x_{n-1})$  where $y=\sum_{v=0}^{n-1}x_v\cdot2^v$ and $x_v\in \Z_2$. Let $x_v^{\{0\}}$ denote the negation of $x_v$, i.e.  $x_v^{\{0\}}=1-x_v$. Conversely, let $x_v^{\{1\}}=x_v$. We have
\begin{equation*}
g_{i}(x_0, x_1, \cdots, x_{n-1})=\prod_{v=0}^{n-1}x_v^{\{i_v\}},
\end{equation*}
where $i=\sum_{v=0}^{n-1}i_v\cdot2^v$.
\end{example}

From Lemma \ref{lem-2} and formula (\ref{PU-function}), we have
\begin{equation}\label{function-0}
\omega^{f_{i,j}(y_0,y_1,\cdots, y_{m-1})}=H^{\{0\}}_{i,y_0}\left(\prod^{m-1}_{k=1}H^{\{k\}}_{y_{k-1},y_k}
\right)\cdot H^{\{m\}}_{y_{m-1},j}.
\end{equation}
Let $\widetilde{\bm{H}}^{\{k\}}$ be the phase matrix of the BH matrix $\bm{H}^{\{k\}}$ for $0\leq k\leq m$. The function $f_{i,j}$ can be written in the following way.
\begin{equation}\label{function-1}
f_{i,j}(y_0,y_1,\cdots, y_{m-1})=\widetilde{H}^{\{0\}}_{i,y_0}+\sum^{m-1}_{k=1}\widetilde{H}^{\{k\}}_{y_{k-1},y_k}
+ \widetilde{H}^{\{m\}}_{y_{m-1},j}.
\end{equation}
If we denote $y_{-1}=i$ and $y_{m}=j$ for convenience, the above equation can be re-expressed in a unified form:
\begin{equation}\label{function-2}
f_{i,j}(y_0,y_1,\cdots, y_{m-1})=\sum^{m}_{k=0}\widetilde{H}^{\{k\}}_{y_{k-1},y_k}.
\end{equation}

\begin{lemma}\label{lem-3}
Let  $\widetilde{H}_{y_{\theta}, y_{\eta}}$ be the entry of  $\widetilde{\bm{H}}$ at $(y_{\theta}, y_{\eta})$, and $h(y_{\theta}, y_{\eta})$  a function from $\mathbb{Z}^2_N$ to $\mathbb{Z}_q$ such that $h(y_{\theta}, y_{\eta})=\widetilde{H}_{y_{\theta}, y_{\eta}}$. Then $h(y_{\theta}, y_{\eta})$ can be represented by the following form:
\begin{equation}\label{base-1}
h(y_{\theta}, y_{\eta})=\sum_{i=0}^{N-1}\sum_{j=0}^{N-1}\widetilde{H}_{i, j}g_i(y_{\theta})g_j(y_{\eta}).
\end{equation}
\end{lemma}
\begin{proof}
Recall the definition of basis functions $g_i$ in Definition \ref{basis}. It follows from
$g_i(y_{\theta})=g_j(y_{\eta})=1 \Longleftrightarrow y_{\theta}=i$ and $y_{\eta}=j$.
\end{proof}

Let $\bm{g}(y): \mathbb{Z}_N\rightarrow \mathbb{Z}_q^N$ be a vector function  represented by
\begin{equation}
\bm{g}(y)=(g_0(y), g_1(y), \cdots, g_{N-1}(y)).
\end{equation}
Then the function $h(y_{\theta}, y_{\eta})$ in (\ref{base-1}) can be represented as a matrix form:
\begin{equation}\label{base-2}
h(y_{\theta}, y_{\eta})=\bm{g}(y_{\theta})\widetilde{\bm{H}} \bm{g}(y_{\eta})^T,
\end{equation}
where $\bm{g}(y_{\eta})^T$ denotes the transpose of the vector function $\bm{g}(y_{\eta})$.
According to Lemma \ref{lem-3} and formulas (\ref{function-2}) and (\ref{base-2}), we obtain a formula for computing the function $f_{i,j}$.
\begin{lemma}\label{lem-4}
The function $f_{i,j}$ can be expressed in the form:
\begin{equation}\label{function4}
f_{i,j}(y_0,y_1,\cdots, y_{m-1})=\sum_{k=0}^{m}\bm{g}(y_{k-1})\widetilde{\bm{H}}^{\{k\}} \bm{g}(y_{k})^T.
\end{equation}
\end{lemma}

It seems that we have already obtained an explicit form of $f_{i,j}$ for $m+1$ given phase matrices $\widetilde{\bm{H}}^{\{k\}}$. However, each phase matrix can be randomly chosen from $\widetilde{H}(q,N)$ which is a large set.  To obtain a general form
of $f_{i,j}$, we take the equivalence of phase matrices into our consideration.

\begin{lemma}\label{lem-5}
Assume that $\bm{Q}_1=diag\{\omega^{c_0}, \omega^{c_1}, \cdots, \omega^{c_{N-1}}\}$ and $\bm{Q}_2=diag\{\omega^{d_0}, \omega^{d_1}, \cdots, \omega^{d_{N-1}}\}$, $\bm{H}_1 = \bm{Q}_1 \bm{H} \bm{Q}_2$ for $\bm{H}, \bm{H}_1\in H(q, N)$, and $\widetilde{\bm{H}}, \widetilde{\bm{H}}_1$ are the phase matrices of $\bm{H}$ and $\bm{H}_1$, respectively. We have
$$\bm{g}(y_{\theta})\widetilde{\bm{H}}_1 \bm{g}(y_{\eta})^T=\bm{g}(y_{\theta})\widetilde{\bm{H}} \bm{g}(y_{\eta})^T+\sum_{i=0}^{N-1}c_ig_i(y_{\theta})+\sum_{j=0}^{N-1}d_jg_j(y_{\eta}).$$
\end{lemma}
\begin{proof}
From the definition, we know the entry $\{\widetilde{H_1}\}_{i,j}=\{\widetilde{H}\}_{i,j}+c_i+d_j$. Thus
\begin{eqnarray*}
\bm{g}(y_{\theta})\widetilde{\bm{H}_1} \bm{g}(y_{\eta})^T
&=&\bm{g}(y_{\theta})\widetilde{\bm{H}} \bm{g}(y_{\eta})^T+\sum_{i=0}^{N-1}\sum_{j=0}^{N-1}(c_i+d_j)g_i(y_{\theta})g_j(y_{\eta})\\
&=&\bm{g}(y_{\theta})\widetilde{\bm{H}} \bm{g}(y_{\eta})^T+\sum_{i=0}^{N-1}\sum_{j=0}^{N-1}c_ig_i(y_{\theta})g_j(y_{\eta})+\sum_{i=0}^{N-1}\sum_{j=0}^{N-1} d_jg_i(y_{\theta})g_j(y_{\eta})\\
&=&\bm{g}(y_{\theta})\widetilde{\bm{H}} \bm{g}(y_{\eta})^T+\sum_{i=0}^{N-1} c_ig_i(y_{\theta}) \left(\sum_{j=0}^{N-1}g_j(y_{\eta})\right)+\sum_{j=0}^{N-1} d_jg_j(y_{\eta}) \left(\sum_{i=0}^{N-1} g_i(y_{\theta})\right).
\end{eqnarray*}
From the definition of $g_i$, we have
$$\sum_{j=0}^{N-1}g_j(y_{\eta})=\sum_{i=0}^{N-1} g_i(y_{\theta})=1,$$
which completes the proof.
\end{proof}

Let  $\chi$ be a permutation of symbols $\{0, 1, \cdots, N-1\}$ and $\bm{g}(y)=((g_0(y), g_1(y), \cdots, g_{N-1}(y))$.  Denote the permutation of the vector function $\bm{g}(y)$ by
\begin{equation}\label{chi}
\bm{g}_{\chi}(y)=(g_{\chi(0)}(y), g_{\chi(1)}(y), \cdots, g_{\chi(N-1)}(y)).
\end{equation}
\begin{lemma}\label{lem-6}
Assume that $\bm{H}_2 = \bm{P}_1 \bm{H} \bm{P}_2$, where $\bm{P}_1$ and $\bm{P}_2$ are both permutation matrices, and $\bm{H}, \bm{H}_2\in H(q, N)$.
There exist $\chi_L$ and $\chi_R$, which are both permutations of symbols $\{0, 1, \cdots, N-1\}$, such that
$$\bm{g}(y_{\theta})\widetilde{\bm{H}_2} \bm{g}(y_{\eta})^T=\bm{g}_{\chi_L}(y_{\theta})\widetilde{\bm{H}} \bm{g}_{\chi_R}(y_{\eta})^T.$$
\end{lemma}
\begin{proof} There exists $\chi_L$ such that $\bm{g}_{\chi_L}(y_{\theta})=\bm{g}(y_{\theta})\bm{P}_1$ and
$\chi_R$ such that $\bm{g}_{\chi_R}(y_{\eta})=\bm{g}(y_{\eta})\bm{P}_2^T$, respectively. Then we have
\begin{eqnarray*}
\bm{g}(y_{\theta})\widetilde{\bm{H}_2} \bm{g}(y_{\eta})^T
&=&\bm{g}(y_{\theta})\bm{P}_1\widetilde{\bm{H}} \bm{P}_2\bm{g}(y_{\eta})^T\\
&=&\left(\bm{g}(y_{\theta})\bm{P}_1\right)\widetilde{\bm{H}} \left(\bm{P}_2\bm{g}(y_{\eta})^T\right)\\
&=&\bm{g}_{\chi_L}(y_{\theta})\widetilde{\bm{H}} \bm{g}_{\chi_R}(y_{\eta})^T,
\end{eqnarray*}
which completes the proof.
\end{proof}

\begin{corollary}\label{rep1}
The functions extracted from the seed PU matrices can be represented by a general form:
\begin{equation}\label{func-rep1}
f(y_0,y_1,\cdots, y_{m-1})=\sum_{k=1}^{m-1}\bm{g}_{\chi_{L_k}}(y_{k-1})\widetilde{\bm{H}}^{\{k\}} (\bm{g}_{\chi_{R_k}}(y_{k}))^T+\sum_{k=0}^{m-1}\sum_{i=0}^{N-1}c_{k,i}g_i(y_{k}),
\end{equation}
where $\chi_{L_k}$ and $\chi_{R_k}$ are both permutations of symbols $\{0, 1, \cdots, {N-1}\}$, $\widetilde{H}^{\{k\}}\in S_{\widetilde{H}}(q, N)$ for $1\leq k\leq m-1$, $c_{k,i}\in \Z_q$ for $0\leq k< m, 0\leq i\leq N-1$.
\end{corollary}

\begin{proof}
It is straightforward that $\sum_{i=0}^{N-1}c_ig_{\chi(i)}(y)$ is still a linear combination of the functions $g_{i}(y)$.
According to Lemmas \ref{lem-5} and \ref{lem-6}, the functions extracted from the seed PU matrices can be expressed in the form:
\begin{equation}\label{func-rep4}
f(y_0,y_1,\cdots, y_{m-1})=\sum_{k=0}^{m}\bm{g}_{\chi_{L_k}}(y_{k-1})\widetilde{\bm{H}}^{\{k\}} (\bm{g}_{\chi_{R_k}}(y_{k}))^T+\sum_{k=0}^{m-1}\sum_{i=0}^{N-1}c_{k,i}g_i(y_{k}),
\end{equation}
where $y_{-1}=i$, $y_{m}=j$. However, for the case $k=0$ in the first term of the above formula, we have
$$\bm{g}_{\chi_{L}}(y_{-1}=i)\widetilde{\bm{H}}^{\{0\}}\bm{g}_{\chi_{R}}(y_{0})^T=\sum_{i=0}^{N-1}c_ig_i(y_0),$$
which is a  linear combination of $g_i(y_0)$. Similarly, for the case $k=m$ in (\ref{func-rep4}), it is still a linear combination of $g_i(y_{m-1})$. Then formula (\ref{func-rep1}) is obtained from formula (\ref{func-rep4}) immediately.
\end{proof}

From Corollary \ref{rep1}, we only need to study two kinds of terms in formula (\ref{func-rep1}). The first one is called a
{\em quadratic term}: $$\bm{g}_{\chi_{L}}(y_{0})\widetilde{\bm{H}}\bm{g}_{\chi_{R}}(y_{1})^T, $$
 and the another one is called a {\em linear term}: $$\sum_{k=0}^{m-1}\sum_{i=0}^{N-1}c_{k,i}g_i(y_{k}).$$
 Here the words \lq\lq quadratic\rq\rq \ and \lq\lq linear\rq\rq \  are with respect to the functions $g_i(y_k)$ for $0\leq i\leq N-1, 0\leq k\leq m-1$.

A collection of the linear terms is denoted by $S_L(q, N)$ in the rest of the paper. We study the algebraic structure of the $S_L(q, N)$ firstly.
\begin{lemma}\label{lem-7}
Every linear term is a linear combination of the terms $g_i(y_{k})$ and $1$ over $\Z_q$ for $0\leq k\leq m-1, 1\leq i\leq N-1$, i.e.,
\begin{equation}\label{linear-terms}
S_L(q, N)=\left\{\sum_{k=0}^{m-1}\sum_{i=1}^{N-1}c_{k,i}g_i(y_{k})+c'\bigg| c_{k,i}, c'\in \Z_q, 0\leq k\leq m-1, 1\leq i\leq N-1 \right\}.
\end{equation}
Moreover, the terms $g_i(y_{k})$ and $1$ are linear independent over $\Z_q$ for $0\leq k\leq m-1, 1\leq i\leq N-1$. The number of the  linear  terms is given by
$$|S_L(q, N)|=q^{Nm-m+1}.$$
\end{lemma}
\begin{proof}
 From the definition of $g_i$, it is known
$\sum_{i=0}^{N-1} g_i(y)=1$ for $\forall y\in \Z_N$. Moreover, for $G=(\Z_N^m, \cdot)$,  $\Z_q[G]$ is a semigroup ring, and also a $\Z_q$-module. Treat $1$ as the identity element in $\Z_q[G]$, or the function $\bm{1}: G\rightarrow \Z_q$. Then $S_L(q, N)$ is a free $\Z_q$-submodule of dimension $m(N-1)+1$ with basis $\{g_i(y_{k}), 1| 0\leq k\leq m-1, 1\leq i\leq N-1 \}.$
\end{proof}

Now we study the quadratic term. Let $S'_Q(q, N)$ be a set consisting of all the functions of the form $\bm{g}_{\chi_{L}}(y_{0})\widetilde{\bm{H}}\bm{g}_{\chi_{R}}(y_{1})^T$ for  $\forall \widetilde{\bm{H}}\in S_{\widetilde{H}}(q, N)$.
\begin{definition}\label{func-equiv}
Two functions $h_1(y_0,y_1), h_2(y_0,y_1)\in S'_Q(q, N)$  are said to be {\em equivalent}, if the difference between the functions $h_1(y_0,y_1)$ and $h_2(y_0,y_1)$ is a linear combination of $g_i(y_0)$ and $g_i(y_1)$ for $0\leq i<N$.  Let $S_Q(q, N)$ be a set consisting of  all the representatives of the  functions in $S'_Q(q, N)$ with respect to this equivalent relationship.
\end{definition}

\begin{remark}
The set consisting of  all the following terms
$$\sum_{i=0}^{N-1}c_ig_i(y_0)+\sum_{j=0}^{N-1}d_jg_j(y_1)$$ is a subgroup of the additive group $\Z_q[y_0, y_1]$,  so the definition of equivalence in Definition \ref{func-equiv} is reasonable. Moreover,
for $G=(\Z_N^m, \cdot)$, $\Z_q[G]$ is a $\Z_q$-module, and $S_L(q, N)$ is a free $\Z_q$-submodule. $S_Q(q, N)$ can be treated as the quotient set of $S'_Q(q, N)$ modulo $S_L(q, N)$.
\end{remark}

\begin{theorem}\label{thm-3}
All the  functions extracted from the seed PU matrices  can be represented in a general form
\begin{equation}\label{func-rep2}
f(y_0,y_1,\cdots, y_{m-1})=\sum_{k=1}^{m-1} h_k(y_{k-1}, y_{k})+\sum_{k=0}^{m-1}\sum_{i=1}^{N-1}c_{k,i}g_i(y_{k})+c',
\end{equation}
where $h_k(\cdot,\cdot)\in S_Q(q, N) (1\leq k\leq m-1)$, $c_{k,i}, c'\in \Z_q, (0\leq k\leq m-1, 1\leq i\leq N-1)$.
\end{theorem}
\begin{proof}
It follows from Corollary \ref{rep1}, Lemma \ref{lem-7} and Definition \ref{func-equiv}.
\end{proof}

\begin{corollary}\label{enum-1}
Theorem \ref{thm-3} determines
$$|S_Q(q, N)|^{m-1}\cdot q^{Nm-m+1}$$
distinct  functions, where $|S_Q(q, N)|$ denotes the number of the functions in the set $S_Q(q, N)$.
\end{corollary}
\begin{proof}
The set consisting of  all the functions with the  form (\ref{func-rep2}) is a union of the cosets of $S_L(q, N)$.
We can prove that
$$\sum_{k=1}^{m-1} h_k(y_{k-1}, y_{k})-\sum_{k=1}^{m-1} h'_k(y_{k-1}, y_{k})\in S_L(q, N)$$
for $h_k(\cdot,\cdot), h'_k(\cdot,\cdot)\in S_Q(q, N)$ if and only if
$h_k=h'_k$ for $1\leq k\leq m-1$. So different choice of $h_k$ leads to different coset leaders, and the number of the coset leaders is $|S_Q(q, N)|^{m-1}$.
\end{proof}

The following result follows immediately from Lemma \ref{lem-4} and Theorem \ref{thm-3}.
\begin{theorem}\label{thm-4}
The entries of the corresponding function matrix of the seed PU matrix $\bm{M}(z)$ in (\ref{seed-PU})  at position $(i, j)$ ($0\leq i,j< N$) can be expressed by $$f_{i,j}(y_0,y_1,\cdots, y_{m-1})=f(y_0,y_1,\cdots, y_{m-1})+h(i,y_0)+h'(y_{m-1},j),$$
where $f(y_0,y_1,\cdots, y_{m-1})$ is a function with the form (\ref{func-rep2}) and $h(\cdot,\cdot), h'(\cdot,\cdot)\in S_Q(q, N)$.
\end{theorem}

\section{Parameterized Constructions and Examples}

We still assume $N=p$ in this section.

\subsection{Parameterized Constructions}

By applying Theorem \ref{thm-1} in Section 3 and Theorem \ref{thm-4} in Section 4, we obtain the following results immediately.
\begin{theorem}\label{construction}
Let $f(y_0,y_1,\cdots, y_{m-1})$ be a function with the form  (\ref{func-rep2}) and $h(\cdot,\cdot), h'(\cdot,\cdot)\in S_Q(q, N)$.
\begin{itemize}
	\item[(1)] The following arrays form a CCA of size $N$:
$$f_{i,j}(y_0,y_1,\cdots, y_{m-1})=f(y_0,y_1,\cdots, y_{m-1})+h(i,y_0)+h'(y_{m-1},j), \ 0\leq i,j< N. $$
\item[(2)] For $\forall \pi$, the  sequences evaluated by following functions form a CCC of size $N$:
$$\pi\cdot f_{i,j}=f_{i,j}(y_{\pi(0)},y_{\pi(1)},\cdots, y_{\pi(m-1)}), 0\leq i, j< N.$$
	\item[(3)] The following arrays form a CAS of size $N$:
$$f_i(y_0,y_1,\cdots, y_{m-1})=f(y_0,y_1,\cdots, y_{m-1})+h(i,y_0),\  0\leq i< N.$$
\item[(4)] For $\forall \pi$, the $q$-ary sequences evaluated by following functions form a CSS of size $N$ and length $N^m$:
$$\pi\cdot f_{i}=f_i(y_{\pi(0)},y_{\pi(1)},\cdots, y_{\pi(m-1)}),\  0\leq i< N.$$
\end{itemize}
\end{theorem}

We denote a set consisting of all the sequences in the CSSs constructed in Theorem \ref{construction}  by $S(q,N)$, i.e.,
\begin{equation}\label{seq-rep}
S(q,N)=\left\{\sum_{k=1}^{m-1} h_k(y_{\pi(k-1)}, y_{\pi(k)})+\sum_{k=0}^{m-1}\sum_{i=1}^{N-1}c_{k,i}g_i(y_{k})+c'\bigg|c_{k,i}, c'\in \Z_q \right\}.
\end{equation}
Obviously, $S(q,N)$ is a union of the cosets of linear terms $S_L(q,N)$ with coset representatives
$$\sum_{k=1}^{m-1} h_k(y_{\pi(k-1)}, y_{\pi(k)}),$$
where $h_k(\cdot,\cdot)\in S_Q(q, N) (1\leq k\leq m-1)$. Each coset representative can be associated with a labeled graph on $m$ vertices as shown in Figure \ref{fig-2}, where we label the vertices of the coset representative by $y_{\pi(0)}, y_{\pi(1)}, \cdots y_{\pi(m-1)}$ and join vertices $y_{\pi(k-1)}$ and $y_{\pi(k)}$ by an edge labeled $h_k(\cdot,\cdot)$.

\begin{figure}
\centering
\begin{tikzpicture}
\draw(0,0)--(4,0);\draw[dashed](4,0)--(6,0);\draw[dashed](8,0)--(10,0);
\fill (0,0) circle(2pt);\node[above] at (0,0){$y_{\pi(0)}$};
\node[below] at (1,0){$h_{1}(\cdot,\cdot)$};
\fill (2,0) circle(2pt);\node[above] at (2,0){$y_{\pi(1)}$};
\node[below] at (3,0){$h_{2}(\cdot,\cdot)$};
\fill (4,0) circle(2pt);\node[above] at (4,0){$y_{\pi(2)}$};
\node[below] at (5,0){$h_{3}(\cdot,\cdot)$};
\fill (10,0) circle(2pt);\node[above] at (10,0){$y_{\pi(m-1)}$};
\node[below] at (9,0){$h_{m-1}(\cdot,\cdot)$};
\end{tikzpicture}
\caption{The Graph of Coset Representatives}\label{fig-2}
\end{figure}
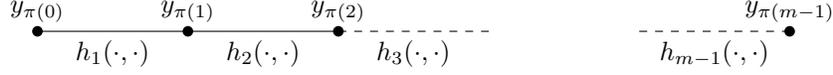

\begin{corollary}\label{seq-enum}
The set $S(q,N)$ determines exactly
$$\frac{1}{2}m!\cdot |S_Q(q, N)|^{m-1}\cdot q^{Nm-m+1}$$
distinct sequences.
\end{corollary}
\begin{proof}
The sequences evaluated by the functions in the set $S(q,N)$ can be obtained by permutations acting on the functions in (\ref{func-rep2}). There are  $m!$ permutations and $|S_Q(q, N)|^{m-1}\cdot q^{Nm-m+1}$ different functions in  (\ref{func-rep2}) from Corollary \ref{enum-1}, respectively.

For $G=(\Z_N^m, \cdot)$, $\Z_q[G]$ is a $\Z_q$-module and $S_L(q, N)$ is a $\Z_q$-submodule. For $\forall f \in \Z_q[G]$, denote the quotient element of $f$ in the quotient module, $\Z_q[G]$ modulo $S_L(q, N)$, by $\overline{f}$.

Let $f$ and $f'$ be two functions with the forms
$$f=\sum_{k=1}^{m-1} h_k(y_{k-1}, y_{k})+\sum_{k=0}^{m-1}\sum_{i=1}^{N-1}c_{k,i}g_i(y_{k})+c'$$
and
$$f'=\sum_{k=1}^{m-1} h'_k(y_{k-1}, y_{k})+\sum_{k=0}^{m-1}\sum_{i=1}^{N-1}c'_{k,i}g_i(y_{k})+c''. $$
If $\pi\cdot f=\pi'\cdot  f'$, we must have $\overline{\pi\cdot f}=\overline{\pi'\cdot f'}$, i.e.,
$$\sum_{k=1}^{m-1} \overline{h}_k(y_{\pi(k-1)}, y_{\pi(k)})=\sum_{k=1}^{m-1} \overline{h'}_k(y_{\pi'(k-1)}. y_{\pi'(k)}),$$
which leads to
\begin{itemize}
\item[(1)] $\pi=\pi', \overline{h}_k(y_0,y_1)-\overline{h'}_k(y_0,y_1)=0$ for $1\leq k\leq m-1$;
\item[(2)] $\pi(k)=\pi'(m-k), \overline{h}_k(y_0,y_1)-\overline{h'}_{m-k}(y_1,y_0)=0$
for $1\leq k\leq m-1$.
\end{itemize}
From the definition of $S_Q(q,N)$, the above two conditions are equivalent to
\begin{itemize}
\item[(1)] $\pi=\pi', {h}_k(y_0,y_1)={h'}_k(y_0,y_1)$ for $1\leq k\leq m-1$;
\item[(2)] $\pi(k)=\pi'(m-k), {h}_k(y_0,y_1)-{h'}_{m-k}(y_1,y_0)\in S_L(q, N)$
for $1\leq k\leq m-1$.
\end{itemize}

Assume that ${h}_k(y_0,y_1)$ is a quadratic term such that ${h}_k(y_0,y_1)=\bm{g}_{\chi_{L}}(y_{0})\widetilde{\bm{H}}\bm{g}_{\chi_{R}}(y_{1})^T$ for permutations $\chi_{L}, \chi_{R}$ and Hadamard matrix $\bm{H}$. The transpose of $\bm{H}$, denoted by $\bm{H}^{T}$, is still a Hadamard matrix. Then ${h}_k(y_1,y_0)$   is still a quadratic term, since
$${h}_k(y_1,y_0)=\bm{g}_{\chi_{L}}(y_{1})\widetilde{\bm{H}}\bm{g}_{\chi_{R}}(y_{0})^T=\bm{g}_{\chi_{R}}(y_{0})\widetilde{\bm{H}^T}\bm{g}_{\chi_{L}}(y_{1})^T. $$
Let ${h'}_{m-k}(y_0,y_1)$ be the representative of the quadratic term ${h}_k(y_1,y_0)$ (modulo $S_L(q, N)$) in $ S_Q(q, N)$. Then  ${h'}_{m-k}$ in the above case (2) must be exist and unique.

From the discussion above, for a given pair $(f, \pi)$, there exists exactly another pair $(f', \pi')$, which is determined by the case (2) in the above condition, such that $\pi\cdot f=\pi'\cdot  f'$. Thus, the number of the sequences in set $S(q,N)$ is equal to half of the  product of the number of the permutations and the number of the functions in (\ref{func-rep2}).
\end{proof}
\begin{remark}
It was shown in  \cite[Section VI]{TSP2018}  there are $m!\cdot \left(\frac{(N-1)!)^2}{n_0}\right)^{m}\cdot q^{Nm-m+1}$ distinct sequences, where $n_0$ is a specified number. The results in Corollary \ref{seq-enum} correct the enumeration in \cite{TSP2018}, since the explicit function form is given here while \cite{TSP2018} did not. For example, for the case $N=2$ and the case $N=q=3$, $n_0$ in \cite{TSP2018} equals to Euler function $\varphi(N)$. For the case $N=2$, Corollary \ref{seq-enum} in this paper determines $\frac{1}{2}m!\cdot q^{m+1}$ distinct Golay sequences, whereas there are $m!\cdot q^{m+1}$ in \cite{TSP2018}. For the case $N=q=3$, Corollary \ref{seq-enum} determines $m!\cdot2^{m-2}\cdot3^{2m+1}$  distinct ternary sequences, whereas there are $m!\cdot2^{m} 3^{2m+1}$ in \cite{TSP2018}. A general explicit form of the functions constructed here will be given in Subsections 5.2 and 5.3 for the cases of $N=2$ and $N=q=3$ respectively.
\end{remark}

We now compute the set $S_L(q, N)$ and $S_Q(q, N)$ to determine the explicit form of functions for given $q$ and $N$. The detailed processes are shown as follows.

\begin{itemize}
\item[(1)]Determine the set $S_{\widetilde{H}}(q, N)$, which consists of all the representatives of phase matrices in $\widetilde{H}(q, N)$.
\item[(2)]Find the functions $g_i$: $\mathbb{Z}_N\rightarrow \mathbb{Z}_q$ such that $g_i(j)=\delta_{i,j}$ for $i,j\in \Z_N$. Then compute the set
$$S_L(q, N)=\left\{\sum_{k=0}^{m-1}\sum_{i=1}^{N-1}c_{k,i}g_i(y_{k})+c'\bigg| c_{k,i}, c'\in \Z_q, 0\leq k\leq m-1, 1\leq i\leq N-1 \right\}.$$
\item[(3)]Compute the bivariate functions by $\bm{g}_{\chi_{L}}(y_{0})\widetilde{\bm{H}}\bm{g}_{\chi_{R}}(y_{1})^T$ for $\forall \widetilde{\bm{H}}\in S_{\widetilde{H}}(q, N)$, and $\forall \chi_{L}, \chi_{R}$ being the permutations of symbols $\{0, 1, \cdots, {N-1}\}$. The set $S_Q(q, N)$ is obtained by collecting all the representatives of the functions of the form $\bm{g}_{\chi_{L}}(y_{0})\widetilde{\bm{H}}\bm{g}_{\chi_{R}}(y_{1})^T$ modulo $S_L(q, N)$.
\item[(4)]Any function extracted from the seed PU matrices can be expressed in a general form
$$f=\sum_{k=1}^{m-1} h_k(y_{k-1}, y_{k})+S_L(q, N), $$
where $h_k(\cdot,\cdot)\in S_Q(q, N)$ for $1\leq k\leq m-1$.
\item[(5)]CCAs, CASs, CCCs and CSSs are constructed by Theorem \ref{construction}.
\end{itemize}

For Step (1), there is a large volume of the research papers on the existences, enumerations and constructions of the BH matrices. A source of BH matrices with small $N$ and $q$ can be found in \cite{Hada}. We will see from the following subsections that the computations on Step (2) is easy, but the computations on Step (3) is heavy. It is obvious that Steps (4) and (5) are straightforward if Steps (2) and (3) are done. We give the cases for $N=2$ and $q$ even, $N=q=3$, and $N=4$ $q=2$ or $4$ in the rest of the four subsections.

\subsection{Case $N=2$: Golay Sequences}

From Example \ref{exam-1}, BH matrices of order 2 do not exist for $q$ odd. Let $N=2$ and $q$  even here.

(1) Determine the set  $S_{\widetilde{H}}(q, 2)$.
From Example \ref{exam-4}, there is only one phase matrix in $S_{\widetilde{H}}(q, 2)$:
$$
\begin{bmatrix}
0 & 0 \\
  0 &q/2
\end{bmatrix}.
$$

(2) Computation of the set $S_L(q, 2)$.
From Example \ref{exam-7}, we have the basis functions
$$\left\{
\begin{aligned}
g_0(x)&=1-x,\\
g_1(x)&=x.
\end{aligned}
\right.
$$
Then we obtain the linear terms
$$S_L(q, 2)=\left\{\sum_{k=0}^{m-1}c_kx_{k}+c'\bigg| c_k, c'\in \Z_q, 0\leq k\leq m-1 \right\}.$$

(3) Computation of the set $S_Q(q, 2)$. We have
$$\bm{g}(x)=(g_0(x), g_1(x))=(1-x, x)=(-1,1)x+(1,0)\cdot 1.$$
The term $(1,0)\cdot 1$ only produces the linear terms in $\bm{g}_{\chi_{L}}(x_{0})\widetilde{\bm{H}}\bm{g}_{\chi_{R}}(x_{1})$. So we only need to consider the term $\bm{g}'(x)=(-1,1)x.$ Let $\chi$ be a permutation of symbols $\{0, 1\}$. Then $\bm{g}'_{\chi}(x)=(-1,1)x$ or $\bm{g}'_{\chi}(x)=(1,-1)x$. For any permutations $\chi_{L}$ and $\chi_{R}$, we have
$$\bm{g}'_{\chi_L}(x_0)
\begin{bmatrix}
0 & 0 \\
  0 &q/2
\end{bmatrix}
\bm{g}'_{\chi_R}(x_1)^T
={\chi_L}(-1,1)
\begin{bmatrix}
0 & 0 \\
  0 &q/2
\end{bmatrix}
{\chi_R}(-1,1)^Tx_0x_1=\frac{q}{2}x_0x_1.
$$
So $S_Q(q, 2)$ contains only one function: $\frac{q}{2}x_0x_1$.

By applying  Theorems \ref{thm-3} and \ref{construction},  we obtain the following construction.
\begin{construction}
Let $\pi$ be a permutation of $\{0, 1, \cdots, m-1\}$. For $N=2$ and $q$ even, we have
\begin{itemize}
\item[(1)]
The Golay arrays extracted from the seed PU matrices can be expressed  by
\begin{equation}\label{extrac-func-1}
f(x_0,x_1,\cdots, x_{m-1})=\frac{q}{2}\sum_{k=1}^{m-1} x_{k-1}x_{k}+\sum_{k=0}^{m-1}c_kx_{k}+c',
\end{equation}
for  $c_k, c'\in \Z_q, 0\leq k\leq m-1$.
\item[(2)]
The corresponding function matrix $\widetilde{\bm{M}}(\bm{x})$  is given by
$$\widetilde{\bm{M}}(\bm{x})=f\cdot \bm{J}_2+\frac{q}{2}\cdot
\begin{bmatrix}
     0 & x_{m-1} \\
     x_{0} & x_0+x_{m-1} \\
   \end{bmatrix},
$$
where the GBF $f$ is in the form (\ref{extrac-func-1}). The sequences evaluated by the matrix $\pi\cdot \widetilde{\bm{M}}(\bm{x})$ form a CCC of size $2$.
\item[(3)]
The set $S(q,2)$ consisting of all the Golay sequences  derived from the seed PU matrices of order $2$ is given by
$$S(q,2)=\left\{\frac{q}{2}\sum_{k=1}^{m-1} x_{\pi(k-1)}x_{\pi(k)}+\sum_{k=0}^{m-1}c_kx_{k}+c' \bigg| \forall \pi, c_k, c'\in \Z_q, 0\leq k\leq m-1 \right\}.$$
\item[(4)] The sequences evaluated by
$$\left\{
\begin{aligned}
&f,\\
&f+\frac{q}{2}x_{\pi(0)}\\
\end{aligned}\right.
$$
form a GSP for $\forall f\in S(q,2)$.
\end{itemize}
\end{construction}

This construction  coincides with the well known results on standard Golay sequences \cite{DavisJedwab99} by Davis and Jedwab. The enumeration in Corollary \ref{seq-enum} for $N=2$ also agrees with the number of the standard Golay sequences.

\subsection{Case $N=q=3$: Ternary Sequences of Size 3.}

Let $N=q=3$ in this subsection.

(1) Determine the set $S_{\widetilde{H}}(3, 3)$.
There is only one phase matrix in  $\widetilde{H}(3,3)$:
$$
\begin{bmatrix}
0 & 0 & 0\\
0 & 1 & 2\\
0 & 2 & 1
\end{bmatrix}.
$$

(2) Computation of the set $S_L(3, 3)$.
From Definition \ref{basis}, it is easy to verify that
$$\left\{
\begin{aligned}
g_0(y)&=2y^2+1,\\
g_1(y)&=2y^2+2y,\\
g_2(y)&=2y^2+y.
\end{aligned}
\right.
$$
Then we obtain the linear terms
$$S_L(3, 3)=\left\{\sum_{k=0}^{m-1}c_{k_2}y_{k}^2+\sum_{k=0}^{m-1}c_{k_1}y_{k}+c'\bigg| c_{k_1},c_{k_2}, c'\in \Z_3, 0\leq k\leq m-1 \right\}.$$

(3) Computation of the set $S_Q(3, 3)$. We have
\begin{eqnarray*}
\bm{g}(y)&=&(g_0(y), g_1(y), g_2(y))\\
&=&(2y^2+1, 2y^2+2y, 2y^2+y)\\
&=&(2,2,2)y^2+(0,2,1)y+(1,0,0)\cdot 1.
\end{eqnarray*}

The term $(1,0,0)\cdot 1$ only produces the linear terms with respect to $S_L(3, 3)$ in $\bm{g}_{\chi_{L}}(y_{0})\widetilde{\bm{H}}\bm{g}_{\chi_{R}}(y_{1})$. So we only need to consider
$$\bm{g}'(y)=(2,2,2)y^2+(0,2,1)y$$
and
$$\bm{g}'_{\chi}(y)=(2,2,2)y^2+\bm{e}y,$$
where
$$\bm{e}\in E=\{(0,2,1), (0,1,2),(1,0,2),(1,2,0),(2,1,0), (2,0,1)\}.$$

 For any permutations $\chi_{L}$ and $\chi_{R}$ of $\{0,1,2\}$, we have
$$
\bm{g}'_{\chi_L}(y_0)
\begin{bmatrix}
0 & 0 & 0\\
0 & 1 & 2\\
0 & 2 & 1
\end{bmatrix}
\bm{g}'_{\chi_R}(y_1)^T
=\left((2,2,2)y_0^2+\bm{e}_1y_0\right)
\begin{bmatrix}
0 & 0 & 0\\
0 & 1 & 2\\
0 & 2 & 1
\end{bmatrix}
\left((2,2,2)^Ty_1^2+\bm{e}_2^Ty_1\right),
$$
where $\bm{e}_1, \bm{e}_2\in E$. Since
$$(2,2,2)
\begin{bmatrix}
0 & 0 & 0\\
0 & 1 & 2\\
0 & 2 & 1
\end{bmatrix}=(0,0,0)$$
and
$$\begin{bmatrix}
0 & 0 & 0\\
0 & 1 & 2\\
0 & 2 & 1
\end{bmatrix}
(2,2,2)^T=(0,0,0)^T,$$
we have

$$\bm{g}'_{\chi_L}(y_0)
\begin{bmatrix}
0 & 0 & 0\\
0 & 1 & 2\\
0 & 2 & 1
\end{bmatrix}
\bm{g}'_{\chi_R}(y_1)^T=\bm{e}_1
\begin{bmatrix}
0 & 0 & 0\\
0 & 1 & 2\\
0 & 2 & 1
\end{bmatrix}
\bm{e}_2^Ty_0y_1.
$$
By substituting the vectors $\bm{e}_1, \bm{e}_2\in E$ into the above formula, we obtain

$$S_Q(3, 3)=\{y_0y_1,2y_0y_1\}.$$

By applying Theorems \ref{thm-3} and \ref{construction},  we obtain the following construction.

\begin{construction}
Let $\pi$ be a permutation of $\{0, 1, \cdots, m-1\}$. For $N=q=3$, we have
\begin{itemize}
\item[(1)]
The ternary arrays extracted from the seed PU matrices can be expressed by
\begin{equation}\label{extrac-func-2}
f(y_0,y_1,\cdots, y_{m-1})=\sum_{k=1}^{m-1}d_k y_{k-1}y_{k}+\sum_{k=0}^{m-1}c_{k_2}y_{k}^2+\sum_{k=0}^{m-1}c_{k_1}y_{k}+c',
\end{equation}
for $d_k\in \Z^*_3 \ (1\leq k< m), c_{k_1}, c_{k_2}, c'\in \Z_3 \ (0\leq k< m).$
\item[(2)]
The corresponding function matrix $\widetilde{\bm{M}}(\bm{y})$  is given by
$$\widetilde{\bm{M}}(\bm{y})=f\cdot \bm{J}_3+
          \begin{bmatrix}
            0 & y_{m-1} & 2y_{m-1} \\
            y_{0} & y_0+y_{m-1} & y_0+2y_{m-1} \\
            2y_{0} & 2y_0+y_{m-1} & 2y_0+2y_{m-1} \\
          \end{bmatrix},
          $$
where the function $f$ is in the form (\ref{extrac-func-2}).  The sequences evaluated by the matrix $\pi\cdot \widetilde{\bm{M}}(\bm{y})$ form a CCC of size $3$.
\item[(3)] The collection of ternary sequences in
CSSs of size 3 derived from the seed PU matrices of order 3 can be expressed by
$$S(3,3)=\left\{\sum_{k=1}^{m-1}d_k y_{\pi(k-1)}y_{\pi(k)}+\sum_{k=0}^{m-1}c_{k_2}y_{k}^2+\sum_{k=0}^{m-1}c_{k_1}y_{k}+c' \right\}$$
for $\forall \pi, d_k\in \Z^*_3 \ (1\leq k< m), c_{k_1}, c_{k_2}, c'\in \Z_3 \ (0\leq k< m).$
\item[(3)] The sequences evaluated by
$$\left\{
\begin{aligned}
&f,\\
&f+y_{\pi(0)},\\
&f+2y_{\pi(0)}
\end{aligned}\right.
$$
form a ternary CSS of size 3 for $\forall f\in S(3,3)$.
\end{itemize}
\end{construction}
From the enumeration in Corollary \ref{seq-enum}, the set $S(3,3)$  determines exactly
$m!\cdot2^{m-2}\cdot3^{2m+1}$
distinct sequences for $m>1$. As far as our knowledge, the sequences in Construction 2 are never reported in the literature.

\subsection{Case $N=4$ $q=2$: Binary Sequences of Size 4}

Let $N=4,q=2$ in this subsection.

(1) Determine the set $S_{\widetilde{H}}(2, 4)$.
From Example \ref{exam-5}, there is only one phase matrix in $S_{\widetilde{H}}(2, 4)$:
$$
\begin{bmatrix}
0 & 0 &0&0\\
0 & 1&0&1\\
 0 &0&1&1\\
  0 &1&1& 0
\end{bmatrix}.
$$

(2) Computation of the set $S_L(2, 4)$.
From the proof in the Subsection 4.3, if we take the direct product of two $\Z_2$ to replace $\Z_4$, the results are still valid.
So any variable $y\in \Z_4$ can be represented by $(x_0, x_1)$, where $y=2x_1+x_0$ and $x_0, x_1\in \Z_2$.
From Example \ref{exam-7}, we have the basis functions
$$\left\{
\begin{aligned}
g_0(y)&=(1-x_0)(1-x_1),\\
g_1(y)&=x_0(1-x_1),\\
g_2(y)&=(1-x_0)x_1,\\
g_3(y)&=x_0x_1.
\end{aligned}
\right.
$$
Let $y_k=(x_{2k},x_{2k+1})$. Then the function $f(y_0,y_1,\cdots, y_{m-1})$ can be re-expressed by a GBF
$$f(x_0, x_1, \cdots, x_{2m-1}).$$ The linear terms are given by
$$S_L(2, 4)=\left\{\sum_{k=0}^{m-1}d_kx_{2k}x_{2k+1} +\sum_{k=0}^{2m-1}c_kx_{k}+c'\bigg| d_k, c_k, c'\in \Z_2, 0\leq k\leq 2m-1 \right\}.$$

(3) Computation of the set $S_Q(2, 4)$. We have
\begin{eqnarray*}
\bm{g}(y)&=&(g_0(y), g_1(y), g_2(y), g_3(y))\\
&=&(1,-1,-1,1)x_0x_1+(-1,1,0,0)x_0+(-1,0,1,0)x_1+(1,0,0,0)\cdot 1.
\end{eqnarray*}

It is obvious that $(1,0,0,0)\cdot 1$ only produces linear terms. For any permutation $\chi$ of $\{0,1,2,3\}$, we have
$${\chi}(1,-1,-1,1)
\begin{bmatrix}
0 & 0 &0&0\\
0 & 1&0&1\\
 0 &0&1&1\\
  0 &1&1& 0
\end{bmatrix}=(0,0,0,0)$$
and
$$\begin{bmatrix}
0 & 0 &0&0\\
0 & 1&0&1\\
 0 &0&1&1\\
  0 &1&1& 0
\end{bmatrix}
{\chi}(1,-1,-1,1)^T=(0,0,0,0)^T.$$

So we only need to consider
$$\bm{g}'_{\chi}(y)=\chi(-1,1,0,0)x_0+\chi(-1,0,1,0)x_1=\chi(-x_1-x_0,x_0,x_1,0).$$

By  computing
\begin{eqnarray*}
S&=&\bm{g}'_{\chi_L}(y_0)
\begin{bmatrix}
0 & 0 &0&0\\
0 & 1&0&1\\
 0 &0&1&1\\
  0 &1&1& 0
\end{bmatrix}
\bm{g}'_{\chi_R}(y_1)^T\\
&=&\chi_L(-x_1-x_0,x_0,x_1,0)
\begin{bmatrix}
0 & 0 &0&0\\
0 & 1&0&1\\
 0 &0&1&1\\
  0 &1&1& 0
\end{bmatrix}
\chi_R(-x_3-x_2,x_2,x_3,0)^T
\end{eqnarray*}
for all permutation $\chi_L, \chi_R$ of $\{0,1,2,3\}$, we obtain the set $$S_Q(2, 4)=\{\varphi_i((x_0, x_1), (x_2,x_3))\mid 1 \leq i\leq 6\},$$
where
$$\left\{
\begin{aligned}
&\varphi_1=x_0x_2+x_0x_3+x_1x_2,\\
&\varphi_2=x_0x_2+x_0x_3+x_1x_3,\\
&\varphi_3=x_0x_2+x_1x_3+x_1x_2,\\
&\varphi_4=x_1x_3+x_0x_3+x_1x_2,\\
&\varphi_5=x_0x_3+x_1x_2,\\
&\varphi_6=x_1x_3+x_0x_2.
\end{aligned}\right.
$$

By applying  Theorems \ref{thm-3} and \ref{construction},  we obtain the following construction.
\begin{construction}\label{construction-3}
Let $\pi$ be a permutation of $\{0, 1, \cdots, m-1\}$ and $\bm{x}=(x_0, x_1, \cdots x_{2m-1})$. For $N=4$ and $q=2$, we have
\begin{itemize}
\item[(1)]
The binary arrays extracted from the seed PU matrices can be expressed by

\begin{equation}\label{extrac-func-3}
f(\bm{x})=\sum_{k=1}^{m-1}h_k((x_{2k-2},x_{2k-1}), (x_{2k},x_{2k+1}) )+\sum_{k=0}^{m-1}d_kx_{2k}x_{2k+1} +\sum_{k=0}^{2m-1}c_kx_{k}+c',
\end{equation}
for  $h_k(\cdot,\cdot)\in S_Q(2, 4) (1\leq k\leq m-1)$ and  $d_k, c_k, c'\in \Z_2, 0\leq k\leq m-1$.

\item[(2)] The corresponding function matrix $\widetilde{\bm{M}}(\bm{x})$  is given by
$$\widetilde{\bm{M}}(\bm{x})=f\cdot \bm{J}_4+A(x_0, x_1)\cdot \bm{J}_4+\bm{J}_4 \cdot A(x_{2m-2}, x_{2m-1}),$$
where the GBF $f$ is in the form (\ref{extrac-func-3}) and  $A(x_0, x_1)=diag(0, x_0, x_1, x_0+x_1)$ is a diagonal matrix.
\item[(3)]
The collection of binary  sequences in
CSSs of size $4$ derived from the seed PU matrices of order $4$ can be expressed by
$$S(2,4)=\left\{\sum_{k=1}^{m-1}h_k((x_{2\cdot\pi(k-1)},x_{2\cdot\pi(k-1)+1}), (x_{2\cdot\pi(k)},x_{2\cdot\pi(k)+1}) )+\sum_{k=0}^{m-1}d_kx_{2k}x_{2k+1} +\sum_{k=0}^{2m-1}c_kx_{k}+c' \right\}$$
for  $h_k(\cdot,\cdot)\in S_Q(2, 4) (1\leq k\leq m-1)$ and  $d_k, c_k, c'\in \Z_2, 0\leq k\leq m-1$.
\item[(4)] The sequences evaluated by
$$\left\{
\begin{aligned}
&f,\\
&f+x_{2\cdot\pi(0)},\\
&f+x_{2\cdot\pi(0)+1},\\
&f+x_{2\cdot\pi(0)}+x_{2\cdot\pi(0)+1}
\end{aligned}\right.
$$
form a binary CSS of size 4 for $\forall f\in S(2,4)$.
\end{itemize}
\end{construction}

The sequences in Construction 3 are not shown in \cite{Paterson00,Schmidt07}, but are reported in \cite{Wu2016}. Furthermore, the first-order Reed-Muller code is a sub-code of $S_L(2, 4)$, which is also a linear code. Note that $|RM_2(1, 2m)|=2^{2m+1}$, while $|S_L(2, 4)|=2^{3m+1}$. The sequences shown here fill up $\frac{1}{2}m!\cdot 6^{m-1}$ distinct cosets of $S_L(2, 4)$ where the collection of the coset representatives are in the set $$\left\{\sum_{k=1}^{m-1}h_k((x_{2\cdot\pi(k-1)},x_{2\cdot\pi(k-1)+1}), (x_{2\cdot\pi(k)},x_{2\cdot\pi(k)+1}) )\bigg| h_k(\cdot,\cdot)\in S_Q(2, 4), 1\leq k\leq m-1 \right\}.$$ From the enumeration in Corollary \ref{seq-enum}, the set $S(2,4)$  determines exactly
$m!\cdot 6^{m-1}2^{3m}$
distinct sequences.

\begin{remark}\label{remark-2-6}
Note that the size of the arrays extracted here is $4\times 4\times \cdots \times 4$. So the number of the variables of GBFs in the form (\ref{extrac-func-3}) is $2m$, but $\pi$ in Construction 3 is a permutation of $\{0, 1, \cdots, m-1\}$. We will construct the arrays of size $2\times 2\times \cdots \times 2$ and extend the results in Construction 3 in next section.
\end{remark}

\subsection{Case $N=4$, $q=4$: Quaternary Sequences of Size 4}

Let $N=4,q=4$ in this subsection.

(1) Determine the set $S_{\widetilde{H}}(4, 4)$.
From Example \ref{exam-6}, there are two phase matrices in $S_{\widetilde{H}}(4, 4)$, which are
$$
\begin{bmatrix}
0 & 0 &0&0\\
0 & 2&0&2\\
 0 &0&2&2\\
  0 &2&2& 0
\end{bmatrix}
~and~ \begin{bmatrix}
0 & 0 &0&0\\
0 &1 &2&3\\
 0 &2&0&2\\
 0 &3&2& 1
\end{bmatrix}.
$$

(2) Computation of the set $S_L(4, 4)$.
The arguments are the same as the case $S_L(2, 4)$. Any variable $y\in \Z_4$ can be represented by $(x_0, x_1)$ for $x_0, x_1\in \Z_2$. From Example \ref{exam-7}, we have the basis functions
$$\left\{
\begin{aligned}
g_0(y)&=(1-x_0)(1-x_1),\\
g_1(y)&=x_0(1-x_1),\\
g_2(y)&=(1-x_0)x_1,\\
g_3(y)&=x_0x_1.
\end{aligned}
\right.
$$
Let $y_k=(x_{2k},x_{2k+1})$. We obtain the linear terms
$$S_L(4, 4)=\left\{\sum_{k=0}^{m-1}d_kx_{2k}x_{2k+1} +\sum_{k=0}^{2m-1}c_kx_{k}+c'\bigg| d_k, c_k, c'\in \Z_4, 0\leq k\leq 2m-1 \right\}.$$

(3) Computation of the set $S_Q(4, 4)$.
Note that the case for the first phase matrix has been studied in the previous subsection.
It is known $S_Q(2, 4)=\{\varphi_i\}.$ Then $\{2\varphi_i\}$, which must be derived from first phase matrix, is a subset of $S_Q(4, 4)$. So we only need to consider the second phase matrix. We have
\begin{eqnarray*}
\bm{g}(y)
=(1,-1,-1,1)x_0x_1+(-1,1,0,0)x_0+(-1,0,1,0)x_1+(1,0,0,0)\cdot 1.
\end{eqnarray*}
The term $(1,0,0,0)\cdot 1$ only produces linear terms. So we only consider
$$\bm{g}'(y)=(1,-1,-1,1)x_0x_1+(-1,0,1,0)x_1+(-1,1,0,0)x_0.$$
By  computing
\begin{eqnarray*}
\bm{g}'_{\chi_L}(y_0=(x_0, x_1))
\begin{bmatrix}
0 & 0 &0&0\\
0 & 1&2&3\\
 0 &2&0&2\\
  0 &3&2& 1
\end{bmatrix}
\bm{g}'_{\chi_R}(y_1=(x_2, x_3))^T
\end{eqnarray*}
for all permutation $\chi_L, \chi_R$ of $\{0,1,2,3\}$, and combining the functions in $S_Q(2, 4)$, we obtain the set
$$S_Q(4, 4)=\{\psi_i((x_0, x_1), (x_2,x_3)) \mid 1\leq i\leq 9, a_0\in\{0,1,2,3\}, a_1\in\{1,2,3\},a_2\in\{1,3\} \},$$
where
$$\left\{
\begin{aligned}
&\psi_1=a_0 x_1 x_3 + 2x_0 x_3 + 2 x_1x_2,\\
&\psi_2=a_0 x_1 x_2 + 2 x_0 x_2 + 2 x_1 x_3,\\
&\psi_3=a_1 x_0 x_3 + 2 x_0 x_2 + 2 x_1 x_3,\\
&\psi_4=a_1 x_0 x_2 + 2 x_1 x_2 + 2 x_0 x_3,\\
&\psi_5=a_2 x_1 x_3 + (a_2+2) x_0 x_3 + 2 x_1 x_2 + 2 x_0 x_2+ 2x_0 x_1 x_3,\\
&\psi_6=a_2 x_1 x_2 + (a_2+2) x_0 x_2 + 2 x_1 x_3 + 2 x_0 x_3+2x_0 x_1 x_2,\\
&\psi_7=a_2 x_1 x_3 + (a_2+2) x_1 x_2 + 2 x_0 x_3 + 2 x_0 x_2+2x_1 x_2 x_3,\\
&\psi_8=a_2 x_0 x_2 + (a_2+2) x_0 x_3 + 2 x_1 x_2 + 2 x_1 x_3+2x_0 x_2 x_3,\\
&\psi_9=a_2 x_0 x_2 +(a_2+2)  x_0 x_3 + (a_2+2) x_1 x_2 + a_2 x_1 x_3 + 2x_0 x_1 x_2+ 2 x_0 x_1 x_3 + 2 x_0 x_2x_3 + 2 x_1 x_2 x_3.
\end{aligned}\right.
$$

By applying  Theorems \ref{thm-3} and \ref{construction},  we obtain the following construction.
\begin{construction}\label{construction-4}
Let $\pi$ be a permutation of $\{0, 1, \cdots, m-1\}$ and $\bm{x}=(x_0, x_1, \cdots x_{2m-1})$. For $N=4$ and $q=4$, we have
\begin{itemize}
\item[(1)]
The quaternary arrays extracted from the seed PU matrices can be expressed by
\begin{equation}\label{extrac-func-4}
f(\bm{x})=\sum_{k=1}^{m-1}h_k((x_{2k-2},x_{2k-1}), (x_{2k},x_{2k+1}) )+\sum_{k=0}^{m-1}d_kx_{2k}x_{2k+1} +\sum_{k=0}^{2m-1}c_kx_{k}+c'
\end{equation}
for  $h_k(\cdot,\cdot)\in S_Q(4, 4) (1\leq k\leq m-1)$ and  $d_k, c_k, c'\in \Z_4, 0\leq k\leq m-1$.
\item[(2)]
The collection of quaternary sequences in
CSSs of  size $4$ derived from the seed PU matrices of order $4$ is given by
$$S(4,4)=\left\{\sum_{k=1}^{m-1}h_k((x_{2\cdot\pi(k-1)},x_{2\cdot\pi(k-1)+1}), (x_{2\cdot\pi(k)},x_{2\cdot\pi(k)+1}) )+\sum_{k=0}^{m-1}d_kx_{2k}x_{2k+1} +\sum_{k=0}^{2m-1}c_kx_{k}+c' \right\}$$
for  $h_k(\cdot,\cdot)\in S_Q(4, 4) (1\leq k\leq m-1)$ and  $d_k, c_k, c'\in \Z_4, 0\leq k\leq m-1$.
\end{itemize}
\end{construction}

For the sequences in Construction 4, if $h_k(\cdot,\cdot)$ are all chosen from the subset $\{\psi_i\mid 1\leq i\leq 4 \}$ in $S_Q(4, 4)$ for  $1\leq k\leq m-1$, these sequences have been introduced in \cite{Wu2016}.

\begin{example}\label{exam-8}
Let $m=2$, and $h_1(\cdot,\cdot)$ is chosen as $\psi_5$ with $a_2=1$, then we have
$$f_1=x_1 x_3 + 3x_0 x_3 + 2 x_1 x_2 + 2 x_0 x_2+ 2x_0 x_1 x_3.$$
Let $m=2$, and $h_1(\cdot,\cdot)$ is chosen as $\psi_9$ with $a_2=1$, we have
$$f_2=x_0 x_2 +3x_0 x_3 + 3x_1 x_2 + x_1 x_3 + 2x_0 x_1 x_2+ 2 x_0 x_1 x_3 + 2 x_0 x_2x_3 + 2 x_1 x_2 x_3.$$
Each sequence evaluated by the above GBFs lie in a CSS of size 4 respectively.
\end{example}

 From Example \ref{exam-8}, even for the sequences of length $2^4$ evaluated by the functions in $\{\psi_i\mid 5\leq i\leq 9 \}$, they are never reported in the literature, such as \cite{Chen06,Paterson00,Schmidt07,Stinchcombe,Wu2016}. Moreover, for sequence $f\in S(4,4)$, if there exists $k$, such that $h_k(\cdot,\cdot)$ is chosen from the subset $\{\psi_i\mid 5\leq i\leq 9 \}$, sequence $f$ lying in CSSs of size 4 must be new, which asserts that the number of the quaternary  sequences in CSSs of size $4$ can be  exponentially increased.

The first-order GRM code is a sub-code of $S_L(4, 4)$ which is also a linear code. Note that $|RM_4(1, 2m)|=4^{2m+1}$, while $|S_L(4, 4)|=4^{3m+1}$. The sequences constructed here fill up $\frac{1}{2} m!\cdot 24^{m-1}$ distinct cosets of $S_L(4, 4)$ where the collection of the coset representatives are shown in the set $$\left\{\sum_{k=1}^{m-1}h_k((x_{2\pi(k-1)},x_{2\pi(k-1)+1}), (x_{2\pi(k)},x_{2\pi(k)+1}) )| h_k(\cdot,\cdot)\in S_Q(4, 4), 1\leq k\leq m-1 \right\}.$$
From the enumeration in Corollary \ref{seq-enum}, the set $S(4,4)$  determines exactly
$\frac{1}{2} m!\cdot 24^{m-1}\cdot 4^{3m+1}$
distinct sequences.

For the known complementary sequences with explicit Boolean functions, quaternary sequences are always generalized from the binary case. Here the sequences in  $S(4, 4)$ are generalized from the binary case if and only if all $h_k(\cdot,\cdot)$ are chosen from the  subset $\{\psi_i\mid 1\leq i\leq 4 \}$ in $S_Q(4, 4)$.  One important reason is that the  coset representatives are computed by 2 BH matrices in $S_Q(4, 4)$ which lead to that the number of the coset representatives in $S(4, 4)$ is  $\frac{1}{2}m!\cdot 24^{m-1}$. On the other hand, for the case that the coset representatives are computed by only one phase matrix in $S_Q(2, 4)$, the number of the coset representatives in $S(2, 4)$ is $\frac{1}{2} m!\cdot  6^{m-1}$.

\begin{remark}\label{remark-2-7}
 Similar to Remark \ref{remark-2-6}, we will construct the arrays of size $2\times 2\times \cdots \times 2$ and extend the results in Construction 4 in next section.
\end{remark}

\section{Generalized Seed PU Matrices and Corresponding Sequences}

It is shown in Constructions \ref{construction-3} and \ref{construction-4} in Section 5 that the binary and quaternary  sequences in
CSSs of  size 4 are obtained by permutations acting on the arrays of size $4\times 4\times \cdots \times 4$.  Although these sequences are all evaluated by the GBFs with Boolean variables $(x_0,x_1,\cdots, x_{2m-1})$, we can only permute the variables $(y_0,y_1,\cdots, y_{m-1})$ for $y_k\in \Z_4$.
On the other hand, the sequences proposed in \cite{DavisJedwab99,Paterson00,Schmidt07} are also evaluated by the  GBFs. Permuting the Boolean variables of the complementary sequences  in \cite{DavisJedwab99,Paterson00,Schmidt07} are still complementary sequences of the same CSS size. It is rational to ask whether the GBFs extracted from the seed PU matrices of size $N=2^n$ can represent the arrays of size $2\times 2 \times \cdots \times 2$.

We set $p=2$ and $N=2^n$ in this section, and construct PU matrices of size $2^n$ and arrays of size $2\times 2 \times \cdots \times 2$. Then we can apply Theorem \ref{thm-1}-(5) in Section 3 to obtain more new CCCs and CSSs.

\subsection{Generalized Seed PU Matrices}

For $N=2^n$, we generalize the delay matrix in Section 4 to a Multivariate diagonal polynomial matrix by the Kronecker product of $\bm{D}(z)$, where $\bm{D}(z)$ is the univariate  delay matrix of order $2$ in Subsection 4.1,  i.e.,
$$\bm{D}(z)=\begin{bmatrix}1&0\\
0& z \end{bmatrix}.$$

\begin{definition}\label{G-delay-M}
A generalized delay matrix $\bm{D}(\bm{z})$ with multi-variables $\bm{z}=(z_0, z_{1}, \cdots,  z_{n-1})$
can be represented by the Kronecker product of $\bm{D}(z_v)$ for $0\leq v<n$, i.e.,
\begin{equation}\label{eq-of-delay}
\bm{D}(\bm{z})=\bm{D}(z_{n-1}) \otimes\cdots \otimes \bm{D}(z_{1})\otimes \bm{D}(z_{0}).
\end{equation}
\end{definition}

It is obvious that $\bm{D}(\bm{z})$ is also a  diagonal matrix, so we can write $\bm{D}(\bm{z})$ by
$$\bm{D}(\bm{z})=diag(\phi_0(\bm{z}), \phi_1(\bm{z}), \cdots, \phi_{N-1}(\bm{z})),$$
where $\phi_y(\bm{z})$ $(0\leq y<N)$ is a function of $\bm{z}$.
\begin{example}
Let $N=2^2$ and $\bm{z}=(z_0, z_{1})$. The generalized delay matrix $\bm{D}(z_0,z_1)$ can be written in the form of
\begin{eqnarray*}
\bm{D}(z_0,z_1)=\begin{bmatrix}1&0\\
0& z_1 \end{bmatrix}\otimes\begin{bmatrix}1&0\\
0& z_0 \end{bmatrix}=\begin{bmatrix}
1&0&0&0\\
0&z_{0}&0&0\\
0&0&z_{1}&0\\
0&0&0&z_{1}z_{0}
\end{bmatrix}=\begin{bmatrix}
z_1^0\cdot z_0^0&0&0&0\\
0&z_{1}^0\cdot z_0^1&0&0\\
0&0&z_1^1\cdot z_{0}^0&0\\
0&0&0&z_{1}^1\cdot z_{0}^1
\end{bmatrix}.
\end{eqnarray*}
\end{example}

From the above example, $\phi_y(\bm{z})$ in the main diagonal can be expressed by $z_1^{x_1}\cdot z_0^{x_0}$, where $(x_0,x_1)$ is the binary expansion of integer $y$. In general, the generalized delay matrix $\bm{D}(\bm{z})$ can be explicitly represented by the following lemma.

\begin{lemma}\label{lemma-8}
Let $(x_0, x_1,\cdots, x_{n-1})$ be the binary expansion of integer $y$,  i.e., $y=\sum_{v=0}^{n-1}x_v\cdot 2^v$. Then the generalized delay matrix $\bm{D}(\bm{z})$ in Definition \ref{G-delay-M} can be expressed as
\begin{eqnarray*}
\bm{D}(\bm{z})=\begin{bmatrix}
\phi_0(\bm{z})&0&\cdots&0\\
0&\phi_1(\bm{z})&\cdots&0\\
\vdots&\vdots&\ddots&\vdots\\
0&0&\cdots&\phi_{N-1}(\bm{z})
\end{bmatrix},
\end{eqnarray*}
where
\begin{equation}\label{expansion}
\phi_y(\bm{z})=\prod_{v=0}^{n-1}z_{v}^{x_v}.
\end{equation}
\end{lemma}

We omit the proof of the above lemma, since it can be easily done by mathematical induction.

Let $\bm{H}^{\{k\}}$ be an arbitrary BH matrix chosen from $H(q,N)$ for $0\leq k\leq m$ and $\bm{D}(\bm{z}_k)$ be the generalized delay matrix from Definition \ref{G-delay-M} where $\bm{z}_k=(z_{kn}, z_{kn+1}, \cdots,  z_{kn+n-1})$ for $0\leq k<m$. And let $\bm{z}=(\bm{z}_{0}, \bm{z}_{1}, \cdots,  \bm{z}_{m-1})=(z_{0}, z_{1}, \cdots,  z_{nm-1})$. We define a multivariate polynomial matrix $\bm{M}(\bm{z})$ as follow.
\begin{eqnarray}\label{seed-PU-2}
\bm{M}(\bm{z})&=&\bm{H}^{\{0\}}\cdot \bm{D}(\bm{z}_0)\cdot
\bm{H}^{\{1\}}\cdot\bm{D}(\bm{z}_1)\cdots \bm{H}^{\{m-1\}}\cdot
\bm{D}(\bm{z}_{m-1})\cdot \bm{H}^{\{m\}}\\
&=&\left(\prod_{k=0}^{m-1}\left(\bm{H}^{\{k\}}\cdot \bm{D}(\bm{z}_k)\right)\right)\cdot \bm{H}^{\{m\}}\nonumber.
\end{eqnarray}

\begin{remark}
The matrix $\bm{M}(\bm{z})$ defined here is similar to the seed PU matrix defined in Subsection 4.1, but the delay matrices are chosen differently.   Let $\pi$ be a permutation of $\{0, 1, \cdots, mn-1\}$ and $z_t=Z^{2^{\pi(t)}}$ in $\bm{M}(\bm{z})$ for $0\leq t< mn$,  we obtain a  matrix $\bm{M}(Z)$, which is identical to the single-variable PU matrix proposed in \cite{SPL2019}.
\end{remark}

\begin{theorem}\label{thm-6}
$\bm{M}(\bm{z})$ defined in formula (\ref{seed-PU-2}), called a generalized seed PU matrix,  is a desired PU matrix, i.e.,
\begin{itemize}
	\item[(1)] $\bm{M}(\bm{z})\cdot\bm{M}^{\dagger}(\bm{z}^{-1})=2^{(m+1)n}\cdot \bm{I}_N$;
	\item[(2)] Each entry of $\bm{M}(\bm{z})$ can be expressed by the generating function of an array $f(\bm{x}): \Z_2^{mn} \rightarrow \Z_q$.
\end{itemize}
\end{theorem}

\begin{proof}
It is obvious that $\bm{M}(\bm{z})\cdot\bm{M}^{\dagger}(\bm{z}^{-1})=2^{(m+1)n}\cdot \bm{I}_N$.
We expand $\bm{M}(\bm{z})$ in the following form.
\begin{equation*}
\bm{M}(\bm{z})=\sum_{0\leq x_0, x_1, \cdots, x_{mn-1}\leq 1}\left({\bm{{M}}}(x_0, x_1, \cdots, x_{mn-1})\cdot z_0^{x_0}z_1^{x_1}\cdots z_{mn-1}^{x_{mn-1}}\right),
\end{equation*}
where matrix $\bm{{M}}(x_0, x_1, \cdots, x_{mn-1})$ is the coefficient matrix of matrix $\bm{M}(\bm{z})$ of term $z_0^{x_0}z_1^{x_1}\cdots z_{mn-1}^{x_{mn-1}}$.

Based on formula (\ref{expansion}), the generalized delay matrix $\bm{D}(\bm{z}_k)$ can be represented by
$$ \bm{D}(\bm{z}_k)=\sum_{y_k=0}^{N-1} \left(\bm{E}_{y_k}\cdot\phi_{y_k}(\bm{z}_k)\right),$$
where $y_k=\sum_{v=0}^{n-1}x_{kn+v}\cdot 2^v.$

On the other hand,
\begin{eqnarray*}
\bm{M}(\bm{z})
&=&\left(\prod_{k=0}^{m-1}\left(\bm{H}^{\{k\}}\cdot \bm{D}(\bm{z}_k)\right)\right)\cdot \bm{H}^{\{m\}}\\
&=&\left(\prod_{k=0}^{m-1}\left(\bm{H}^{\{k\}}\cdot \sum_{y_k=0}^{N-1} \left(\bm{E}_{y_k}\cdot\phi_{y_k}(\bm{z}_k)\right)\right)\right)\cdot \bm{H}^{\{m\}}\\
&=&\sum_{y_0=0}^{N-1}\cdots\sum_{y_{m-1}=0}^{N-1}\left( \left(\prod_{k=0}^{m-1} \left(\bm{H}^{\{k\}}\cdot \bm{E}_{y_k}\cdot \right)\cdot\bm{H}^{\{m\}}\right)\cdot \prod_{k=0}^{m-1}\phi_{y_k}(\bm{z}_k)\right) \\
&=&\sum_{0\leq x_0, x_1, \cdots, x_{mn-1}\leq 1}\left(\prod_{k=0}^{m-1} \left(\bm{H}^{\{k\}}\cdot \bm{E}_{y_k}\cdot \right)\cdot\bm{H}^{\{m\}}\right)\cdot z_0^{x_0}z_1^{x_1}\cdots z_{mn-1}^{x_{mn-1}}.
\end{eqnarray*}
Thus we have
\begin{equation*}
\bm{{M}}(x_0, x_1, \cdots, x_{mn-1})=\left(\prod^{m-1}_{k=0}\left(\bm{H}^{\{k\}}\cdot
\bm{E}_{y_k}\right)\right)\cdot \bm{H}^{\{m\}},
\end{equation*}
and each entry
\begin{equation}\label{abc}
{{M}}_{i,j}(x_0, x_1, \cdots, x_{mn-1})={H}^{\{0\}}_{i,y_0}\cdot\left(\prod^{m-1}_{k=1}{H}^{\{k\}}_{y_{k-1},y_k}
   \right)\cdot {H}^{\{m\}}_{y_{m-1},j}.
\end{equation}
Therefore, for given $i$ and $j$, there exists GBF $f_{i,j}$ from $\Z_2^{mn}$ to $\Z_q$ such that
$$
{M}_{i,j}(\bm{z})=\sum_{0\leq x_0, x_1, \cdots, x_{mn-1}\leq 1} \left(\omega^{f_{i,j}(x_0, x_1, \cdots, x_{mn-1})}
\cdot z_0^{x_0}z_1^{x_1}\cdots z_{mn-1}^{x_{mn-1}}\right)
$$
which complete the proof.
\end{proof}

Let both the order of the seed PU matrices in Subsection 4.1 and the generalized seed PU matrices in this subsection be $N=2^n$. Then we have
$$\omega^{f_{i,j}(y_0, y_1, \cdots, y_{m-1})}={H}^{\{0\}}_{i,y_0}\cdot\left(\prod^{m-1}_{k=1}{H}^{\{k\}}_{y_{k-1},y_k}
\right)\cdot {H}^{\{m\}}_{y_{m-1},j}$$
from formula (\ref{PU-function}) and
$$\omega^{f_{i,j}(x_0, x_1, \cdots, x_{mn-1})}={H}^{\{0\}}_{i,y_0}\cdot\left(\prod^{m-1}_{k=1}{H}^{\{k\}}_{y_{k-1},y_k}
\right)\cdot {H}^{\{m\}}_{y_{m-1},j}$$
from formula  (\ref{abc}), respectively, where $y_k=\sum_{v=0}^{n-1}x_{kn+v}\cdot 2^v$. Consequently we have $$f_{i,j}(x_0, x_1, \cdots, x_{mn-1})=f_{i,j}(y_0, y_1, \cdots, y_{m-1}).$$
Therefore, if we replace $y_k$ by $\bm{x}_k=(x_{kn}, x_{kn+1}, \cdots,  x_{kn+n-1})$ and denote $\bm{x}=(\bm{x}_{0}, \bm{x}_{1}, \cdots,  \bm{x}_{m-1})=(x_{0}, x_{1}, \cdots,  x_{nm-1})$, then
the GBFs extracted from the generalized seed PU matrices are identical to the functions  extracted from the seed PU matrices.
However, functions $f_{i,j}(y_0, y_1, \cdots, y_{m-1})$ extracted from the seed PU matrices are the arrays of size $2^n\times 2^n \times \cdots \times 2^n$ and dimension $m$, while the GBFs $f_{i,j}(x_0, x_1, \cdots, x_{mn-1})$ extracted from the generalized seed PU matrices are arrays of size $2\times 2 \times \cdots \times 2$ and dimension $mn$.

\subsection{Constructions from the Generalized Seed PU Matrices}\label{sec6.2}

We now extend the results in Theorem \ref{construction} for $N=2^n$.
\begin{theorem}\label{construction-2}
For $N=2^n$, suppose that   $\{ f_{i,j}(y_0,y_1,\cdots, y_{m-1})\}$ and $\{f_i(y_0,y_1,\cdots, y_{m-1})\}$ respectively form a CAS and a CCC of size $N$, where
both $f_{i,j}$ and $f_i$ are given in Theorem \ref{construction}. By replacing the functions $f_{i,j}(y_0,y_1,\cdots, y_{m-1})$ and $f_i(y_0,y_1,\cdots, y_{m-1})$ by GBFs $f_{i,j}(x_0,x_1,\cdots, x_{mn-1})$ and  $f_i(x_0,x_1,\cdots, x_{mn-1})$ respectively, for arbitrary permutation $\pi$ of symbols $\{0, 1, \cdots, mn-1\}$, we have
\begin{itemize}
\item[(1)] The  sequences evaluated by the following GBFs form a CCC of size $N$:
$$\pi\cdot f_{i,j}=f_{i,j}(x_{\pi(0)},x_{\pi(1)},\cdots, x_{\pi(mn-1)}), 0\leq i, j< N.$$
\item[(2)]The sequences evaluated by the following GBFs form a CSS of size $N$ and length $N^m$:
$$\pi\cdot f_{i}=f_i(x_{\pi(0)},x_{\pi(1)},\cdots, x_{\pi(mn-1)}),\  0\leq i< N.$$
\end{itemize}
\end{theorem}

\begin{remark}
Without taking the duplication into consideration, the number of the sequences proposed in
Theorem \ref{construction-2} is  about $\frac{(mn)!}{m!}$ times that in Theorem \ref{construction}. Moreover,
 the results in this section also give a method to show GBF form of the sequences proposed in \cite{SPL2019}.
 The comparisons of the number of permutations in Theorems \ref{construction} and \ref{construction-2} for some special cases, can be found in  \cite[Table 1]{SPL2019}.
\end{remark}

We now extend the results in Constructions 3 and 4 in Section 5, which are given by following Constructions 5 and 6.

\begin{construction}\label{construction-5}
Let $\pi$ be a permutation of $\{0, 1, \cdots, 2m-1\}$ and $\bm{x}=(x_0, x_1, \cdots x_{2m-1})$. For $N=4$ and $q=2$, the GBFs extracted from the generalized seed PU matrices are the same as (\ref{extrac-func-3}).
\begin{itemize}
\item[(1)]
The corresponding function matrix $\widetilde{\bm{M}}(\bm{x})$  is given by
$$\widetilde{\bm{M}}(\bm{x})=f\cdot \bm{J}_4+A(x_0, x_1)\cdot \bm{J}_4+\bm{J}_4 \cdot A(x_{2m-2}, x_{2m-1}),$$
where the GBF $f$ is in the form (\ref{extrac-func-3}) and  $A(x_0, x_1)=diag\{0, x_0, x_1, x_0+x_1\}$ is a diagonal matrix. The sequences evaluated by the matrix $\pi\cdot \widetilde{\bm{M}}(\bm{x})$ form a CCC of size $4$.
\item[(3)]
The collection of binary sequences in CSSs of size $4$ derived from the generalized seed PU matrices can be expressed by
$$S'(2,4)=\left\{\sum_{k=1}^{m-1}h_k(x_{\pi(2 k-2)},x_{\pi(2 k-1)},x_{\pi(2 k)},x_{\pi(2 k+1)} )+\sum_{k=0}^{m-1}d_kx_{\pi(2k)}x_{\pi(2k+1)} +\sum_{k=0}^{2m-1}c_kx_{k}+c' \right\}$$
for  $h_k(\cdot,\cdot,\cdot,\cdot)\in S_Q(2, 4) (1\leq k\leq m-1)$ and  $d_k, c_k, c'\in \Z_2, 0\leq k\leq m-1$.
\item[(4)] The sequences evaluated by
$$\left\{
\begin{aligned}
&f,\\
&f+x_{\pi(0)},\\
&f+x_{\pi(1)},\\
&f+x_{\pi(0)}+x_{\pi(1)}
\end{aligned}\right.
$$
form a binary CSS of size 4 for $\forall f\in S'(2,4)$.
\end{itemize}
\end{construction}

\begin{construction}\label{construction-6}
Let $\pi$ be a permutation of $\{0, 1, \cdots, 2m-1\}$ and $\bm{x}=(x_0, x_1, \cdots x_{2m-1})$. For $N=4$ and $q=4$,  the GBFs extracted from the generalized seed PU matrices are the same as (\ref{extrac-func-4}).
\begin{itemize}
\item[(1)]
The corresponding function matrix $\widetilde{\bm{M}}(\bm{x})$  is given by
$$\widetilde{\bm{M}}(\bm{x})=f\cdot \bm{J}_4+A(x_0, x_1)\cdot \bm{J}_4+\bm{J}_4 \cdot B(x_{2m-2}, x_{2m-1}),$$
where the GBF $f$ is in the form (\ref{extrac-func-4}) and the diagonal matrices $A(x_0, x_1), B(x_0,x_1)$ are arbitrarily chosen from the set $\{diag(0, 2x_0, 2x_1, 2x_0+2x_1),
diag(0, 2x_0+x_1, 2x_1, 2x_0+3x_1),
diag(0, x_0+3x_1+2x_0x_1, 2x_0+2x_1, 3x_0+x_1+2x_0x_1).\}$ The sequences evaluated by the matrix $\pi\cdot \widetilde{\bm{M}}(\bm{x})$ form a CCC of size $4$.
\item[(2)]
The collection of quaternary sequences in CSSs of size $4$  derived from the generalized seed PU matrices can be represented by
$$S'(4,4)=\left\{\sum_{k=1}^{m-1}h_k(x_{\pi(2 k-2)},x_{\pi(2 k-1)},x_{\pi(2 k)},x_{\pi(2 k+1)} )+\sum_{k=0}^{m-1}d_kx_{\pi(2k)}x_{\pi(2k+1)} +\sum_{k=0}^{2m-1}c_kx_{k}+c' \right\}$$
for  $h_k(\cdot,\cdot,\cdot,\cdot)\in S_Q(4, 4) (1\leq k\leq m-1)$ and  $d_k, c_k, c'\in \Z_4, 0\leq k\leq m-1$.

\item[(3)]
The sequences evaluated by
$$\left\{
\begin{aligned}
&f,\\
&f+2x_{\pi(0)},\\
&f+2x_{\pi(1)},\\
&f+2x_{\pi(0)}+2x_{\pi(1)},
\end{aligned}\right.
\left\{
\begin{aligned}
&f,\\
&f+2x_{\pi(0)}+x_{\pi(1)},\\
&f+2x_{\pi(1)},\\
&f+2x_{\pi(0)}+3x_{\pi(1)},
\end{aligned}\right.
\mbox{or}\
\left\{
\begin{aligned}
&f,\\
&f+3x_{\pi(0)}+x_{\pi(1)}+2x_{\pi(0)}x_{\pi(1)},\\
&f+2x_{\pi(0)}+2x_{\pi(1)},\\
&f+x_{\pi(0)}+3x_{\pi(1)}+2x_{\pi(0)}x_{\pi(1)}
\end{aligned}\right.
$$
form a complementary set of size $4$ for $\forall f\in S'(4,4)$.
\end{itemize}
\end{construction}

\section{A Recursive Construction of CSSs and CCCs}

For both binary and quaternary sequences in CSSs of size 4 in  Constructions \ref{construction-5} and \ref{construction-6}, the length of sequences must be $4^m$. However, the length of the sequences can be $2^m$ in \cite{Paterson00,Schmidt07} for arbitrary $m$. Moreover, it is obvious that the constructions of sequences in \cite{Paterson00,Schmidt07} cannot be directly obtained by the proposed (generalized) seed PU matrices in Section 6. The motivation of this section and next section is to find a generalized framework for the construction of the desired PU matrices from the generalized seed PU matrices and extract the corresponding GBFs, from which we hope it can explain the known constructions of CSSs, CCCs and beyond.

\subsection{Basic Lemma on the Generating Matrices and Corresponding Function Matrices}

In this section, we present some basic results on the generating matrices and their corresponding function matrices.

Define multivariate variables
\begin{equation}\label{z1z2}
\left\{
\begin{aligned}
&\bm{z_0}=(z_{0}, z_{1},\dots, z_{n-1}),\\
&\bm{z_1}=(z_{n}, z_{n+1},\dots, z_{n+m_1-1}),\\
&\bm{z_2}=(z_{n+m_{1}}, z_{n+m_{1}+1},\dots, z_{n+m_{1}+m_{2}-1}),\\
&(\bm{z_0}, \bm{z_1}, \bm{z_2})=(z_{0}, z_{1},\dots, z_{n+m_{1}+m_{2}-1})
\end{aligned}\right.
\end{equation}
and Boolean variables
\begin{equation}\label{x1x2}
\left\{
\begin{aligned}
&\bm{x_0}=(x_{0}, x_{1},\dots, x_{n-1})\in \Z_2^{n},\\
&\bm{x_1}=(x_{n}, x_{n+1},\dots, x_{n+m_1-1})\in \Z_2^{m_1},\\
&\bm{x_2}=(x_{n+m_{1}}, x_{n+m_{1}+1},\dots, x_{n+m_{1}+m_{2}-1})\in \Z_2^{m_{2}},\\
&(\bm{x_0}, \bm{x_1}, \bm{x_2})=(x_{0}, x_{1},\dots, x_{n+m_{1}+m_{2}-1})\in \Z_2^{n+m_{1}+m_{2}}.
\end{aligned}\right.
\end{equation}

Let $\bm{A}({\bm{z}_1})$ and $\bm{B}({\bm{z}_2})$  be the generating matrices of the GBF matrices $\widetilde{\bm{A}}(\bm{x}_{1})=\{a_{i,j}(\bm{x}_{1})\}$ and $\widetilde{\bm{B}}(\bm{x}_{2})=\{b_{i,j}(\bm{x}_{2})\}$ of size $N$ respectively, where $a_{i,j}(\bm{x}_{1})$ and  $b_{i,j}(\bm{x}_{2})$ are GBFs from $\Z_{2}^{m_1}$ to $\Z_q$ and $\Z_{2}^{m_2}$ to $\Z_q$ respectively.

Recall the generalized delay matrix $\bm{D}(\bm{z}_0)$ of order $N=2^n$ in Definition \ref{G-delay-M} and the basis function $g_{i}:  \mathbb{Z}_N \rightarrow \Z_q$ shown in Example \ref{exam-7}.

\begin{lemma}\label{lem-9}
Let $\bm{C}(\bm{z}_0,\bm{z}_1,\bm{z}_2)$ be a multivariate polynomial matrix defined by
\begin{equation*}
 \bm{C}(\bm{z}_0,\bm{z}_1,\bm{z}_2)=\bm{A}(\bm{z}_{1})\cdot\bm{D}(\bm{z}_{0})\cdot\bm{B}(\bm{z}_{2}).
\end{equation*}
Then $\bm{C}(\bm{z}_0,\bm{z}_1,\bm{z}_2)$ is the generating matrix of the GBF matrix $\widetilde{\bm{C}}(\bm{x}_0,\bm{x}_1,\bm{x}_2)$ where each entry is a GBF from
$ \mathbb{Z}_2^{n+m_1+m_2} $ to $\Z_q$ with the expression
\begin{equation*}
c_{r,s}(\bm{x}_0,\bm{x}_1,\bm{x}_2)
=\sum_{i=0}^{N-1}\left( a_{r,i}(\bm{x}_{1})+b_{i,s}(\bm{x}_{2})\right){g}_{i}(\bm{x}_{0})
\end{equation*}	
for $0\leq r, s<N$. Moreover, if both $\bm{A}(\bm{z}_{1})$ and $\bm{B}(\bm{z}_{2})$ are  desired PU matrices, then $\bm{C}(\bm{z}_0,\bm{z}_1,\bm{z}_2)$ is a desired PU matrix.
\end{lemma}

\begin{proof}
According to formula (\ref{gene-matrix-dec}),  the generating matrices $\bm{A}(\bm{z}_1)$ and $\bm{B}(\bm{z}_2)$ can be expanded in the  form:
\begin{equation*}
\bm{A}(\bm{z}_1)=\sum_{\bm{x}_{1}}\bm{A}(\bm{x}_{1})\cdot \bm{z}_1^{\bm{x}_1}
\end{equation*}
and
\begin{equation*}
\bm{B}(\bm{z}_2)=\sum_{\bm{x}_{2}}\bm{B}(\bm{x}_{2})\cdot \bm{z}_2^{\bm{x}_2},
\end{equation*}
where $A_{i,j}(\bm{x}_{1})=\omega^{a_{i,j}(\bm{x}_1)}$ and $B_{i,j}(\bm{x}_{2})=\omega^{b_{i,j}(\bm{x}_2)}$, respectively, for  $0\leq i, j<N$.
The generalized delay matrix $\bm{D}({z}_0)$ can be expressed in the form
$$ \bm{D}(\bm{z}_0)=\sum_{i=0}^{N-1} \bm{E}_{i}\cdot\phi_{i}(\bm{z}_0)=\sum_{i=0}^{N-1} \bm{E}_{i}\cdot \bm{z}_0^{\bm{x}_0}$$
where $i=\sum_{v=0}^{n-1}x_v\cdot 2^v$.
Then we have
\begin{eqnarray*}
\bm{C}(\bm{z}_0,\bm{z}_1,\bm{z}_2)&=&\bm{A}(\bm{z}_{1})\cdot\bm{D}(\bm{z}_{0})\cdot\bm{B}(\bm{z}_{2})\\
&=&\left(\sum_{\bm{x}_{1}}\bm{A}(\bm{x}_{1})\cdot \bm{z}_1^{\bm{x}_1}\right) \left(\sum_{i=0}^{N-1}\bm{E}_{i}\cdot \bm{z}_0^{\bm{x}_0}\right) \left(\sum_{\bm{x}_{2}}\bm{B}(\bm{x}_{2})\cdot \bm{z}_2^{\bm{x}_2}\right)\\
&=&\sum_{\bm{x}_{0},\bm{x}_{1}, \bm{x}_{2}}\bm{A}(\bm{x}_{1})\cdot\bm{E}_{i}\cdot\bm{B}(\bm{x}_{2})\cdot \bm{z}_0^{\bm{x}_0}\bm{z}_1^{\bm{x}_1}\bm{z}_2^{\bm{x}_2}.
\end{eqnarray*}

Suppose that $\bm{C}(\bm{x}_0,\bm{x}_1,\bm{x}_2)=\bm{A}(\bm{x}_{1})\cdot\bm{E}_{i}\cdot\bm{B}(\bm{x}_{2})$. Then the entries of $\bm{C}(\bm{x}_0, \bm{x}_1, \bm{x}_2)$ can be expressed by
$$C_{r,s}(\bm{x}_0,\bm{x}_1,\bm{x}_2)=A_{r,i}(\bm{x}_{1})B_{i,s}(\bm{x}_2)=\omega^{a_{r,i}(\bm{x}_{1})+b_{i,s}(\bm{x}_2)}.$$
From  the definition of the basis function $g_i$, we have
$$a_{r,i}(\bm{x}_{1})+b_{i,s}(\bm{x}_2)=\sum_{i=0}^{{N-1}} a_{r,i}(\bm{x}_{1}){g}_{i}(\bm{x}_{0})+\sum_{i=0}^{{N-1}}b_{i,s}(\bm{x}_{2}){g}_{i}(\bm{x}_{0}),$$
which is a GBF (array) from
$ \mathbb{Z}_2^{n+m_1+m_2} $ to $\Z_q$,
 denoted by $c_{r,s}(\bm{x_0}, \bm{x_1}, \bm{x_2})$. Let $\widetilde{\bm{C}}(\bm{x_0}, \bm{x_1}, \bm{x_2})$ be the function matrix with entry $c_{r,s}(\bm{x_0}, \bm{x_1}, \bm{x_2})$ at position $(r, s)$. Then $\bm{C}(\bm{z_0}, \bm{z_1}, \bm{z_2})$ is the generating matrix of $\widetilde{\bm{C}}(\bm{x_0}, \bm{x_1}, \bm{x_2})$, which completes the proof.
\end{proof}
\begin{example}
Here comes a simple example  for $N=q=2$, assume that $\bm{A}(z_1)=\begin{bmatrix}1+z_1&-1+z_1\\1-z_1&-1-z_1\end{bmatrix}$ and $\bm{B}(z_2)=\begin{bmatrix}1+z_2&1-z_2\\1-z_2&1+z_2\end{bmatrix}$. Then their GBF matrix are $\widetilde{\bm{A}}(x_1)=\begin{bmatrix}0&1+x_1\\x_1&1\end{bmatrix}$ and $\widetilde{\bm{B}}(z_2)=\begin{bmatrix}0&x_2\\x_2&0\end{bmatrix}$ respectively. Let
 $\bm{C}({z}_0,{z}_1,{z}_2)=\bm{A}({z}_{1})\cdot\bm{D}({z}_{0})\cdot{B}(\bm{z}_{2}).$ From Lemma 9, we obtain the  GBF matrix $$\widetilde{\bm{C}}({x}_0,{x}_1,{x}_2)=
 \begin{bmatrix}
x_0+x_0x_1+x_0x_2 &x_0+x_2+x_0x_1+x_0x_2\\
x_0+x_1+x_0x_1+x_0x_2&x_0+x_1+x_2+x_0x_1+x_0x_2
 \end{bmatrix}.$$
\end{example}

\subsection{The Structure of the PU Matrices of the Known Constructions}

In this subsection, we propose an approach to construct new desired PU matrices by the seed PU matrices of order 2. We will show that the constructions of CSSs proposed in \cite{Paterson00,Schmidt07} and the construction of CCCs proposed in \cite{CCC} are all special cases of the CSSs and CCCs derived from the proposed desired PU matrices.

We have introduced the seed PU matrices of order $N$ in Subsection 4.1. For $N=2$, the corresponding GBF matrices are extracted in Subsection 5.2. We denote the seed PU matrix of order $2$ by $\bm{U}(\bm{z})$ and its corresponding GBF matrix by $\widetilde{\bm{U}}(\bm{x})$. Then we have
\begin{equation}\label{U-order2}
\widetilde{\bm{U}}(\bm{x})=
f(\bm{x})\begin{bmatrix}
1&1\\
1&1
\end{bmatrix}+\frac{q}{2}\begin{bmatrix}
0&x_{\pi({m-1})}\\
x_{\pi(0)}&x_{\pi({0})}+x_{\pi(m-1)}
\end{bmatrix},
\end{equation}
where
\begin{equation}\label{f-order2}
{f}(\bm{x})=
\frac{q}{2}\sum_{k=1}^{m-1}x_{\pi{(k-1)}}x_{\pi({k})}
+\sum_{k=0}^{m-1}c_kx_{k}+c'.
\end{equation}

\begin{definition}\label{def-P}
Let $\bm{P}$ be a permutation matrix of order $2^{n+1}$ with entries
\begin{equation}
P_{u,v}=\begin{cases}
1, &\text{if} \ v\equiv 2u \ (\bmod~{2^{n+1}-1})\\
0, &\text{otherwise}
\end{cases}
\end{equation}
for $0\le u,v<2^{n+1}$. In other word, $P_{u,v}=1$ if and only if $(u_n, u_0, \dots, u_{n-1})=(v_0, v_1 ,\dots, v_n)$ where $(u_{0},u_1 ,\dots ,u_{n})$ and $(v_0 ,v_1,\dots,v_{n})$ are the binary expansions of integers $u$ and $v$ respectively.
\end{definition}

For  $0\le j< 2^n$, let ${\bm{U}}^{\{j\}}(\bm{z})$ be the generating matrix of the GBF matrix $\widetilde{\bm{U}}^{\{j\}}(\bm{x})$ with the form (\ref{U-order2}),  where the function
${f}(\bm{x})$ and permutation $\pi$ are replaced by ${f}^{\{j\}}(\bm{x})$ and $\pi_j$, respectively.
Suppose that $\bm{H}^{\{0\}}, \bm{H}^{\{1\}}\in H(q,2^{n}) $ and their corresponding phase matrices are
$\widetilde{\bm{H}}^{\{0\}},\widetilde{\bm{H}}^{\{1\}}$.

\begin{theorem}\label{thm-8}
Let matrices $\bm{P}$, $\bm{H}^{\{0\}}$, $\bm{H}^{\{1\}}$ and $\bm{U}^{\{j\}}(\bm{z})$ be given as above.
Define a multivariate polynomial matrix of order $2^{n+1}$ by
\begin{equation}\label{U-1}
\bm{G}(\bm{z})
=\begin{bmatrix}\bm{H}^{\{0\}}&\bm{0}\\\bm{0}&\bm{H}^{\{1\}}\end{bmatrix}
\bm{P}	
\begin{bmatrix}
\bm{U}^{\{0\}}(\bm{z})&0&\cdots&0\\
0&\bm{U}^{\{1\}}(\bm{z})&\cdots&0\\
\vdots&\vdots&\ddots&\vdots\\
0&0&\cdots&\bm{U}^{\{2^{n-1}\}}(\bm{z})
\end{bmatrix}
\bm{P}^{T}.
\end{equation}
Then $\bm{G}(\bm{z})$  is a desired PU matrix and the entries of its corresponding GBF $\widetilde{\bm{G}}(\bm{x})$  are functions from $\mathbb{Z}_2^{m}$  to $\Z_q$ with the expression 	
\begin{equation}\label{U1-func} \widetilde{{G}}_{u,v}(\bm{x})=f^{\{j\}}(\bm{x})+\frac{q}{2}\alpha{x}_{\pi_j(0)}+\frac{q}{2}\beta{x}_{\pi_j(m-1)}+\widetilde{{H}}_{i,j}^{\{\alpha\}},
\end{equation}	
where $u=\alpha\cdot 2^{n}+i$ and $v=\beta\cdot 2^{n}+j$ ($0\leq \alpha,\beta\leq 1$ ,$0\leq i,j< 2^{n}$).
  \end{theorem}

\begin{proof}
It is obvious that $\bm{G}(\bm{z})$  is a PU matrix.
From Definition \ref{def-P}, we know that the entries of the permutation matrix $\bm{P}$ equals to 1 at position $(j, 2j)$ and $(j+2^n, 2j+1)$ for $0\le j< 2^{n}$. Let $\bm{A}$ and $\bm{B}$ be two matrices of order $2^{n+1}$ such that $\bm{B}=\bm{P}\bm{A}\bm{P}^T$. Then the entries $B_{j,j}=A_{2j, 2j}$, $B_{j,j+2^n}=A_{2j, 2j+1}$, $B_{j+2^n,j}=A_{2j+1, 2j}$ and $B_{j+2^n,j+2^n}=A_{2j+1, 2j+1}$. We have
\begin{equation*}
\bm{P}\cdot	
\begin{bmatrix}
\bm{U}^{\{0\}}(\bm{z})&0&\cdots&0\\
0&\bm{U}^{\{1\}}(\bm{z})&\cdots&0\\
\vdots&\vdots&\ddots&\vdots\\
0&0&\cdots&\bm{U}^{\{2^{n-1}\}}(\bm{z})
\end{bmatrix}
 \cdot\bm{P}^{T}
=\begin{bmatrix}
\bm{V}_{0,0}(\bm{z})&\bm{V}_{0,1}(\bm{z})\\
\bm{V}_{1,0}(\bm{z})&\bm{V}_{1,1}(\bm{z})
\end{bmatrix}
\end{equation*}
where $ \bm{V}_{\alpha,\beta}(\bm{z}) $ ($0\le\alpha,\beta\le1$) is a diagonal matrix given by
\begin{equation*}
\bm{V}_{\alpha,\beta}(\bm{z})=diag\{{U}_{\alpha,\beta}^{\{0\}}(\bm{z}),{U}_{\alpha,\beta}^{\{1\}}(\bm{z}),\cdots,{U}_{\alpha,\beta}^{\{2^n-1\}}(\bm{z})\}.
\end{equation*}
Note that ${U}_{\alpha,\beta}^{\{j\}}(\bm{z})$ is the generating function of the GBF $f^{\{j\}}(\bm{x})+\frac{q}{2}\alpha{x}_{\pi_j(0)}+\frac{q}{2}\beta{x}_{\pi_j(m-1)}$ for $0\le j<2^n$.

The multivariate polynomial matrix $\bm{G}(\bm{z})$ can be represented by a block matrix
$$\bm{G}(\bm{z})= \begin{bmatrix}
\bm{G}_{0,0}(\bm{z})&\bm{G}_{0,1}(\bm{z})\\
\bm{G}_{1,0}(\bm{z})&\bm{G}_{1,1}(\bm{z})
\end{bmatrix}.$$
On the other hand, we have
\begin{equation*}
 \bm{G}(\bm{z})=
 \begin{bmatrix}\bm{H}^{\{0\}}&\bm{0}\\\bm{0}&\bm{H}^{\{1\}}\end{bmatrix}
 \begin{bmatrix}
\bm{V}_{0,0}(\bm{z})&\bm{V}_{0,1}(\bm{z})\\
\bm{V}_{1,0}(\bm{z})&\bm{V}_{1,1}(\bm{z})
\end{bmatrix}
= \begin{bmatrix}
 \bm{H}^{\{0\}}\bm{V}_{0,0}(\bm{z})&\bm{H}^{\{0\}}\bm{V}_{0,1}(\bm{z})\\
 \bm{H}^{\{1\}}\bm{V}_{1,0}(\bm{z})&\bm{H}^{\{1\}}\bm{V}_{1,1}(\bm{z})
 \end{bmatrix}.
 \end{equation*}
Then each block matrix can be represented by
$$\bm{G}_{\alpha,\beta}(\bm{z})=\bm{H}^{\{\alpha\}}\bm{V}_{\alpha,\beta}(\bm{z})$$
and each entry of $\bm{G}(\bm{z})$ can be given by
\begin{equation*}
{G}_{u,v}(\bm{z})
={H}_{i,j}^{\{\alpha\}}{U}^{\{j\}}_{\alpha,\beta}(\bm{z}),
\end{equation*}
which is the generating function of the GBF
\begin{equation*} \widetilde{{G}}_{u,v}(\bm{x})=f^{\{j\}}(\bm{x})+\frac{q}{2}\alpha{x}_{\pi_j(0)}+\frac{q}{2}\beta{x}_{\pi_j(m-1)}+\widetilde{{H}}_{i,j}^{\{\alpha\}},
\end{equation*}	
where $u=\alpha\cdot 2^{n}+i$ and $v=\beta\cdot 2^{n}+j$ for $0\leq \alpha,\beta\leq 1$	and $0\leq i,j< 2^{n}$. Thus the assertion is established.
\end{proof}

\begin{remark}
It was shown in \cite{CCC} that the CCCs can be  directly constructed by the sequences in CSSs proposed by Paterson in \cite{Paterson00}.  Theorem \ref{thm-8} provides another type of the  direct construction of CCCs, in which the sequences are Golay sequences proposed by Davis and Jedwab in \cite{DavisJedwab99}.
\end{remark}

\begin{example}
Here we give an example of $ \bm{G}(\bm{z}) $ of order $8$. For convenience, the PU matrices ${\bm{U}}^{\{j\}}(\bm{z})$ of order $2$ are denoted by $\begin{bmatrix}
A_{j}&B_{j}\\
C_{j}&D_{j}
\end{bmatrix}$,
and their corresponding GBF matrices $\widetilde{\bm{U}}^{\{j\}}(\bm{x})$ are denoted by $\begin{bmatrix}
a_{j}&b_{j}\\
c_{j}&d_{j}
\end{bmatrix}$ for $0\le j<4$.

Let $\bm{P}$ be a permutation matrix of order $8$ in Definition \ref{def-P}. Then the entry $P_{u,v}=1$ if and only if $(u,v)$ belongs to $\{(0,0), (1,2), (2,4), (3,6), (4,1), (5,3), (6,5), (7,7)\}$, i.e.,
\begin{equation*}
	\bm{P}= \begin{bmatrix}
	\bm{1}&0&0&0&0&0&0&0\\
	0&0&\bm{1}&0&0&0&0&0\\
	0&0&0&0&\bm{1}&0&0&0\\
	0&0&0&0&0&0&\bm{1}&0\\
	0&\bm{1}&0&0&0&0&0&0\\
	0&0&0&\bm{1}&0&0&0&0\\
	0&0&0&0&0&\bm{1}&0&0\\
	0&0&0&0&0&0&0&\bm{1}
	\end{bmatrix}.
	\end{equation*}

 Then we have
\begin{equation*}
\bm{P}\cdot\left[\begin{array}{@{}cc:cc:cc:cc@{}}
A_0&B_0&&&&&&\\
C_0&D_0&&&&&&\\\cdashline{1-8}
&&A_1&B_1&&&&\\
&&C_1&D_1&&&&\\\cdashline{1-8}
&&&&A_2&B_2&&\\
&&&&C_2&D_2&&\\\cdashline{1-8}
&&&&&&A_3&B_3\\
&&&&&&C_3&D_3\end{array}\right] \cdot\bm{P}^{T}
=\left[\begin{array}{@{}cccc:cccc@{}}
A_0&&&&B_0&&&\\
&A_1&&&&B_1&&\\
&&A_2&&&&B_2&\\
&&&A_3&&&&B_3\\\cdashline{1-8}
C_0&&&&D_0&&&\\
&C_1&&&&D_1&&\\
&&C_2&&&&D_2&\\
&&&C_3&&&&D_3\end{array} \right].
\end{equation*}
If we set
\begin{equation*}
\bm{H}^{\{0\}}=\bm{H}^{\{1\}}=
\begin{bmatrix}
1&1&1&1\\
1&-1&1&-1\\
1&1&-1&-1\\
1&-1&-1&1
\end{bmatrix},
\end{equation*}
then we have
\[
\bm{G}(\bm{z})=\left[
\begin{array}{@{}cccc:cccc@{}}
   {A}_{0}&{A}_{1}&{A}_{2}&{A}_{3}&{B}_{0}&{B}_{1}&{B}_{2}&{B}_{3}\\
{A}_{0}&-{A}_{1}&{A}_{2}&-{A}_{3}&{B}_{0}&-{B}_{1}&{B}_{2}&-{B}_{3}\\
{A}_{0}&{A}_{1}&-{A}_{2}&-{A}_{3}&{B}_{0}&{B}_{1}&-{B}_{2}&-{B}_{3}\\
{A}_{0}&-{A}_{1}&-{A}_{2}&{A}_{3}&{B}_{0}&-{B}_{1}&-{B}_{2}&{B}_{3}\\
\cdashline{1-8}
{C}_{0}&{C}_{1}&{C}_{2}&{C}_{3}&{D}_{0}&{D}_{1}&{D}_{2}&{D}_{3}\\
{C}_{0}&-{C}_{1}&{C}_{2}&-{C}_{3}&{D}_{0}&-{D}_{1}&{D}_{2}&-{D}_{3}\\
{C}_{0}&{C}_{1}&-{C}_{2}&-{C}_{3}&{D}_{0}&{D}_{1}&-{D}_{2}&-{D}_{3}\\
{C}_{0}&-{C}_{1}&-{C}_{2}&{C}_{3}&{D}_{0}&-{D}_{1}&-{D}_{2}&{D}_{3}
\end{array}
\right]
\]
and its corresponding function matrix can be expressed by
\[
\widetilde{\bm{G}}(\bm{x})=\left[
\begin{array}{@{}cccc:cccc@{}}
{a}_{0}&{a}_{1}  &{a}_{2}  &{a}_{3}  &{b}_{0}&{b}_{1}  &{b}_{2}  &{b}_{3}\\
{a}_{0}&{a}_{1}+\frac{q}{2}&{a}_{2}  &{a}_{3}+\frac{q}{2}&{b}_{0}&{b}_{1}+\frac{q}{2}&{b}_{2}  &{b}_{3}+\frac{q}{2}\\
{a}_{0}&{a}_{1}  &{a}_{2}+\frac{q}{2}&{a}_{3}+\frac{q}{2}&{b}_{0}&{b}_{1}  &{b}_{2}+\frac{q}{2}&{b}_{3}+\frac{q}{2}\\
{a}_{0}&{a}_{1}+\frac{q}{2}  &{a}_{2}+\frac{q}{2}&{a}_{3}  &{b}_{0}&{b}_{1}+\frac{q}{2}&{b}_{2}+\frac{q}{2}&{b}_{3}\\
\cdashline{1-8}
{c}_{0}&{c}_{1}  &{c}_{2}  &{c}_{3}  &{d}_{0}&{d}_{1}  &{d}_{2}  &{d}_{3}\\
{c}_{0}&{c}_{1}+\frac{q}{2}&{c}_{2}  &{c}_{3}+\frac{q}{2}&{d}_{0}&{d}_{1}+\frac{q}{2}&{d}_{2}  &{d}_{3}+\frac{q}{2}\\
{c}_{0}&{c}_{1}  &{c}_{2}+\frac{q}{2}&{c}_{3}+\frac{q}{2}&{d}_{0}&{d}_{1}  &{d}_{2}+\frac{q}{2}&{d}_{3}+\frac{q}{2}\\
{c}_{0}&{c}_{1}+\frac{q}{2}&{c}_{2}+\frac{q}{2}&{c}_{3}  &{d}_{0}&{d}_{1}+\frac{q}{2}&{d}_{2}+\frac{q}{2}&{d}_{3}
\end{array}
\right].
\]
\end{example}


Before introducing another type of desired PU matrices, we define multivariate variables
\begin{equation}\label{z1z2-2}
\left\{
\begin{aligned}
&\bm{z_0}=(z_{0}, z_{1},\dots, z_{m-1}),\\
&\bm{z_1}=(z_{m}, z_{m+1},\dots, z_{m+n-1}),\\
\end{aligned}\right.
\end{equation}
and Boolean variables
\begin{equation}\label{x1x2-2}
\left\{
\begin{aligned}
&\bm{x_0}=(x_{0}, x_{1},\dots, x_{m-1})\in \Z_2^{m},\\
&\bm{x_1}=(x_{m}, x_{m+1},\dots, x_{m+n-1})\in \Z_2^{n}.\\
\end{aligned}\right.
\end{equation}

Let $\bm{G}(\bm{z}_{0})$ be a desired PU matrix of order $ 2^{n+1} $ obtained from Theorem \ref{thm-8}, and $\widetilde{\bm{G}}(\bm{x}_0)$ its corresponding GBF matrix.
Suppose that $\bm{H}^{\{2\}}, \bm{H}^{\{3\}}\in H(q,2^{n}) $ and their corresponding phase matrices are
$\widetilde{\bm{H}}^{\{2\}},\widetilde{\bm{H}}^{\{3\}}$. Recall the definition of the generalized delay matrix $\bm{D}(\bm{z}_1)= \bm{D}(z_{m+n-1})\otimes\cdots\otimes\bm{D}(z_{m+1})\otimes\bm{D}(z_{m})$. Then we have the following theorem.

\begin{theorem}\label{thm-9}
Let matrices $\bm{G}(\bm{z}_{0})$, $\bm{H}^{\{2\}}$ and $\bm{H}^{\{3\}}$ be given as above. Then $\bm{M}(\bm{z_0}, \bm{z_1})$ defined by
\begin{equation}\label{U-2}
\bm{M}(\bm{z_0}, \bm{z_1})=
\bm{G}(\bm{z}_{0})
\begin{bmatrix}\bm{D}(\bm{z}_{1})&\bm{0}\\\bm{0}&\bm{D}(\bm{z}_{1})\end{bmatrix}
\begin{bmatrix}\bm{H}^{\{2\}}&\bm{0}\\\bm{0}&\bm{H}^{\{3\}}\end{bmatrix}
\end{equation}
is a desired PU matrix, and the entries of its corresponding GBF $\widetilde{\bm{M}}(\bm{x}_0, \bm{x}_1)$  are functions from $\mathbb{Z}_2^{m+n}$  to $\Z_q$ with the expression 	
\begin{equation}\label{U2-func}
\widetilde{{M}}_{u,v}(\bm{x_0}, \bm{x_1})=\sum_{i=0}^{2^{n}-1} \left(f^{\{i\}}(\bm{x_0})+\frac{q}{2}\alpha{x}_{\pi_i(0)}+\frac{q}{2}\beta{x}_{\pi_i(m-1)}+
\widetilde{{H}}_{l,i}^{\{\alpha\}}+\widetilde{{H}}_{i,j}^{\{\beta+2\}}\right){g}_{i}(\bm{x}_{1}),
\end{equation}
where $u=\alpha\cdot 2^{n}+l$ and $v=\beta\cdot 2^{n}+j$ for $0\leq \alpha,\beta\leq 1$	and $0\leq l,j< 2^{n}$.	
\end{theorem}

\begin{proof}
It is obvious that $\bm{M}(\bm{z_0}, \bm{z_1})$  is a PU matrix.
The multivariate polynomial matrices $\bm{G}(\bm{z_1})$ and $\bm{M}(\bm{z_0}, \bm{z_1})$
can be interpreted by  block matrices
$$\bm{G}(\bm{z_0})= \begin{bmatrix}
\bm{G}_{0,0}(\bm{z_0})&\bm{G}_{0,1}(\bm{z_0})\\
\bm{G}_{1,0}(\bm{z_0})&\bm{G}_{1,1}(\bm{z_0})
\end{bmatrix}$$
and
$$\bm{M}(\bm{z_0}, \bm{z_1})=
\begin{bmatrix}
\bm{M}_{0,0}(\bm{z_0}, \bm{z_1})&\bm{M}_{0,1}(\bm{z_0}, \bm{z_1})\\
\bm{M}_{1,0}(\bm{z_0}, \bm{z_1})&\bm{M}_{1,1}(\bm{z_0}, \bm{z_1})
\end{bmatrix},$$
respectively. From the definition, we have
\begin{eqnarray*}
\bm{M}(\bm{z_0}, \bm{z_1})&=&
\begin{bmatrix}
\bm{G}_{0,0}(\bm{z_0})&\bm{G}_{0,1}(\bm{z_0})\\
\bm{G}_{1,0}(\bm{z_0})&\bm{G}_{1,1}(\bm{z_0})
\end{bmatrix}
\begin{bmatrix}\bm{D}(\bm{z}_{1})&\bm{0}\\\bm{0}&\bm{D}(\bm{z}_{1})\end{bmatrix}
\begin{bmatrix}\bm{H}^{\{2\}}&\bm{0}\\\bm{0}&\bm{H}^{\{3\}}\end{bmatrix}\\
&=&
\begin{bmatrix}
\bm{G}_{0,0}(\bm{z_0})\bm{D}(\bm{z}_{1})\bm{H}^{\{2\}}&\bm{G}_{0,1}(\bm{z_0})\bm{D}(\bm{z}_{1})\bm{H}^{\{3\}}\\
\bm{G}_{1,0}(\bm{z_0})\bm{D}(\bm{z}_{1})\bm{H}^{\{2\}}&\bm{G}_{1,1}(\bm{z_0})\bm{D}(\bm{z}_{1})\bm{H}^{\{3\}}
\end{bmatrix},
\end{eqnarray*}
which imply
\begin{equation*}
\bm{M}_{\alpha,\beta}(\bm{z_0}, \bm{z_1})=\bm{G}_{\alpha,\beta}(\bm{z}_{0})\bm{D}(\bm{z}_{1})\bm{H}^{\{\beta+2\}}
\end{equation*}
for $0\leq \alpha,\beta\leq 1$.

According to Lemma \ref{lem-9}, we have
$$\widetilde{{M}}_{\alpha\cdot 2^{n}+l, \beta\cdot 2^{n}+j}(\bm{x}_0, \bm{x}_1)=\sum_{i=0}^{2^n-1}\left( \widetilde{{G}}_{\alpha\cdot 2^{n}+l, \beta\cdot 2^{n}+i}(\bm{x}_0)+{\widetilde{H}}^{\{\beta+2\}}_{i,j}\right){g}_{i}(\bm{x}_{1}),$$
which completes the proof by substituting (\ref{U1-func}) into the above formula.
\end{proof}

In Theorem \ref{thm-9}, suppose $ \widetilde{H}_{0,i}^{\{0\}}=\widetilde{H}_{i,0}^{\{2\}}=0 $ in formula (\ref{U2-func}) for $0\le i\le 2^{n}-1$, and $\alpha=\beta=0$. We have
\begin{equation*}
\widetilde{{M}}_{0,0}(\bm{x_0}, \bm{x_1})=\sum_{i=0}^{2^{n}-1} f^{\{i\}}(\bm{x_0}){g}_{i}(\bm{x}_{1}).
\end{equation*}
We use it to illustrate the CSSs proposed in \cite{Paterson00,Schmidt07} and  CCCs proposed in \cite{CCC}.

\begin{corollary}\label{coro-4}
The array   (GBF)
\begin{equation}\label{exam-array}
f(\bm{x_0}, \bm{x_1})=\sum_{i=0}^{2^{n}-1} f^{\{i\}}(\bm{x_0}){g}_{i}(\bm{x}_{1})
\end{equation}
from $\Z_2^{m+n}$ to $\Z_q$ lies in a CAS of size $2^{n+1}$.
\end{corollary}

Notice that  if $f(\bm{x_0}, \bm{x_1})$ is given in formula (\ref{exam-array}), we have
$$ f(\bm{x_0}, \bm{x_1})|_{\bm{x}_{1}=\bm{i}}=f^{\{i\}}(\bm{x_0}),$$
where $\bm{i}=(i_0 , i_1, \cdots, i_{n-1})$ is the binary expansion of integer $i$.
Therefore, the sequences proposed in \cite{Schmidt07} are exactly evaluated by the GBF: $\pi \cdot f(\bm{x_0}, \bm{x_1})$,  where $f(\bm{x_0}, \bm{x_1})$ is shown in formula (\ref{exam-array}) and $\pi$ is an arbitrary  permutation of binary variables $(\bm{x_0}, \bm{x_1})$. The sequences proposed in \cite{Paterson00} are a subset of those  sequences where all the permutation $\pi_i$ in functions $f^{\{i\}}(\bm{x_0})$ are the identity permutation for all $i$ and $f(\bm{x_0}, \bm{x_1})$ is a quadratic GBF.

\begin{remark}
In Section 5, it is illustrated that all the $q$-ary Golay sequences proposed in \cite{DavisJedwab99} can be derived from a single Golay array of size $2\times 2 \times \cdots \times 2$. Corollary \ref{coro-4} showed that all the sequences in  CSSs proposed in \cite{Paterson00, Schmidt07} can also be derived from arrays of size $2\times 2 \times \cdots \times 2$ in some CASs.
\end{remark}

Let  $ \bm{H}^{\{k\}}=\bm{H} (0\leq k\leq 3)$  in formulas (\ref{U-1}) and (\ref{U-2}) be Hadamard matrices, where $H_{i,j}= (-1)^{\bm{i}\cdot\bm{j}}$, $\bm{i}$ and $\bm{j}$ are the binary expansions of integer $i$ and $j$ respectively, and $\bm{i}\cdot\bm{j}$ denotes the dot product over $\Z_2$. The following result can be  obtained immediately from Theorem \ref{thm-9}.
\begin{corollary}
Let $f(\bm{x_0}, \bm{x_1})$ be of the form (\ref{exam-array}). For $0\le u, v\le 2^{n+1}$,
the  arrays   (GBFs)
\begin{equation}\label{exam-seq1}
f_{u,v}(\bm{x_0}, \bm{x_1})=f(\bm{x_0}, \bm{x_1})+\frac{q}{2}({\bm{i}\cdot\bm{x}_{0}}+{\bm{j}\cdot\bm{x}_{0}})+\frac{q}{2}\alpha\sum_{i=0}^{2^{n}-1}
{x}_{\pi_i(n)}{g}_{i}(\bm{x}_{0})+\frac{q}{2}\beta\sum_{i=0}^{2^{n}-1}
{x}_{\pi_i(n+m-1)}{g}_{i}(\bm{x}_{0})
\end{equation}
form  a CCA of size $2^{n+1}$, where
$u=\alpha\cdot 2^{n}+i$ and $v=\beta\cdot 2^{n}+j$ for $0\leq \alpha,\beta\leq 1$	and $0\leq i,j< 2^{n}$.
\end{corollary}

Suppose that all the permutations $\pi_i$ are the identity permutation  for $0\le i\le 2^{n}-1$. The following result follows immediately from the above corollary.
\begin{corollary}\label{coro3}
Let $f(\bm{x_0}, \bm{x_1})$ be of the form (\ref{exam-array}) and every $\pi_i$ in  $f^{\{i\}}(\bm{x_0})$ be the identity permutation. For $0\le u, v\le 2^{n+1}$,
the  arrays   (GBFs)
\begin{equation}\label{exam-seq2}
f_{u,v}(\bm{x_0}, \bm{x_1})=f(\bm{x_0}, \bm{x_1})+\frac{q}{2}({\bm{i}\cdot\bm{x}_{0}}+{\bm{j}\cdot\bm{x}_{0}}+\alpha{x}_{n}+\beta{x}_{n+m-1})
\end{equation}
form  a CCA of size $2^{n+1}$, where $u=\alpha\cdot 2^{n}+i$ and $v=\beta\cdot 2^{n}+j$ for $0\leq \alpha,\beta\leq 1$	and $0\leq i,j< 2^{n}$.
\end{corollary}
The CCC proposed in \cite{CCC} is a special case of Corollary \ref{coro3} by applying  permutation $\pi$ acting on binary variables $(\bm{x_0}, \bm{x_1})$, where $f(\bm{x_0}, \bm{x_1})$ is a quadratic function.

From the discussions above, the sequences and CSSs proposed in \cite{Paterson00, Schmidt07} and CCCs proposed in \cite{CCC} can be derived from  Theorem \ref{thm-9}.  Furthermore, there are many new direct constructions of CCCs with explicit GBFs are proposed in Theorems \ref{thm-8} and \ref{thm-9}.
\section{Generalized Constructions}

From Theorem \ref{thm-9}, we have shown that the known constructions of CSSs in \cite{Paterson00, Schmidt07} and CCCs in \cite{CCC} can be derived form the seed PU matrices of order 2. From Section 5, we know that the explicit GBF forms of the seed PU matrices of higher order were not appeared into literature before. In this section, we generalize the results in Theorems \ref{thm-8} and \ref{thm-9} by applying desired PU matrices with flexible order. It is straightforward that new CSSs and CCCs can be constructed from Theorem \ref{thm-1}.

The first generalization is immediately from Theorem \ref{thm-8} by replacing the BH matrices $\bm{H}^{\{\alpha\}}$ ($\alpha=0, 1$) by desired PU matrices $\bm{V}^{\{\alpha\}}(\bm{z_1})$ of order $2^n$.
Let $\widetilde{\bm{V}}^{\{\alpha\}}(\bm{x_1})$ be the corresponding GBF matrix of $\bm{V}^{\{\alpha\}}(\bm{z_1})$ .
Note that here we use multivariate variables
\begin{equation}\label{z1z2-3}
\left\{
\begin{aligned}
&\bm{z_0}=(z_{0}, z_{1},\dots, z_{m-1}),\\
&\bm{z_1}=(z_{m}, z_{m+1},\dots, z_{m+m'-1}),\\
\end{aligned}\right.
\end{equation}
and Boolean variables
\begin{equation}\label{x1x2-3}
\left\{
\begin{aligned}
&\bm{x_0}=(x_{0}, x_{1},\dots, x_{m-1})\in \Z_2^{m},\\
&\bm{x_1}=(x_{m}, x_{m+1},\dots, x_{m+m'-1})\in \Z_2^{m'}.\\
\end{aligned}\right.
\end{equation}

\begin{corollary}\label{coro-7}
Let the matrices $\bm{U}^{\{j\}}(\bm{z_0})$ and the permutation matrix $\bm{P}$ be the same as those in Theorem \ref{thm-8}, and $\bm{V}^{\{0\}}(\bm{z_1})$, $\bm{V}^{\{1\}}(\bm{z_1})$ given as above. Define a multivariate polynomial matrix of order $2^{n+1}$ by
\begin{equation}\label{U-3}
\bm{G}(\bm{z_0}, \bm{z_1})
=\begin{bmatrix}\bm{V}^{\{0\}}(\bm{z_1})&\bm{0}\\\bm{0}&\bm{V}^{\{1\}}(\bm{z_1})\end{bmatrix}
\cdot\bm{P}\cdot	
\begin{bmatrix}
\bm{U}^{\{0\}}(\bm{z_0})&0&\cdots&0\\
0&\bm{U}^{\{1\}}(\bm{z_0})&\cdots&0\\
\vdots&\vdots&\ddots&\vdots\\
0&0&\cdots&\bm{U}^{\{2^{n-1}\}}(\bm{z_0})
\end{bmatrix}
\bm{P}^{T}.
\end{equation}
Then $\bm{G}(\bm{z_0}, \bm{z_1})$  is a desired PU matrix. And it is the generating matrix of GBF matrix $\widetilde{\bm{G}}(\bm{x_0}, \bm{x_1})$ whose  entries are GBFs from $\mathbb{Z}_2^{m+m'}$  to $\Z_q$ with the expression 	
\begin{equation}\label{U3-func} \widetilde{{G}}_{u,v}(\bm{x_0}, \bm{x_1})=f^{\{j\}}(\bm{x_0})+\frac{q}{2}\alpha{x}_{\pi_j(0)}+\frac{q}{2}\beta{x}_{\pi_j(m-1)}+\widetilde{V}_{i,j}^{\{\alpha\}}(\bm{x_1}),
\end{equation}	
where $u=\alpha\cdot 2^{n}+i$ and $v=\beta\cdot 2^{n}+j$ for $0\leq \alpha,\beta\leq 1$	and $0\leq i,j< 2^{n}$.
\end{corollary}

Moreover, let $\bm{U}^{\{j\}}(\bm{z_0})$ ($0\le j<2^n$) and $\bm{V}^{\{\alpha\}}(\bm{z_1})$ ($0\le \alpha<2^{n'}$) be desired PU matrices of order $2^{n'}$ and $2^n$,  whose corresponding GBF matrices are $\widetilde{\bm{U}}^{\{j\}}(\bm{x_0})$ and $\widetilde{\bm{V}}^{\{\alpha\}}(\bm{x_1})$,
respectively. Note that $\bm{z_0}, \bm{z_1}, \bm{x_0}, \bm{x_1}$ are defined the same as those in (\ref{z1z2-3}) and (\ref{x1x2-3}). Then we have the following assertion.

\begin{theorem}\label{thm-10}
Let $\bm{P}$ be a permutation matrix of order $2^{n+n'}$ with each entry $P_{u,v}=1$ if and only if $v\equiv 2^{n'}u$  (mod  $2^{n+n'}-1$), $\bm{U}^{\{j\}}(\bm{z_0})$ ($0\le j<2^n$) and $\bm{V}^{\{\alpha\}}(\bm{z_1})$ ($0\le \alpha<2^{n'}$) shown above.
Define a multivariate polynomial matrix of order $2^{n+n'}$ by
\begin{equation}\label{U4}
\bm{G}(\bm{z_0}, \bm{z_1})=
\begin{bmatrix}
\bm{V}^{\{0\}}(\bm{z_1})&0&\cdots&0\\
0&\bm{V}^{\{1\}}(\bm{z_1})&\cdots&0\\
\vdots&\vdots&\ddots&\vdots\\
0&0&\cdots&\bm{V}^{\{2^{n'}-1\}}(\bm{z_1})
\end{bmatrix}
\bm{P}
\begin{bmatrix}
\bm{U}^{\{0\}}(\bm{z_0})&0&\cdots&0\\
0&\bm{U}^{\{1\}}(\bm{z_0})&\cdots&0\\
\vdots&\vdots&\ddots&\vdots\\
0&0&\cdots&\bm{U}^{\{2^n-1\}}(\bm{z_0})
\end{bmatrix}
\bm{P}^{T}.
\end{equation}
Then $\bm{G}(\bm{z_0}, \bm{z_1})$ is a desired PU matrix of order $2^{n+n'}$. And it is the generating matrix of GBF matrix $\widetilde{\bm{G}}(\bm{x_0}, \bm{x_1})$ whose  entries are GBFs from $\mathbb{Z}_2^{m+m'}$  to $\Z_q$ with the expression 	
\begin{equation}\label{U4-func} \widetilde{{G}}_{u,v}(\bm{x_0}, \bm{x_1})=\widetilde{U}_{\alpha,\beta}^{\{j\}}(\bm{x_0})+\widetilde{V}_{i,j}^{\{\alpha\}}(\bm{x_1}),
\end{equation}	
where $u=\alpha\cdot 2^{n}+i$ and $v=\beta\cdot 2^{n}+j$ for $0\leq \alpha,\beta\leq 2^{n'}$	and $0\leq i,j< 2^{n}$.
\end{theorem}

We omit the proof of Theorem \ref{thm-10} since it is similar to the proof of Theorem \ref{thm-8}.
\begin{remark}
If $\bm{U}^{\{j\}}(\bm{z_0})=\bm{U}(\bm{z_0})$ for  $0\le j<2^n$ and $\bm{V}^{\{\alpha\}}(\bm{z_1})=\bm{V}(\bm{z_1})$ for $0\le \alpha<2^{n'}$ in Theorem \ref{thm-10}, then we have
$$\bm{G}(\bm{z_0}, \bm{z_1})=\bm{U}(\bm{z_0})\otimes \bm{V}(\bm{z_1}).$$
\end{remark}

\begin{corollary}\label{coro4}
Let $f_0(\bm{x_0})$ and $f_1(\bm{x_1})$ be two arrays (GBFs) lying in  CCAs  of size $2^{n}$ and  $2^{n'}$, respectively. Then the GBF
 \begin{equation*}
f(\bm{x}_{0}, \bm{x}_{1})=f_0(\bm{x_0})+f_1(\bm{x_1})
\end{equation*}
is an array of size $\underbrace{2\times 2\times \cdots \times 2}_{m+m'}$ lying in a CAS of size $2^{n+n'}$. The sequence evaluated by GBF $\pi\cdot f(\bm{x}_{0}, \bm{x}_{1})+f'(\bm{x}_{0}, \bm{x}_{1})$ lies in a CSS of size $2^{n+n'}$, where $\pi$ is a permutation of Boolean variables $(\bm{x}_{0}, \bm{x}_{1})$ and $f'(\bm{x}_{0}, \bm{x}_{1})$ is an affine GBF.
\end{corollary}
\begin{proof}
It can be verified by setting $u=v=0$ in Theorem \ref{thm-10}.
\end{proof}

Next we generalize the results in Theorem \ref{thm-9}. Let
$\bm{V}(\bm{z_1})$ and $\bm{U}^{\{\beta \}}(\bm{z_2})$ ($0\le \beta<2^{n'}$)
 be desired PU matrices of order $2^{n+n'}$ and $2^{n}$, whose corresponding GBF matrices are $\widetilde{\bm{V}}(\bm{x_1})$ and $\widetilde{\bm{U}}^{\{\beta\}}(\bm{x_2})$, respectively.  Note that here $\bm{z_0}, \bm{z_1}, \bm{z_2}$ and  $\bm{x_0}, \bm{x_1}, \bm{x_2}$ are defined the same as those in (\ref{z1z2}) and (\ref{x1x2}), respectively.
Then we have the following theorem.

\begin{theorem}\label{thm-11}
Let matrices $\bm{V}(\bm{z_1})$ and $\bm{U}^{\{\beta\}}(\bm{z_2})$ ($0\le \beta<2^{n'}$) be given as above. Then
\begin{equation}\label{U5}
\bm{M}(\bm{z_0}, \bm{z_1}, \bm{z_2})=
\bm{V}(\bm{z}_{1})
\begin{bmatrix}
\bm{D}(\bm{z}_{0})&0&\cdots&0\\
0&\bm{D}(\bm{z}_{0})&\cdots&0\\
\vdots&\vdots&\ddots&\vdots\\
0&0&\cdots&\bm{D}(\bm{z}_{0})
\end{bmatrix}\cdot
\begin{bmatrix}
\bm{U}^{\{0\}}(\bm{z_2})&0&\cdots&0\\
0&\bm{U}^{\{1\}}(\bm{z_2})&\cdots&0\\
\vdots&\vdots&\ddots&\vdots\\
0&0&\cdots&\bm{U}^{\{2^{n'}-1\}}(\bm{z_2})
\end{bmatrix}
\end{equation}
is a desired PU matrix.
And it is the generating matrix of GBF matrix $\widetilde{\bm{M}}(\bm{x}_0, \bm{x}_1, \bm{x}_2)$ whose  entries are functions from $\Z_2^{m_1+m_2+n}$  to $\Z_q$ with the expression 	
\begin{equation}\label{U5-func}
\widetilde{{M}}_{u,v}(\bm{x_0}, \bm{x_1}, \bm{x_2})=\sum_{i=0}^{2^{n}-1} \left(\widetilde{V}_{u,\beta\cdot 2^{n}+i}(\bm{x_1})+\widetilde{U}_{i,j}^{\{\beta\}}(\bm{x_2}) \right){g}_{i}(\bm{x}_{0}),
\end{equation}	
where $0\le u, v<2^{n+n'}$,   $v=\beta\cdot 2^{n}+j$ for $0\leq \beta< 2^{n'}$	and $0\leq j< 2^{n}$.
\end{theorem}

We omit the proof of Theorem \ref{thm-11},  since it is similar to the proof of Theorem \ref{thm-9} by applying Lemma \ref{lem-9}.

\begin{corollary}\label{coro-9}
Let  $\{f^{\{1\}}_i (\bm{x_1}) \}_{0\le i<2^{n+n'}}$ and  $\{f^{\{2\}}_i(\bm{x_2}) \}_{0\le i<2^{n}}$ be two CASs lying in  CCAs  of size $2^{n+n'}$ and  $2^{n}$, respectively. Then we have GBFs
 \begin{equation*}
f(\bm{x}_{0}, \bm{x}_{1}, \bm{x}_{2})=\sum_{i=0}^{2^{n}-1} (f^{\{1\}}_i(\bm{x_1}) +f^{\{2\}}_i(\bm{x_2})){g}_{i}(\bm{x}_{0})
\end{equation*}
is an array of size $\underbrace{2\times 2\times \cdots \times 2}_{m_1+m_2+n}$ lying in a CAS of size $2^{n+n'}$. The sequences  evaluated by  GBFs $\pi\cdot f(\bm{x}_{0}, \bm{x}_{1}, \bm{x}_{2})+f'(\bm{x}_{0}, \bm{x}_{1}, \bm{x}_{2})$ lies in a CSS of size $2^{n+n'}$, where $\pi$ is a permutation of Boolean variables $(\bm{x}_{0}, \bm{x}_{1}, \bm{x}_{2})$ and $f'(\bm{x}_{0}, \bm{x}_{1}, \bm{x}_{2})$ is an affine GBF.
\end{corollary}
\begin{proof}
It can be verified by setting $u=v=0$ in Theorem \ref{thm-11}.
\end{proof}

\begin{example}
It is experimentally shown in  \cite[Section VII]{Paterson00}, the PMEPR of the sequences evaluated by first order Reed-Muller coset of Boolean function
$$f(x_0, x_1, x_2, x_3, x_4)=x_0x_1+x_0x_4+x_1x_4+x_2x_4+x_3x_4$$
equals to $3.449$. However all the theorems in \cite{Paterson00} can only show that the PMEPR is bounded by $8$, since it lies in a CSS of size $8$. Note that an explanation of this example is given in \cite{Chen06}. Here  we illustrate the function $f(x_0, x_1, x_2, x_3, x_4)$ lies in a CAS of size $4$ by applying Corollary \ref{coro-9}.

Let $n=2, n'=0$, $f^{\{2\}}_i=0$ for $1\le i\le 3$, and $f^{\{1\}}_i$ be given as follows.
\begin{equation*}
\left\{
\begin{aligned}
&f^{\{1\}}_0(x_1, x_3, x_4)=x_1x_4+x_3x_4,\\
&f^{\{1\}}_1(x_1, x_3, x_4)=x_1x_4+x_3x_4+x_1+x_4,\\
&f^{\{1\}}_2(x_1, x_3, x_4)=x_1x_4+x_3x_4+x_4,\\
&f^{\{1\}}_3(x_1, x_3, x_4)=x_1x_4+x_3x_4+x_1.
\end{aligned}\right.
\end{equation*}
It is obvious that $\{f^{\{1\}}_0, f^{\{1\}}_3\}$ and $\{f^{\{1\}}_1, f^{\{1\}}_2\}$ are both GAPs. So $\{f^{\{1\}}_0, f^{\{1\}}_1, f^{\{1\}}_2, f^{\{1\}}_3\}$ lies in a CAS of size 4. Then it is easy to check
$$f(x_0, x_1, x_2, x_3, x_4)=\sum_{i=0}^{3}f^{\{1\}}_i(x_1, x_3, x_4){g}_{i}(x_0, x_2),$$
so $f(x_0, x_1, x_2, x_3, x_4)$ lies in a CAS of size 4 by Corollary \ref{coro-9} and the PMEPR of the sequences evaluated by the first order Reed-Muller coset of this Boolean function is bounded by $4$.
\end{example}

\section{Concluding Remarks}

In this paper, we present a state-of-the-art method  to construct CCSs and CCCs by PU matrices, which can be divided into three steps:
\begin{itemize}
\item[(1)]Construct a desired PU matrix $U(\bm{z})$;
\item[(2)]Extract the corresponding function matrix $\widetilde{U}(\bm{x})$;
\item[(3)] Construct CCA, CAS, CSS and CCC according to Theorem \ref{thm-1}.
\end{itemize}

The results in this paper are summarized in Figure \ref{fig-3}. Our proposed approach not only  constructs    a large number of new CSSs and CCCs with explicit function forms, but also  explains  existing constructions with explicit  GBF form in the literature, i.e., they all can be considered as a special case of our constructions.  Especially,  the well known construction of Golay sequences in  \cite{DavisJedwab99} can be extracted from the seed PU matrices of order $2$,  and the known
constructions of CSSs and CCCs with explicit GBFs in \cite{Paterson00,CCC,Schmidt07} can be obtained by the recursive formula given  in Section 7, which also only involves the seed PU matrices of order 2. Although we did not explicitly show the  constructions in \cite{Chen06,Chen08, Wu2016}, they  can also be obtained by the recursive formula which only involves the seed PU matrices of order 2. This observation makes the functions derived from (generalized) seed PU matrices of order $N>2$ very interesting.

\usetikzlibrary{shapes,arrows}
\begin{center}
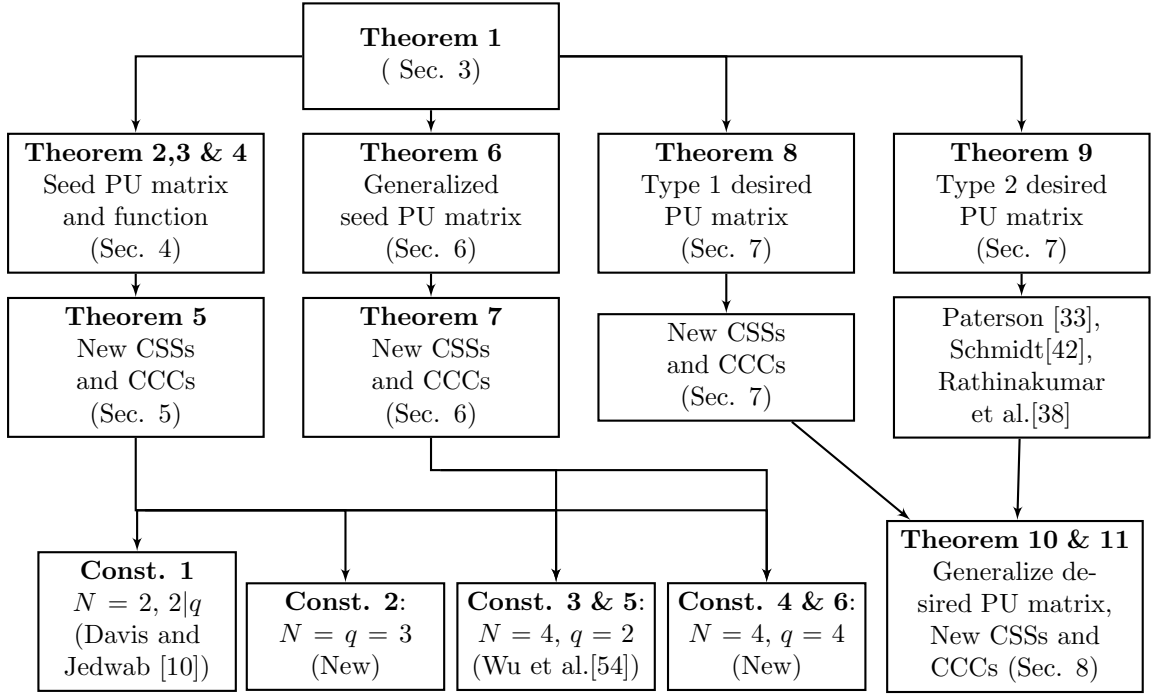
\begin{figure}[!t]
	\ls{1.0}
	\begin{tikzpicture}
	[auto,
	block/.style ={rectangle, draw=black, thick, fill=white,
		text width=9em, text centered,
		minimum height=4em},
	block'/.style ={rectangle, draw=black, thick, fill=white,
		text width=10em, text ragged,
		minimum height=4em},
	block''/.style ={rectangle, draw=black, thick, fill=white,
		text width=6.8em, text centered,
		minimum height=4em},
	line/.style ={draw, thick, -latex', shorten >=0pt}]
	\matrix [column sep=5mm,row sep=3mm] {
		& \node [block] (02) {\textbf{Theorem 1} \\ (  Sec. 3)};
		\\
		\node [block] (11) {\textbf{Theorem 2,3 \& 4} \\Seed PU matrix and function\\(Sec. 4)};
		& \node [block] (12) {\textbf{Theorem 6} \\Generalized seed PU matrix\\(Sec. 6)};
		& \node [block] (13) {\textbf{Theorem 8} \\Type 1 desired PU matrix\\(Sec. 7) };
		& \node [block] (14) {\textbf{Theorem 9} \\Type 2 desired PU matrix\\(Sec. 7) };\\
		\node [block] (21) {\textbf{Theorem 5} \\New CSSs and CCCs\\(Sec. 5)};
		& \node [block] (22) {\textbf{Theorem 7} \\New CSSs and CCCs\\(Sec. 6)};
		& \node [block] (23) {New CSSs and CCCs\\(Sec. 7)};
		& \node [block] (24) {Paterson \cite{Paterson00},\\ Schmidt\cite{Schmidt07},\\ Rathinakumar et al.\cite{CCC}};
		\\
		\node(31){}; & \node(32){};
		\\
		\node(41){}; & \node(42){};\\
	};
	
	\node at (-5.87,-4.8) [block''] (51) {\textbf{Const. 1} \\ $N=2$, $2|q$  \\(Davis and Jedwab \cite{DavisJedwab99})};
	\node at (-3.1,-5) [block''] (52) {\textbf{Const. 2}: $N=q=3$ \\ (New)};
	\node at (-0.3,-5) [block''] (53) {\textbf{Const. 3 \& 5}: $N=4$, $q=2$ \\ (Wu et al.\cite{Wu2016})};
	\node at (2.5,-5) [block''] (54) {\textbf{Const. 4 \& 6}: $N=4$, $q=4$\\ (New)};
	\node at (5.8,-4.6) [block] (55) {\textbf{Theorem 10 \& 11} \\Generalize desired PU matrix, New CSSs and CCCs (Sec. 8)};
	\begin{scope}[every path/.style=line]
	\path (02)   -| (11);	\path (02)   -- (12);	\path (02)   -| (13);	\path (02)   -| (14);
	\path (11)   -- (21);	\path (12)   -- (22);	\path (13)   -- (23);	\path (14)   -- (24);
	\draw[-] (21)   |- (41);	\draw[-](22)   |- (32);  \path(24)   -- (55);\path(23)   -- (55);
	\path (41)   -|(51);	\path (41)   -| (52);	\path (41)   -|(53);	\path (41)   -| (54);
	\path (32)   -| (53);	\path (32)   -| (54);
	\end{scope}
	\end{tikzpicture}
	\caption{The Summarization of the Results}\label{fig-3}
\end{figure}
\end{center}

On the other hand, we propose a systematic approach to extract the explicit forms of the functions from  seed PU matrices, which is a key contribution of this paper. The general form of these functions only depends on a basis of functions from  $\Z_N$ to $\Z_q$ and representatives in the equivalent class of phase BH matrices.  And we proved that the sequences extracted from the seed PU matrices must fill up a large number of cosets of a linear code. For $N=q=3$, the ternary complementary sequences in CSS of size 3, constructed from the seed PU matrices of order 3, are never reported before. Moreover, for the quaternary complementary sequences of size 4, most of the quaternary  sequences constructed here are novel. We would like to elaborate this more.  The inspired narrative for  the new sequences extracted from seed PU matrices of order 4 can be explained as follows.

There is only one representative for binary BH matrices which are the Walsh matrices, while there are two representatives for quaternary BH matrices: the Walsh matrices and  the Fourier matrices.
We observe that all the known sequences reported in the literature are extracted from the PU matrices for which only the equivalent class of the  Walsh matrices of order 4 is involved. Note that Walsh matrices of order 4 can be generated by the Kronecker product of the Walsh matrices of order 2. So these sequences can also be produced by the recursive formula in Section 8 by seed PU matrices of order 2. On the other hand, even if there is only  one BH matrix which is equivalent to a Fourier matrix in the seed PU matrices of order 4, any sequences extracted from such a seed PU matrix are new.

Note that there are 15 equivalent classes of  BH matrices in $H(4,8)$, and 319 equivalent classes of  BH matrices in $H(4,12)$ \cite{Hada}. Thus, the results for the seed PU matrices of order $N$, combined with  the recursive formula in Section 8, may exponentially increase the number of the quaternary  sequences in CSSs of size $N$.

The open problem  given in \cite{Tseng72} has influenced the field of constructing new CSSs and CCCs for almost 5 decades, i.e.,
{\em Obtain direct construction procedures for complementary sets with given parameters, namely, the number of sequences in the set and their lengths.} From the PU matrix construction in \cite{TSP2018, Marziani,Suehiro88}, the CCSs constructed in \cite{Tseng72} by Hadamard matrices are a subset of the CSSs derived from our seed PU matrices. It is stated  that the results in \cite{Paterson00}  provide a partial solution to this open problem. Nevertheless,  the results in this paper present a much more deeply nailed down solution to this  open problem.

We now pay attention back to  the functions derived from seed PU matrices, which contain two components for each function. One component is the linear term $S_L(q, N)$ which  can be easily obtained.   However, we need to calculate the second component, i.e., the quadratic term  $S_{Q}(q, N)$ intensely.  Note that  the numbers of $\chi_{L}$ and $\chi_{R}$ which are permutations of symbols $\{0, 1, \cdots, {N-1}\}$, are both equal to $N!$. Thus, the computation complexity of the quadratic terms is about $(N!)^2$ from only one representative  of
BH matrices. Is it possible to  find a general explicit GBF form derived from seed PU matrices for specified BH matrices? Currently,  we are able to
obtain the sequences derived from the seed PU matrices involving in Fourier matrices or Walsh matrices, which will be presented in a separate work \cite{CSS-Permutation}. Surprisedly these constructions provide a connection between the complementary sequences and permutation polynomials over finite fields.

An important application of sequences in CSSs is  PMEPR control. From our recursive framework, binary and quaternary sequences in CSSs of order 4 are determined by the generalized seed PU matrices of order 2 and 4. The GBFs extracted from the generalize seed PU matrices of order 2 and 4 have been explicitly given case by case. A question we ask is whether  there is  a general form    which contains  all (or a large subset of)  binary and quaternary sequences in the CSSs of size 4 proposed in this paper.  There are more interesting questions that we could ask by exploring those new constructions presented in this paper. Here, we just list a few.

\section*{Acknowledgment}

The first author wishes to thank Dr. K.-U. Schmidt for his valuable discussions of the known constructions of complementary sets in \cite{Chen06,Chen08,Schmidt07,Stinchcombe} when the author presented  that the new CSSs were constructed in SETA 2016 \cite{SETA2016}. The authors wish to thank Dr. S.~Z.~Budi\v{s}in for his valuated suggestions on the enumeration of sequences in Corollary \ref{seq-enum} and the general constructions in Section 8.

\end{document}